\documentclass[letterpaper,11pt,cleveref, autoref, thm-restate]{article}
\usepackage[margin=1in]{geometry}
\usepackage{footmisc}
\usepackage{enumitem}
\usepackage[title]{appendix}
\usepackage{xcolor}
\usepackage{extarrows}
\usepackage{amsthm}
\usepackage{amsmath}
\usepackage{amssymb}
\usepackage{amsfonts}
\usepackage{graphicx}
\usepackage[ruled]{algorithm2e}
\usepackage{thm-restate}
\usepackage[colorlinks = true]{hyperref}
\hypersetup{
    linkcolor=black,
    citecolor=black,
    urlcolor=blue}

\usepackage{mathtools}

\allowdisplaybreaks

\newcounter{algsplit}

\newtheorem{thrm}{Theorem}[section]
\newtheorem{lem}[thrm]{Lemma}
\newtheorem{prop}[thrm]{Proposition}
\newtheorem{conj}[thrm]{Conjecture}
\newtheorem{cor}[thrm]{Corollary}
\theoremstyle{definition}
\newtheorem{defn}[thrm]{Definition}
\theoremstyle{definition}
\newtheorem{notation}[thrm]{Notation}
\theoremstyle{definition}
\newtheorem{exmpl}[thrm]{Example}
\theoremstyle{definition}
\newtheorem{rmk}[thrm]{Remark}

\newcommand{\Z}{\mathbb{Z}}
\newcommand{\N}{\mathbb{N}}
\newcommand{\Q}{\mathbb{Q}}
\newcommand{\R}{\mathbb{R}}

\newcommand{\K}{\mathbb{K}}
\newcommand{\tM}{\widehat{M}}

\newcommand{\Rp}{\mathbb{R}_{\geq 0}}

\newcommand{\Qp}{\mathbb{Q}_{\geq 0}}
\newcommand{\Zp}{\mathbb{Z}_{\geq 0}}
\newcommand{\Zpp}{\mathbb{Z}_{>0}}
\newcommand{\UT}{\mathsf{UT}}
\newcommand{\SP}{\operatorname{SP}}
\newcommand{\IP}{\operatorname{IP}}
\newcommand{\SL}{\mathsf{SL}}

\newcommand{\id}{\operatorname{id}}

\newcommand{\PI}{\operatorname{PI}}

\newcommand{\Sym}{\operatorname{S}}
\newcommand{\Lie}{\mathfrak{L}}

\newcommand{\card}{\operatorname{card}}
\newcommand{\supp}{\operatorname{supp}}

\newcommand{\set}{\operatorname{set}}
\newcommand{\mG}{\mathcal{G}}
\newcommand{\sgmG}{\langle \mathcal{G} \rangle}
\newcommand{\mR}{\mathcal{R}}

\newcommand{\mC}{\mathcal{C}}

\newcommand{\mI}{\mathcal{I}}

\newcommand{\mS}{\mathcal{S}}
\newcommand{\mH}{\mathcal{H}}

\newcommand{\mun}{\mathfrak{u}(n)}
\newcommand{\bv}{\boldsymbol{v}}
\newcommand{\bx}{\boldsymbol{x}}

\newcommand{\ba}{\boldsymbol{a}}
\newcommand{\bb}{\boldsymbol{b}}
\newcommand{\bc}{\boldsymbol{c}}

\newcommand{\bj}{\boldsymbol{j}}

\newcommand{\bl}{\boldsymbol{\ell}}
\newcommand{\bzer}{\boldsymbol{0}}

\newcounter{ProblemCounter}

\begin{document}

\title{The Identity Problem in nilpotent groups of bounded class}

\date{}
\author{Ruiwen Dong\footnote{Department of Computer Science, University of Oxford, Oxford, OX1 3QD, United Kingdom, email: \url{ruiwen.dong@kellogg.ox.ac.uk}}}
\maketitle


\begin{abstract}
Let $G$ be a unitriangular matrix group of nilpotency class at most ten.
We show that the Identity Problem (does a semigroup contain the identity matrix?) and the Group Problem (is a semigroup a group?) are decidable in polynomial time for finitely generated subsemigroups of $G$.
Our decidability results also hold when $G$ is an arbitrary finitely generated nilpotent group of class at most ten.
This extends earlier work of Babai et al.\ on commutative matrix groups (SODA'96) and work of Bell et al.\ on $\mathsf{SL}(2, \mathbb{Z})$ (SODA'17).
Furthermore, we formulate a sufficient condition for the generalization of our results to nilpotent groups of class $d > 10$.
For every such $d$, we exhibit an effective procedure that verifies this condition in case it is true.
\end{abstract}

\section{Introduction}\label{sec:intro}
\paragraph*{Algorithmic problems in matrix semigroups.}
The computational theory of groups and semigroups is one of the oldest and most well-developed parts of computational algebra. 
Algorithmic problems for matrix semigroups have been studied in computer science continuously since the work of Markov~\cite{markov1947certain} in the 1940s.
The area now plays an essential role in analysing system dynamics, and has numerous applications in automata theory, randomized algorithms, program analysis, and interactive proof systems~\cite{babai1985trading, beals1993vegas, blondel2005decidable, choffrut2005some, derksen2005quantum, hrushovski2018polynomial}.
Among the most prominent problems in this area are \emph{Semigroup Membership} and \emph{Group Membership}, proposed respectively by Markov and Mikhailova in the mid twentieth century.
For these decision problems, we work in a fixed matrix group $G$. The input is a finite set of matrices $\mG = \{A_1, \ldots, A_K\} \subseteq G$ and a matrix $A$. 
Denote by $\sgmG$ the semigroup generated by $\mG$, and by $\sgmG_{grp}$ the group generated by $\mG$.
\begin{enumerate}[noitemsep,label = (\roman*)]
    \item \textit{(Semigroup Membership)} decide whether $\sgmG$ contains $A$.
    \item \textit{(Group Membership)} decide whether $\langle\mG\rangle_{grp}$ contains $A$.
    \setcounter{ProblemCounter}{\value{enumi}}
\end{enumerate}
Both problems are undecidable in general matrix groups by the classical results of Markov and Mikhailova~\cite{markov1947certain, mikhailova1966occurrence}.
In this paper, we consider two closely related problems introduced by Choffrut and Karhum\"{a}ki~\cite{choffrut2005some} in 2005: the \emph{Identity Problem} and the \emph{Group Problem}.
These two decision problems concern the \emph{structure} of semigroups rather than their \emph{membership}.
Given as input a finite set of matrices $\mG$:
\begin{enumerate}[noitemsep,label = (\roman*)]
    \setcounter{enumi}{\value{ProblemCounter}}
    \item \textit{(Identity Problem)} decide whether $\sgmG$ contains the identity matrix $I$.
    \item \textit{(Group Problem)} decide whether $\langle\mG\rangle$ is a group, in other words, whether $\sgmG = \langle\mG\rangle_{grp}$.
    \setcounter{ProblemCounter}{\value{enumi}}
\end{enumerate}
All four algorithmic problems remain undecidable for matrices in low dimensions: 
for example, for matrices in the group $\SL(4, \Z)$ of $4 \times 4$ integer matrices of determinant one~\cite{bell2010undecidability, mikhailova1966occurrence}.
The undecidability results stem from the fact that $\SL(4,\Z)$ can embed a direct product of two non-abelian free groups.
On the other hand, for matrices in $\SL(2, \Z)$, Semigroup Membership was shown to be decidable in \textbf{EXPSPACE} by Choffrut and
Karhumaki~\cite{choffrut2005some},
Group Membership is in \textbf{PTIME} by a result of Lohrey~\cite{lohrey2021subgroup}, 
and the Identity Problem and the Group Problem are \textbf{NP}-complete by results of Bell, Hirvensalo, and Potapov \cite{bell2017identity}.

The goal of this paper is to solve the Identity Problem and the Group Problem in matrix groups with additional structures.
To this end, we will consider the more general problem of computing \emph{invertible subsets}, which subsumes the Identity Problem and the Group Problem.

\begin{defn}\label{def:invset}
Let $G$ be a matrix group. 
Given a finite set of elements $\mG = \{A_1, \ldots, A_K\} \subseteq G$, the \emph{invertible subset} of $\mG$ is the set of matrices in $\mG$ who inverse lies in $\sgmG$.
\end{defn}


\begin{restatable}{prop}{propinvtoid}\label{prop:invtoid}
Given a finite set of matrices $\mG = \{A_1, \ldots, A_K\}$ in a matrix group $G$.
Denote by $\mG_{inv}$ the invertible subset of $\mG$.
\begin{enumerate}[nosep,label = (\roman*)]
    \item The Identity Problem for $\mG$ has a positive answer if and only if $\mG_{inv}$ is non-empty.
    \item The Group Problem for $\mG$ has a positive answer if and only if $\mG_{inv} = \mG$.
\end{enumerate}
\end{restatable}

\paragraph*{Nilpotent groups, unitriangular matrices, and related work.}

Computation on matrix groups becomes easier in the presence of structural restrictions such as commutativity and nilpotence.  
In~\cite{babai1996multiplicative}, Babai et al.\ famously reduced algorithmic problems in \emph{commutative} matrix groups to computation on \emph{lattices}.
Thus, for commutative matrix groups, Group Membership reduces to linear algebra over $\Z$, and is hence decidable in \textbf{PTIME};
Semigroup Membership is equivalent to integer programming, and is hence \textbf{NP}-complete;
the Identity Problem and the Group Problem reduce to solving \emph{homogeneous} linear Diophantine equations, and are hence in \textbf{PTIME}.
The work of Babai left as an open problem how these complexity results generalize to nilpotent groups and solvable groups.
In this paper we work in the setting of nilpotent groups.

\begin{defn}
Given a group $G$ and a subgroup $H$ of $G$,
define the commutator $[G, H]$ to be the group generated by the elements in $\{ghg^{-1}h^{-1} \mid g \in G, h \in H\}$.
The \emph{lower central series} of a group $G$ is the inductively defined descending sequence of subgroups
\[
G = G_1 \geq G_2 \geq G_3 \geq \cdots,
\]
in which $G_k = [G, G_{k-1}]$.
A group $G$ is called \emph{nilpotent} if its lower central series terminates with $G_{d+1} = \{I\}$ for some $d$.
The smallest such $d$ is called the \emph{nilpotency class} of $G$.
\end{defn}

In particular, abelian groups are nilpotent of class one.
Nilpotent groups are one of the most studied classes of groups due to being the ``simplest" non-commutative groups. 
Much research has focused on algorithms for groups of relatively small nilpotency classes. For finite groups, the decades-old quest for a \textbf{PTIME} algorithm of the group isomorphism problem has focused on the very difficult case of
class two nilpotent groups~\cite{babai2011code, garzon1991isomorphism,sun2023faster}.
For infinite groups, a celebrated result of Grunewald and Segal~\cite{grunewald1980some} showed decidability of group isomorphism for all finitely generated nilpotent groups.
For membership problems, a classic result of Kopytov showed that Group Membership is decidable in nilpotent matrix groups~\cite{kopytov1968solvability}.
On the other hand, Roman'kov~\cite{roman2022undecidability} recently showed that Semigroup Membership is undecidable for a class two nilpotent matrix group.
The decidability and complexity of the Identity Problem and the Group Problem for nilpotent groups remained an intricate open problem.

The most prominent example of nilpotent groups is the group $\UT(n, \Q)$ of $n\times n$ unitriangular rational matrices.

\begin{defn}
Denote by $\UT(n, \Q)$ the group of $n \times n$ upper triangular rational matrices with ones along the diagonal:
\[
\UT(n, \Q) \coloneqq 
\left\{
\begin{pmatrix}
1 & * & \cdots & * \\
0 & 1 & \cdots & * \\
\vdots & \vdots & \ddots & \vdots \\
0 & 0 & \cdots & 1 \\
\end{pmatrix}
, \text{ where $*$ are elements of $\Q$ }. \right\}
\]
A group $G$ is called a \emph{unitriangular matrix group
over $\Q$} if it is a subgroup of $\UT(n, \Q)$ for some $n$.
\end{defn}

The group $\UT(n, \Q)$ is nilpotent of class $n-1$~\cite[Example~16.1.2]{kargapolov1979fundamentals}.
A strong motivation for studying $\UT(n, \Q)$ is the fact that every finitely generated nilpotent group is isomorphic to a subgroup of the direct product $\UT(n,\Q) \times F$ where $F$ is finite~\cite{Baumslag2007LectureNO, kargapolov1979fundamentals}.
For this reason, it suffices to focus our study on unitriangular matrix groups over $\Q$.

In \cite{ko2017identity}, Ko, Niskanen and Potapov showed the \textbf{PTIME} decidability of the Identity Problem in $\UT(3, \Q)$.
Later, utilising the special structure of the first term in the \emph{Baker-Campbell-Hausdorff (BCH) formula}, Colcombet, Ouaknine, Semukhin and Worrell proved the decidability of Semigroup Membership in $\UT(3, \Q)$ by encoding it into a Parikh automaton \cite{colcombet2019reachability}.
Recently, Dong~\cite{dong2022identity} showed the \textbf{PTIME} decidability of the Identity Problem in $\UT(4, \Z)$.
However, Dong's result relies on an \emph{ad hoc} argument from algebraic geometry, which seems unlikely to generalize to higher dimensions.
It was therefore left as an open problem whether the Identity Problem in $\UT(n, \Q)$ is decidable for $n \geq 5$.
On the undecidability side, Roman'kov~\cite{roman2022undecidability} showed that Semigroup Membership in $\UT(3, \Q)^k$ (which is of nilpotency class two) is undecidable for sufficiently large $k$.
His main technique is an embedding of the Hilbert's tenth problem.
In this paper, we generalize some of the above decidability results to unitriangular matrix groups of arbitrary \emph{dimension}, with bounded \emph{nilpotency class}.

\paragraph*{Main contribution.}
The highlight of our approach is combining convex geometry and Lie algebra to study semigroup algorithmic problems, which to the best of our knowledge is a new method in this area.
Convex geometry can be seen as the study of subsemigroups of the \emph{abelian} group $\R^n$.
Combined with Lie algebra techniques, we use it to study subsemigroups of \emph{nilpotent} groups.
The most significant contribution of our paper includes proving several intricate properties of the $k$-th term of the \emph{BCH formula}, from which our main result follows.
All but one of these properties are proven for every term of the BCH formula,
whereas the remaining one is verified term by term using assistance from computer algebra software. 
The huge computational power needed to verify this particular property is the reason why our result stops at nilpotency class ten\footnote{Nilpotent groups of high classes have intrinsically complicated structures.
Many conjectured results on nilpotent groups and $\UT(n, \K)$ are notoriously difficult to prove, but are verified for relatively small nilpotency classes.
For example, classification of nilpotent Lie algebras is done up to dimension seven~\cite{gong1998classification}, and Higman's conjecture~\cite{vera2003conjugacy} on the number of conjugacy classes in $\UT(n, \mathbb{F}_p)$ is verified up to $n \leq 13$. \label{foot:class}}.
However, we exhibit an effective procedure that verifies this property for higher classes in case it is true.



\section{Main results}\label{sec:mainres}
The main result of this paper is the following theorem.
\begin{restatable}{thrm}{thmn}\label{thm:invn}
Let $G$ be a unitriangular matrix group over $\Q$ with nilpotency class at most ten.
Given any finite set $\mG \subseteq G$,
the invertible subset of $\mG$ is computable in polynomial time.
\end{restatable}

Here, the input size is defined as the total bit length of the entries in the matrices of $\mG$.
The proof of Theorem~\ref{thm:invn} will be given in Sections~\ref{sec:alg} and \ref{sec:techthm}.
Together with Proposition~\ref{prop:invtoid}, Theorem~\ref{thm:invn} implies that the Identity Problem and the Group Problem are decidable in \textbf{PTIME} in unitriangular matrix groups over $\Q$ with nilpotency class at most ten.
For example, this result also applies to the direct product $\UT(11, \K)^m$ for any $m \in \N$ and any algebraic number field $\K$, since $\K$ can be embedded as matrices in $\Q^{k \times k}$, where $k$ is the degree of the field extension $\K / \Q$.

The following corollary extends Theorem~\ref{thm:invn} to arbitrary finitely generated nilpotent groups.
However, the complexity will depend on specific group embeddings, which we do not analyse.

\begin{restatable}{cor}{decnilp}\label{cor:decnilp}
Let $G$ be a finitely generated nilpotent group of class at most ten, given by a finite presentation~\cite[Chap.~8]{holt2005handbook}.
Then the Identity Problem and the Group Problem are decidable in $G$.
\end{restatable}


\section{Preliminaries}\label{sec:prelim}
\paragraph*{Words and linear programming.}

All omitted proofs can be found in Appendix~\ref{app:proofCor}.
Given a finite set of matrix $\mG = \{A_1, \ldots, A_K\}$, one can consider $\mG$ as a finite alphabet.
Let $\mG^*$ denote the set of words over $\mG$, and let $\mG^+$ denote the set of \emph{non-empty} words over $\mG$.
Given a word $w \in \mG^+$, by multiplying consecutively the matrices appearing in $w$, we can evaluate $w$ as a matrix, which we denote by $\pi(w)$.
Then the semigroup $\sgmG$ consists of all matrices $\pi(w)$ where $w \in \mG^+$.
We now define some concepts necessary for analysing words with linear algebra.

\begin{defn}[Parikh image]
Given a finite alphabet $\mG = \{A_1, \ldots, A_K\}$,
the \emph{Parikh image} of a word $w = B_1 \cdots B_m$ in $\mG^*$ is the vector $\bl = (\ell_1, \ldots, \ell_K) \in \Zp^K$ defined by $\ell_i = \card(\{j \mid B_j = A_i\})$ (that is, $\ell_i$ is the number of times $A_i$ appears in $w$).
The Parikh image of $w$ over the alphabet $\mG$ is denoted by $\PI^{\mG}(w)$.
\end{defn}

\begin{defn}[Cones]
Let $V$ be a $\Q$-linear space.
A subset $\mC \subseteq V$ is called a \emph{$\Qp$-cone} if $a \in \mC \Longrightarrow a\Qp \subseteq \mC$, and $a, b \in \mC \Longrightarrow a + b \in \mC$.
Given a set of vectors $\mS \subseteq V$, denote by $\langle \mS \rangle_{\Qp}$ the \emph{$\Qp$-cone generated by $\mS$},
that is the smallest $\Qp$-cone of $V$ containing $\mS$.
Similarly, denote by $\langle \mS \rangle_{\Q}$ the $\Q$-linear space generated by $\mS$.
These notations extend to $\Rp$-cones and $\R$-linear spaces.
\end{defn}

\begin{defn}[Support]
A subset $\Lambda \subseteq \Zp^K$ is called a \emph{$\Zp$-cone} if $a, b \in \Lambda \Longrightarrow a + b \in \Lambda$, and $\bzer \in \Lambda$.
The \emph{support} of a vector $\bl = (\ell_1, \ldots, \ell_K) \in \Zp^K$ is defined as the set of indices where the entry of $\bl$ is non-zero:
\[
\supp(\bl) \coloneqq \{i \in \{1, \ldots, K\} \mid \ell_i > 0\}.
\]
The \emph{support} of a $\Zp$-cone $\Lambda$ is defined as the union of supports of all vectors in $\Lambda$:
\[
\supp(\Lambda) \coloneqq \bigcup\nolimits_{\bl \in \Lambda} \supp(\bl) = \{i \mid \exists (\ell_1, \ldots, \ell_K) \in \Lambda, \ell_i > 0\}.
\]
\end{defn}

Let $V$ be a $\Q$-linear subspace of $\Q^K$, represented as the solution set of linear homogeneous equations.
Then $\Zp^K \cap V$ is a $\Zp$-cone.
In this paper, we will need to compute the support of $\Zp$-cones of the form $\Lambda = \Zp^K \cap V$ (namely, in Algorithm~\ref{alg:invn}).

\begin{restatable}{lem}{lemcompsupp}\label{lem:compsupp}
Given $V$ represented as the solution set of linear homogeneous equations, one can compute the support of $\Lambda = \Zp^K \cap V$ in polynomial time.
\end{restatable}

\paragraph*{Lie algebra.}
For a general reference on Lie algebra, see~\cite{erdmann2006introduction}.
\begin{defn}[Lie algebra $\mun$]\label{def:un}
The \emph{Lie algebra} $\mun$ is defined as the $\Q$-linear space of $n \times n$ upper triangular rational matrices with \emph{zeros} on the diagonal.
There exist maps
\[
    \log: \UT(n, \Q) \rightarrow \mun, \quad
    A \mapsto \sum_{k = 1}^n \frac{(-1)^{k-1}}{k} (A - I)^k,
\]
and
\[
    \exp: \mun \rightarrow \UT(n, \Q), \quad
    X \mapsto \sum_{k = 0}^n \frac{1}{k!} X^k,
\]
which are inverse of one another.
In particular, $\log I = 0$ and $\exp(0) = I$.
\end{defn}

The Lie algebra $\mun$ is equipped with a \emph{Lie bracket} $[\cdot, \cdot] : \mun \times \mun \rightarrow \mun$ given by $[X, Y] = XY - YX$.
The Lie bracket is bilinear, anticommutative (meaning $[X, Y] = - [Y, X]$), and it additionally satisfies the \emph{Jacobi Identity}:
\[
[X, [Y, Z]] + [Y, [Z, X]] + [Z, [X, Y]] = 0 \text{ for all } X, Y, Z \in \mun.
\]

\begin{notation}
Given a set of matrices $\mG \subseteq \UT(n, \Q)$, we denote $\log \mG \coloneqq \{\log A \mid A \in \mG\}$. It is a subset of $\mun$.
If $G$ is a subgroup of $\UT(n, \Q)$, then $\log G$ is similarly defined by considering $G$ as a set.
Given a set of elements $\mH \subseteq \mun$ and an integer $k \geq 2$, define 
\[
[\mH]_{k} \coloneqq \big\{[\ldots [[X_1, X_2], X_3], \ldots, X_k] \mid X_1, X_2, \ldots, X_k \in \mH \big\}.
\]
That is, $[\mH]_{k}$ is the set of all ``left bracketing'' of length $k$ of elements in $\mH$.
\end{notation}

It is a standard result that, using bilinearity, anticommutativity and the Jacobi identity, any $k$-iteration of Lie brackets of elements in $\mH$ can be written as a linear combination of elements in $[\mH]_{k}$.
For example, for $k = 4$, one can write
\begin{align*}
    [[X_1, X_2], [X_3, X_4]]
    = & - [[X_2, [X_3, X_4]], X_1] - [[[X_3, X_4], X_1], X_2] \quad \text{(Jacobi identity)} \\
    = & \; [[[X_3, X_4], X_2], X_1] - [[[X_3, X_4], X_1], X_2] \quad \text{(Anticommutativity)}.
\end{align*}

The following lemma 
is a corollary of the so-called \emph{Mal'cev correspondence}~\cite{mal1951some}: 

\begin{restatable}{lem}{nilclass}\label{lem:nilclass}
Let $G$ be a subgroup of $\UT(n, \Q)$. If $G$ has nilpotency class $d$, then $[\log G]_{d+1} = \{0\}$.
\end{restatable}

\paragraph*{The Baker-Campbell-Hausdorff (BCH) formula.}
\begin{thrm}[Baker-Campbell-Hausdorff (BCH) formula \cite{baker1905alternants, campbell1897law, hausdorff1906symbolische}]\label{thm:BCH}
Let $G$ be a unitriangular matrix group over $\Q$, whose nilpotency class is at most $d$. Let $B_1, \ldots, B_m$ be elements of $G$.
Then
\begin{equation}\label{eq:BCH}
\log(B_1 B_2 \cdots B_m) = \sum_{i = 1}^m \log B_i + \sum_{k = 2}^{d} H_k(\log B_1, \ldots, \log B_m),
\end{equation}
where the terms $H_k(\log B_1, \ldots, \log B_m), k = 2, 3, \ldots$, can be expressed as finite $\Q$-linear combinations of elements in $[\{\log B_1, \ldots, \log B_m\}]_k$.
\end{thrm}
In theory, one can compute the expressions $H_{k}$ effectively using recursion (see, for example \cite{casas2009efficient}). 
An explicit expression for the term $H_k$ was discovered by Dynkin~\cite{loday1992serie} (see Section~\ref{sec:techthm}).
However, as $k$ grows, these expressions quickly become very complicated.
For example, here are the explicit expressions of the first two terms.
\begin{align}\label{eq:defH23}
H_{2}(C_1, \ldots, C_m) & = \frac{1}{2} \sum_{i < j} [C_{i}, C_{j}], \nonumber \\
H_{3}(C_1, \ldots, C_m) & = \sum_{i < j < k} \left(\frac{[C_i, [C_j, C_k]]}{3} + \frac{[[C_i, C_k], C_j]}{6}\right) +  \sum_{i < j} \frac{[C_i, [C_i, C_j]] + [[C_i, C_j], C_j]}{12}.
\end{align}

\section{Polynomial time algorithm for Theorem~\ref{thm:invn}}\label{sec:alg}
In this section, we exhibit the algorithm that proves the main result of this paper (Theorem~\ref{thm:invn}).
In order to describe our algorithm, we need to introduce the following notation.
Let $\mH$ be a finite set of elements in the Lie algebra $\mun$. For any $k \geq 1$, denote
\[
\Lie_{\geq k}(\mH) \coloneqq \left\langle \bigcup\nolimits_{i \geq k} [\mH]_{i} \right\rangle_{\Q}.
\]
That is, $\Lie_{\geq k}(\mH)$ is the linear space spanned by the set of all ``left bracketing'' of length \emph{at least} $k$ of elements in $\mH$.
By Lemma~\ref{lem:nilclass}, if a unitriangular matrix group $G$ has nilpotency class $d$, then for any $\mH \subseteq \log G$, we have $\Lie_{\geq d+1}(\mH) = \{0\}$, and $\Lie_{\geq k}(\mH) = \langle[\mH]_{k}\rangle_{\Q} + \langle[\mH]_{k+1}\rangle_{\Q} + \cdots + \langle[\mH]_{d}\rangle_{\Q}$.
We have thus an ascending series of linear spaces $\{0\} = \Lie_{\geq d+1}(\mH) \subseteq \Lie_{\geq d}(\mH) \subseteq \cdots \subseteq \Lie_{\geq 1}(\mH) \subseteq \mun$, such that $[\Lie_{\geq i}(\mH), \Lie_{\geq j}(\mH)] \subseteq \Lie_{\geq i+j}(\mH)$.

\begin{exmpl}[label=exa:cont]\label{example:u4}
    We give a concrete example to show how the subspaces $\Lie_{\geq k}(\mH), k = d, \ldots, 2, 1,$ may look like.
    Let $G = \UT(4, \Q)$, so it has nilpotency class $d = 3$.
    Consider the Lie algebra 
    \[
    \mathfrak{u}(4) = 
    \left\{
    \begin{pmatrix}
        0 & * & * & *\\
        0 & 0 & * & *\\
        0 & 0 & 0 & * \\
        0 & 0 & 0 & 0 \\        
    \end{pmatrix},
    \text{ where $*$ are entries in $\Q$}
    \right\}.
    \]
    It is a $\Q$-linear space of dimension six.
    Let $\mG = \{A_1, A_2, A_3\}$, where
    \begin{equation*}
        A_1 = 
    \begin{pmatrix}
        1 & 2 & -1 & 1\\
        0 & 1 & 2 & 1\\
        0 & 0 & 1 & 2 \\
        0 & 0 & 0 & 1 \\
    \end{pmatrix},
    A_2 = 
    \begin{pmatrix}
        1 & -1 & -1 & 2\\
        0 & 1 & -1 & -1\\
        0 & 0 & 1 & 0 \\
        0 & 0 & 0 & 1 \\
    \end{pmatrix},
    A_3 = 
    \begin{pmatrix}
        1 & 0 & 3 & -1\\
        0 & 1 & 0 & 1\\
        0 & 0 & 1 & -1 \\
        0 & 0 & 0 & 1 \\
    \end{pmatrix}.
    \end{equation*}
    Let $\mH = \{\log A_1, \log A_2, \log A_3\}$. 
    In particular,
    \begin{equation}\label{eq:logexample}
    \log A_1 = 
    \begin{pmatrix}
        0 & 2 & -3 & \frac{11}{3}\\
        0 & 0 & 2 & -1\\
        0 & 0 & 0 & 2 \\
        0 & 0 & 0 & 0 \\
    \end{pmatrix},
    \log A_2 = 
    \begin{pmatrix}
        0 & -1 & -\frac{3}{2} & \frac{3}{2}\\
        0 & 0 & -1 & -1\\
        0 & 0 & 0 & 0 \\
        0 & 0 & 0 & 0 \\
    \end{pmatrix},
    \log A_3 = 
    \begin{pmatrix}
        0 & 0 & 3 & \frac{1}{2}\\
        0 & 0 & 0 & 1\\
        0 & 0 & 0 & -1 \\
        0 & 0 & 0 & 0 \\
    \end{pmatrix}.
    \end{equation}
    We have $[\mH]_4 = \{0\}$ by applying Lemma~\ref{lem:nilclass} with $d = 3$.
    Moreover,
    \begin{align*}
        [\mH]_3 & = \big\{[[\log A_1, \log A_2], \log A_1], [[\log A_1, \log A_2], \log A_2], \ldots, [[\log A_3, \log A_2], \log A_2]\big\}, \\
        [\mH]_2 & = \big\{[\log A_1, \log A_2], [\log A_1, \log A_3], [\log A_2, \log A_3], [\log A_2, \log A_1] = -[\log A_1, \log A_2], \ldots \big\}, \\
        [\mH]_1 & = \big\{\log A_1, \log A_2, \log A_3\big\}.
    \end{align*}
    
    Then, we have $\Lie_{\geq 4}(\mH) = \{0\}$, $\Lie_{\geq 3}(\mH) = \langle[\mH]_{3}\rangle_{\Q}$, $\Lie_{\geq 2}(\mH) = \langle[\mH]_{2}\rangle_{\Q} + \langle[\mH]_{3}\rangle_{\Q}$ and $\Lie_{\geq 1}(\mH) = \langle[\mH]_{1}\rangle_{\Q} + \langle[\mH]_{2}\rangle_{\Q} + \langle[\mH]_{3}\rangle_{\Q}$.
    By direct computation, this yields
    \begin{multline}\label{eq:L321}
        \Lie_{\geq 3}(\mH) = 
    \left\{
    \begin{pmatrix}
        0 & 0 & 0 & a\\
        0 & 0 & 0 & 0\\
        0 & 0 & 0 & 0 \\
        0 & 0 & 0 & 0 \\        
    \end{pmatrix}
    \;\middle|\;
    a \in \Q
    \right\}, 
    \Lie_{\geq 2}(\mH) = 
    \left\{
    \begin{pmatrix}
        0 & 0 & 0 & a\\
        0 & 0 & 0 & b\\
        0 & 0 & 0 & 0 \\
        0 & 0 & 0 & 0 \\        
    \end{pmatrix}
    \;\middle|\;
    a, b \in \Q
    \right\}, \\
    \Lie_{\geq 1}(\mH) = 
    \left\{
    \begin{pmatrix}
        0 & c_1 & c_3 & a\\
        0 & 0 & c_1 & b\\
        0 & 0 & 0 & c_2 \\
        0 & 0 & 0 & 0 \\        
    \end{pmatrix}
    \;\middle|\;
    a, b, c_1, c_2, c_3 \in \Q
    \right\}.
    \end{multline}
    Hence in this example, $\Lie_{\geq 3}(\mH), \Lie_{\geq 2}(\mH), \Lie_{\geq 1}(\mH)$ are respectively subspaces of $\mathfrak{u}(4)$ of dimension one, two and five.
\end{exmpl}

Let $G$ be a subgroup of $\UT(n, \Q)$ of nilpotency class at most ten, and fix $\mG = \{A_1, \ldots, A_K\}$ to be a finite alphabet of elements in $G$.
For any vector $\bl = (\ell_1, \ldots, \ell_K) \in \Zp^K$, define
\[
\mG_{\supp(\bl)} \coloneqq \{A_i \mid A_i \in \mG, i \in \supp(\bl)\}
\]
as the set of matrices in $\mG$ whose index appears in the support set $\supp(\bl)$.

Recall that for a word $w \in \mG^+$, the matrix $\pi(w)$ is obtained by multiplying consecutively the matrices appearing in $w$.
The key ingredient of our algorithm is the following Theorem~\ref{thm:equivn}, which provides a criterion for the existence of a non-empty word $w \in \mG^+$ satisfying $\log \pi(w) = 0$ (equivalently, $\pi(w) = I$).
In particular, this provides a criterion for whether $I \in \sgmG$ (the Identity Problem), and can be extended to the computation of invertible subsets.

\begin{restatable}{thrm}{equivn}\label{thm:equivn}
Let $\mG = \{A_1, \ldots, A_K\}$ be a finite set of matrices in $\UT(n, \Q)$ that satisfies $[\log \mG]_{11} = \{0\}$.
Given a non-zero vector $\bl = (\ell_1, \ldots, \ell_K) \in \Zp^K$:
\begin{enumerate}[nosep,label = (\roman*)]
    \item If there exists a word $w \in \mG^{+}$ with $\PI^{\mG}(w) = \bl$ and $\log \pi(w) = 0$, then
    \begin{equation}\label{eq:condn}
        \sum_{i = 1}^K \ell_i \log A_i \in \Lie_{\geq 2}(\log \mG_{\supp(\bl)}).
    \end{equation}
    \item If $\bl$ satisfies \eqref{eq:condn}, then there exists a word $w \in \mG^{+}$ with $\PI^{\mG}(w) \in \Zpp \cdot \bl$, such that $\log \pi(w) = 0$.
\end{enumerate}
\end{restatable}

Part (i) of Theorem~\ref{thm:equivn} is relatively easy to prove:
\begin{proof}[Proof of part (i) of Theorem~\ref{thm:equivn}]
Let $w$ be a word with $\PI^{\mG}(w) = \bl$. 
Write $w = B_1 B_2 \cdots B_m$ where $B_i \in \mG, i = 1, \ldots, m$.
Regrouping by letters, we have $\sum_{i = 1}^K \ell_i \log A_i = \sum_{i = 1}^m \log B_i$.

If $\log \pi(w) = 0$, then by the BCH formula (Theorem~\ref{thm:BCH}), we have
\[
\sum_{i = 1}^m \log B_i + \sum_{k=2}^{n-1} H_k(\log B_1, \ldots, \log B_m) = \log(B_1 B_2 \cdots B_m) = 0.
\]
The higher order terms $H_k, k \geq n$ vanish because $[\log \mG]_{n} = \{0\}$ (a consequence of $\mG \subseteq \UT(n, \Q)$).
Therefore, $\sum_{i = 1}^K \ell_i \log A_i = \sum_{i = 1}^m \log B_i = - \sum_{k=2}^{n-1} H_k(\log B_1, \ldots, \log B_m)$.

Since the Parikh image of the word $B_1 B_2 \cdots B_m$ is $\bl$, the matrices $B_i$ all lie in the subset $\{A_i \mid i \in \supp(\bl)\}$ of $\mG$.
Therefore, $\log B_i \in \log \mG_{\supp(\bl)}$ for all $i$.
By Theorem~\ref{thm:BCH}, for all $k \geq 2$ we have $- H_k(\log B_1, \ldots, \log B_m) \in \big\langle[\{\log B_i \mid i = 1, \ldots, m\}]_k \big\rangle_{\Q} \subseteq \Lie_{\geq k}(\log \mG_{\supp(\bl)}) \subseteq \Lie_{\geq 2}(\log \mG_{\supp(\bl)})$.
Therefore, we have
$
\sum_{i = 1}^K \ell_i \log A_i = - \sum_{k=2}^{n-1} H_k(\log B_1, \ldots, \log B_m) \in \Lie_{\geq 2}(\log \mG_{\supp(\bl)}).
$
\end{proof}

Proving part~(ii) of Theorem~\ref{thm:equivn} is highly non-trivial and will be the main focus of Section~\ref{sec:techthm}.
We continue Example~\ref{example:u4} to give an intuition of the Condition~\eqref{eq:condn} in Theorem~\ref{thm:equivn}.

\renewcommand\thmcontinues[1]{continued}
\begin{exmpl}[continues=exa:cont]
Let $\mG$ be as in Example~\ref{example:u4}.
As an example, we show that $\bl = (1, 2, 2)$ satisfies Equation~\eqref{eq:condn}.
When $\bl = (1, 2, 2)$, we have $\supp(\bl) = \{1, 2, 3\}$, so $\Lie_{\geq 2}(\log \mG_{\supp(\bl)}) = \Lie_{\geq 2}(\log \mG)$ as defined in Equation~\eqref{eq:L321}.
Therefore,
    \[
    \sum_{i = 1}^3 \ell_i \log A_i
    =
    \log A_1 + 2 \log A_2 + 2 \log A_3
    =
    \begin{pmatrix}
            0 & 0 & 0 & \frac{23}{3}\\
            0 & 0 & 0 & -1\\
            0 & 0 & 0 & 0 \\
            0 & 0 & 0 & 0 \\
        \end{pmatrix}
        \in 
    \Lie_{\geq 2}(\log \mG_{\supp(\bl)}),
    \]
where $\log A_i, i = 1, 2, 3,$ are given in Equation~\eqref{eq:logexample}.
Hence in this example, $\bl$ satisfies Equation~\eqref{eq:condn}. 

\end{exmpl}

Note that finding solutions of Equation~\eqref{eq:condn} relies only on linear algebra.
Assuming Theorem~\ref{thm:equivn}, we can devise the following Algorithm~\ref{alg:invn} that computes the invertible subset of any finite set $\mG \subseteq G$.

\begin{algorithm}[ht]
\caption{Computing the invertible subset of $\mG$}
\label{alg:invn}
\begin{description}[nosep]
\item[Input:] 
A finite set of elements $\mG = \{A_1, \ldots, A_K\}$ in $G$.
\item[Output:] The invertible subset $\mG_{inv}$ of $\mG$.
\end{description}
\begin{enumerate}[nosep,label = Step~\arabic*]
    \item \textbf{Initialization.}
    Set $S \coloneqq \{1, \ldots, K\}$.
    \item\label{step:alg1step2} \textbf{Main loop.} Repeat the following
    \begin{enumerate}
        \item Represent the $\Q$-linear subspace of $\Q^K$: 
        \[
        V \coloneqq \left\{(\ell_1, \ldots, \ell_K) \in \Q^K \middle| \sum\nolimits_{i=1}^K \ell_i \log A_i \in \Lie_{\geq 2}(\{\log A_i \mid i \in S\}) \right\}
        \]
        as the solution set of homogeneous linear equations.
        \item Define $\Lambda \coloneqq \Zp^K \cap V$ and compute $\supp(\Lambda)$ using Lemma~\ref{lem:compsupp}.
        \item If $\supp(\Lambda) = S$, terminate the algorithm and return $\mG_{inv} = \{A_i \mid i \in S\}$. \\
        Otherwise let $S \coloneqq \supp(\Lambda)$ and continue.
    \end{enumerate}
\end{enumerate}
\end{algorithm}

\begin{proof}[Proof of Theorem~\ref{thm:invn} and proof of correctness of Algorithm~\ref{alg:invn} (assuming Theorem~\ref{thm:equivn})]
After each iteration of \ref{step:alg1step2}, the cardinality of $\supp(\Lambda)$ strictly decreases. 
Therefore, the algorithm terminates after at most $K$ iterations of \ref{step:alg1step2}.

Since $G$ has nilpotency class at most ten, by Lemma~\ref{lem:nilclass}, its subset $\mG$ satisfies $[\log \mG]_{11} = \{0\}$.
We start by showing that, when the algorithm terminates, every element of $\{A_i \mid i \in S\}$ has an inverse in the semigroup $\sgmG$.
When the algorithm terminates at \ref{step:alg1step2}(c), we have $\supp(\Lambda) = S$.
By the additivity of $\Lambda$ (that is, $\ba, \bb \in \Lambda \implies \ba+\bb \in \Lambda$), there exists a vector $\bl = (\ell_1, \ldots, \ell_K) \in \Lambda$ such that $\supp(\bl) = \supp(\Lambda) = S$.
This vector then satisfies
$
\sum_{i=1}^K \ell_i \log A_i \in \Lie_{\geq 2}(\{\log A_i \mid i \in \supp(\bl)\})
$
by the definition of $V$.
By Theorem~\ref{thm:equivn}(ii), this shows that there exists a non-empty word $w$, with $\PI^{\mG}(w) \in \Zpp \cdot \bl$ such that $\log \pi(w) = 0$ (that is, $\pi(w) = I$).
For any $i \in S$, since $\supp(\bl) = S$, the letter $A_i$ appears in the word $w$.
Write $w = w_1 A_i w_2$; then since $\pi(w_1 A_i w_2) = I$, we have $\pi(w_1) A_i\pi(w_2) = I$.
Hence, $A_i^{-1} = \pi(w_2) \pi(w_1) \in \sgmG \cup \{I\}$, so $A_i^{-1} \in \sgmG$.

We then show that for every matrix $A_i$ invertible in $\sgmG$, the index $i$ is in the set $S$ at the termination of the algorithm.
Suppose $A_i^{-1}$ is equal to $\pi(w)$, where $w$ is a non-empty word. Then the product of the word $w' = w A_i$ is equal to the identity, that is, $\log \pi(w') = 0$.
By Theorem~\ref{thm:equivn}(i), the Parikh image $\bl = \PI^{\mG}(w')$ satisfies
    $
        \sum_{i = 1}^K \ell_i \log A_i \in \Lie_{\geq 2}(\{\log A_i \mid i \in \supp(\bl)\}).
    $
    
We show that $\supp(\bl) \subseteq S$ is an invariant of the algorithm.
At initialization, we obviously have $\supp(\bl) \subseteq S$.
Before each iteration of \ref{step:alg1step2}(b), suppose we have $\supp(\bl) \subseteq S$, then
\[
    \sum_{i = 1}^K \ell_i \log A_i \in \Lie_{\geq 2}(\{\log A_i \mid i \in \supp(\bl)\}) \subseteq \Lie_{\geq 2}(\{\log A_i \mid i \in S\}).
\]
Hence $\bl \in \Lambda = \Zp^K \cap V$.
Consequently, $\supp(\bl) \subseteq \supp(\Lambda)$ at the beginning of \ref{step:alg1step2}(c), which shows that $\supp(\bl) \subseteq S$ still holds after the iteration of \ref{step:alg1step2}.
This invariant shows that $i \in \supp(\bl) \subseteq S$ by the end of the algorithm.
Combining with the previous implication, we conclude that by the end of the algorithm, $S$ is exactly the set of elements in $\mG$ with inverse in $\sgmG$.

For the complexity analysis, recall that the algorithm terminates after at most $K$ iterations of \ref{step:alg1step2}.
At each iteration of \ref{step:alg1step2}(b), the support $\supp(\Lambda)$ can be computed in polynomial time by Lemma~\ref{lem:compsupp}. 
The total input size of these linear programming instances is polynomial with respect to the total bit length of the matrix entries in $\mG$.
Indeed, a $\Q$-basis of $\Lie_{\geq 2}(\{\log A_i \mid i \in S\})$ is simply the set $\bigcup_{10 \geq k \geq 2} [\{\log A_i \mid i \in S\}]_{k}$, whose total bit length is of polynomial size in $\mG$.
From this, one can express $V$ as the solution set of a system of homogeneous linear equations whose total bit length is polynomial in $\mG$ (note that the total bit length of $\log A_i$ is also polynomial in $\mG$).
Therefore, the overall complexity of Algorithm~\ref{alg:invn} is polynomial with respect to the input $\mG$.
\end{proof}

\section{Proof of Theorem~\ref{thm:equivn}(ii)}\label{sec:techthm}
In this section we give the proof of Theorem~\ref{thm:equivn}(ii).
We first give an intuition of the proof by continuing Example~\ref{example:u4}.
This will illustrate some of the ideas necessary to prove the general case.

\renewcommand\thmcontinues[1]{final part}
\begin{exmpl}[continues=exa:cont]
Let $\mG$ be as in Example~\ref{example:u4}.
Let $\bl = (1, 2, 2)$.
We have already shown that $\bl$ satisfies Equation~\eqref{eq:condn},
so Theorem~\ref{thm:equivn}(ii) claims that there exists a word $w \in \mG^+$ with $\PI^{\mG}(w) \in \Zpp \cdot (1, 2, 2)$, such that $\log \pi(w) = 0$.
We illustrate here how to construct this word $w$ in two steps.
By slight abuse of notation we now write $\log A$ instead of $\log \pi(A)$ for any word $A \in \mG^+$.

\textbf{Step 1.} We find elements $A'_1, A'_2, A'_3$ in $\mG^+$, such that $\log A'_1, \log A'_2, \log A'_3$ generate the subspace $\Lie_{\geq 2}(\log \mG_{\supp(\bl)})$ as a \emph{$\Qp$-cone}.
The idea is to take
\begin{equation}\label{eq:choiceAprime}
A'_1 \coloneqq A_1^t A_2^{2t} A_3^{2t}, \quad A'_2 \coloneqq A_2^{2t} A_3^{2t} A_1^t, \quad A'_3 \coloneqq A_2^{2t} A_1^t A_3^{2t},
\end{equation}
for a suitable $t \in \N$.
Apply the BCH formula~\eqref{eq:BCH} with $B_1 \coloneqq A_1^t, B_2 \coloneqq A_2^{2t}, B_3 \coloneqq A_3^{2t}$, we obtain
\begin{multline}\label{eq:applyBCH}
    \log A'_1 = \log (A_1^t A_2^{2t} A_3^{2t}) = \log A_1^t + \log A_2^{2t} + \log A_3^{2t} + \sum\nolimits_{k = 2}^{3} H_k(\log A_1^t, \log A_2^{2t}, \log A_3^{2t}) \\
    = t \cdot (\log A_1 + 2 \log A_2 + 2 \log A_3) + \sum\nolimits_{k = 2}^{3} t^k \cdot H_k(\log A_1, 2\log A_2, 2 \log A_3).
\end{multline}
The last equality is due to $\log A^t = t \log A$ and because the term $H_k$ is a linear combination of $k$-iterations of Lie brackets.

The linear term $t \cdot (\log A_1 + 2 \log A_2 + 2 \log A_3)$ in \eqref{eq:applyBCH} falls in the subspace $\Lie_{\geq 2}(\log \mG_{\supp(\bl)})$ by Condition~\eqref{eq:condn}.
The non-linear terms $t^k \cdot H_k(\log A_1, 2\log A_2, 2 \log A_3), k = 2, 3,$ also fall in the subspace $\Lie_{\geq 2}(\log \mG_{\supp(\bl)})$ by Theorem~\ref{thm:BCH}.
Therefore, we have $\log A'_1 \in \Lie_{\geq 2}(\log \mG_{\supp(\bl)})$.
Similarly, $\log A'_2$ and $\log A'_3$ are also in $\Lie_{\geq 2}(\log \mG_{\supp(\bl)})$.

Using the exact expression~\eqref{eq:defH23} for the terms $H_2$ and $H_3$, we obtain that the expressions for $\log A'_1, \log A'_2$ and $\log A'_3$ are respectively
\begin{equation*}
\begin{pmatrix}
    0 & 0 & 0 & \frac{4}{3} t^3 + \frac{23}{3} t\\
    0 & 0 & 0 & 2 t^2 - t\\
    0 & 0 & 0 & 0 \\
    0 & 0 & 0 & 0 \\
\end{pmatrix},
\begin{pmatrix}
    0 & 0 & 0 & -\frac{8}{3} t^3 + 2 t^2 + \frac{23}{3} t\\
    0 & 0 & 0 & 2 t^2 - t\\
    0 & 0 & 0 & 0 \\
    0 & 0 & 0 & 0 \\
\end{pmatrix},
\text{ and }
\begin{pmatrix}
    0 & 0 & 0 & \frac{4}{3} t^3 + \frac{23}{3} t\\
    0 & 0 & 0 & - 2 t^2 - t\\
    0 & 0 & 0 & 0 \\
    0 & 0 & 0 & 0 \\
\end{pmatrix}.
\end{equation*}
We then choose $t = 10$. This choice is made so that $t$ is large enough for $\log A'_1, \log A'_2, \log A'_3$ to exhibit their ``asymptotic'' behaviour.
When $t = 10$, we have
\begin{equation}\label{eq:logAprime}
\log A'_1 = 
\begin{pmatrix}
    0 & 0 & 0 & 1410\\
    0 & 0 & 0 & 190\\
    0 & 0 & 0 & 0 \\
    0 & 0 & 0 & 0 \\
\end{pmatrix},
\log A'_2 = 
\begin{pmatrix}
    0 & 0 & 0 & -2390\\
    0 & 0 & 0 & 190\\
    0 & 0 & 0 & 0 \\
    0 & 0 & 0 & 0 \\
\end{pmatrix},
\log A'_3 = 
\begin{pmatrix}
    0 & 0 & 0 & 1410\\
    0 & 0 & 0 & -210\\
    0 & 0 & 0 & 0 \\
    0 & 0 & 0 & 0 \\
\end{pmatrix}.
\end{equation}
Then indeed we have
$
\langle \log A'_1, \log A'_2, \log A'_3 \rangle_{\Qp} = 
\Lie_{\geq 2}(\log \mG_{\supp(\bl)})
$,
which is proved by linear programming.
Furthermore, the Parikh images are
$
\PI^{\mG}(A'_1) = \PI^{\mG}(A'_2) = \PI^{\mG}(A'_3) = (10, 20, 20)
$.

\textbf{Step 2.} Consider the new alphabet $\mG' \coloneqq \{A'_1, A'_2, A'_3\}$. 
We now find a non-empty word $A'' \in \left(\mG'\right)^+$, such that $\log A'' \in \Lie_{\geq 2}(\Lie_{\geq 2}(\log \mG_{\supp(\bl)})) = \{0\}$.
Directly computing from~\eqref{eq:logAprime} yields
\begin{equation}\label{eq:poscomb}
    117 \cdot \log A'_1 + 282 \cdot \log A'_2 + 361 \cdot \log A'_3 = 0.
\end{equation}
Let 
$
A'' \coloneqq \left(A'_1\right)^{117} \cdot \left(A'_2\right)^{282} \cdot \left(A'_3\right)^{361}.
$
By the BCH formula~\eqref{eq:BCH}, we have
$
\log A'' = 117 \cdot \log A'_1 + 282 \cdot \log A'_2 + 361 \cdot \log A'_3 = 0.
$
This is because all the terms $H_k, k \geq 2$ in the BCH formula are in
\begin{equation}\label{eq:L4van}
\Lie_{\geq 2}(\Lie_{\geq 2}(\log \mG_{\supp(\bl)})) \subseteq \Lie_{\geq 4}(\log \mG_{\supp(\bl)}) = \{0\}.
\end{equation}
Furthermore, the Parikh image of $A''$ is
$
\PI^{\mG}(A'') = 117 \cdot \PI^{\mG}(A'_1) + 282 \cdot \PI^{\mG}(A'_2) + 361 \cdot \PI^{\mG}(A'_3) = 7600 \cdot (1, 2, 2)
$.
We have thus found the word $w = A''$ satisfying $\log \pi(w) = 0$, with Parikh image $7600 \cdot (1, 2, 2)$.
This concludes our example.
\end{exmpl}

The following subsections aim to formalize the idea exhibited in this example and provide a rigorous proof of Theorem~\ref{thm:equivn}(ii).
Here is an overview of the main difficulties in formalizing a proof.
\begin{enumerate}[nosep, wide, label = (\roman*)]
    \item\label{dif:1} In Equation~\eqref{eq:choiceAprime} we took a specific choice of $A'_1, A'_2, A'_3$.
    In the general case, we will use a similar idea of taking $A'_i$ to be words of the form $A_{i_1}^t A_{i_2}^t \cdots A_{i_m}^t$. However, the permutations $(i_1, i_2, \ldots, i_m)$ need to be chosen carefully.
    We need to show that there exist enough permutations so that the constructed elements $\log A'_1, \log A'_2, \ldots,$ generate the linear space $\Lie_{\geq 2}(\log \mG_{\supp(\bl)})$. We achieve this by proving a deep combinatorial property of the terms $H_k$ (Proposition~\ref{prop:lieperm}).
    \item\label{dif:2}
    Furthermore, $\log A'_1, \log A'_2, \ldots,$ need to generate $\Lie_{\geq 2}(\log \mG_{\supp(\bl)})$ as a \emph{cone}.
    The coefficients $(117, 282, 361)$ obtained in Equation~\eqref{eq:poscomb} happen to be all positive, but this is \emph{a priori} not always the case.
    We need to show that $0$ can always be written as a \emph{positive} combination of $\log A'_i$.
    This is proved by finding identities over the terms $H_k$ using computer assistance (Proposition~\ref{prop:k}).
    \item\label{dif:3} The exponent $t$ in Equation~\eqref{eq:choiceAprime} needs to be chosen carefully. In fact, we may even need to take several different $t$. Such $t$ are chosen using techniques from convex geometry (Proposition~\ref{prop:cone}).
    \item\label{dif:4} In the above example the nilpotency class of $G$ is three. This is the reason why in Step 2, Equation~\eqref{eq:L4van} holds, and the matrices $A'_1, A'_2, A'_3$ commute with each other.
    In the general case, we deal with groups of nilpotency class up to ten. Then, Equation~\eqref{eq:L4van} no longer holds.
    Hence, we need to repeat the above process for more steps. 
    In general, when $G$ has nilpotency class up to $2^d-1$, we need to repeat the process for $d$ steps (Subsection~\ref{subsec:fullproof}). 
\end{enumerate}

Thus, our formal proof of Theorem~\ref{thm:equivn}(ii) relies on the three following technical propositions. 
For $k \in \Zpp$, denote by $\Sym_k$ the permutation group of the set $\{1, \ldots, k\}$.

\begin{restatable}{prop}{lieperm}\label{prop:lieperm}
For every $k \geq 2$, there exists a function $\mu_k \colon \Sym_k \rightarrow \Z$, such that for any sequence of elements $C_1, \ldots, C_m, m \geq k,$ in the Lie algebra $\mun$ we have
\begin{equation}\label{eq:lieperm}
    [\ldots[[C_1, C_2], C_3], \ldots, C_k] = \sum\nolimits_{\sigma \in \Sym_k} \mu_k(\sigma) H_k(C_{\sigma(1)}, \ldots, C_{\sigma(k)}, C_{k+1}, \ldots, C_m).
\end{equation}
\end{restatable}


\begin{restatable}{prop}{propk}\label{prop:k}
Let $k \leq 10$ and let $\mH \subset \UT(n, \Q)$ be a finite set of matrices for some $n \geq 2$.
Then there exist a non-negative integer $r$, positive rational numbers $\alpha_1, \ldots, \alpha_r$, as well as, for $s = 1, \ldots, r$, words $\bj_s = j_{s,1} j_{s,2} \cdots j_{s,m_s}$ in the alphabet $\mI = \{1, 2, \ldots, k+1\}$, such that $\PI^{\mI}(\bj_s) \in \{(1, \ldots, 1), (2, \ldots, 2)\}$ and
\begin{multline}\label{eq:propk}
    \sum_{\sigma \in \Sym_{k+1}} H_{k}(\log B_{\sigma(1)}, \ldots, \log B_{\sigma(k+1)})
    + \sum_{s = 1}^r \alpha_{s} \sum_{\sigma \in \Sym_{k+1}} H_k(\log B_{\sigma(j_{s,1})}, \ldots, \log B_{\sigma(j_{s,m_s})}) \\
    \in \Lie_{\geq k+1}(\log \mH) + \Lie_{\geq 2}(\Lie_{\geq 2}(\log \mH))
\end{multline}
for all $B_1, \ldots, B_{k+1} \in \UT(n, \Q)$ satisfying $\log B_i \in \Lie_{\geq 1}(\log \mH)$
and $\sum_{i = 1}^{k+1} \log B_i \in \Lie_{\geq 2}(\log \mH)$.
\end{restatable}

\begin{restatable}{prop}{cone}\label{prop:cone}
Let $V$ be a finite dimensional $\Q$-linear space.
Let $d$ be a positive integer, $\mI$ be a finite index set, and $\ba_{1i}, \ldots, \ba_{di}, i \in \mI$ be vectors in $V$.
For any $t \in \Zpp$ and $i \in \mI$, define
\[
P_i(t) \coloneqq t \cdot \ba_{1i} + t^2 \cdot \ba_{2i} + \cdots + t^d \cdot \ba_{di}.
\]
Suppose the following two conditions hold:
\begin{enumerate}[nolistsep, label=(\roman*)]
    \item The $\Qp$-cone $\mC_{d} \coloneqq \langle \ba_{di} \mid i \in \mI \rangle_{\Qp}$ is a linear space.
    \item For $k = d-1, d-2, \ldots, 1$, the inductively defined $\Qp$-cones $\mC_{k} \coloneqq \langle \ba_{ki} \mid i \in \mI \rangle_{\Qp} + \mC_{k+1}$ are linear spaces.
\end{enumerate}
Then the $\Qp$-cone $\langle P_i(t) \mid i \in \mI, t \in \Zpp \rangle_{\Qp}$ is equal to $\mC_1$.
\end{restatable}

This concludes the overview. 
In the following subsections we will gradually prove these technical propositions.
The intuition of Proposition~\ref{prop:lieperm} is as follows.
Theorem~\ref{thm:BCH} showed that in the BCH formula, the terms $H_k(\log B_1, \ldots, \log B_m)$ can be written as a linear combination of $k$-iterated Lie brackets $[\ldots[[\log B_{i_1}, \log B_{i_2}], \log B_{i_3}], \ldots, \log B_{i_k}]$.
Here, Proposition~\ref{prop:lieperm} shows that a converse of it is true:
for any $k \geq 2$, the $k$-iterated Lie bracket $[\ldots[[\log B_1, \log B_2], \log B_3], \ldots, \log B_k]$ can be written as a linear combination of expressions in $H_k$.

Proposition~\ref{prop:k} shows that for $k \leq 10$, one can find a linear combination with \emph{positive} coefficients of the terms $H_k$ that lies in $\Lie_{\geq k+1} + \Lie_{\geq 2}(\Lie_{\geq 2}(\cdot))$. (Note that \emph{a priori} $H_k$ lies in $\Lie_{\geq k}(\cdot)$.)
Proposition~\ref{prop:k} is the only one among the three above propositions that is limited by the nilpotency class.
This constitutes the main obstacle to generalizing Theorem~\ref{thm:equivn} to higher nilpotency classes.

Finally, Proposition~\ref{prop:cone} concerns only convex geometry and is responsible for finding a suitable $t$ from difficulty~\ref{dif:3}.

\subsection{Proof of Proposition~\ref{prop:lieperm}}
For a permutation $\sigma \in \Sym_k$, define $d(\sigma)$ to be the number of \emph{descents} in $\sigma$, that is, the number of $i \in  \{1, \ldots, k-1\}$ such that $\sigma(i) > \sigma(i+1)$.
In order to prove Proposition~\ref{prop:lieperm}, we need an explicit expression for the terms $H_k$.
This expression is provided by Dynkin\footnote{Dynkin originally only proved the bivariate case of Lemma~\ref{lem:dynkin}. It was later been generalized to the multivariate case without much difficulty.}:
\begin{lem}[Dynkin formula \cite{dynkin2000calculation}, {\cite[Proposition~3.4 and Proposition~4.2]{loday1992serie}}]\label{lem:dynkin}
We have
\begin{equation}\label{eq:Dynkin1}
H_k(C_{1}, \ldots, C_m) = \sum_{i_1 + \cdots + i_m = k} \frac{1}{i_1 ! \ldots i_m !} \varphi_k(\underbrace{C_1, \ldots, C_1}_{i_1}, \underbrace{C_2, \ldots, C_2}_{i_2}, \ldots, \underbrace{C_m, \ldots, C_m}_{i_m}),
\end{equation}
where the indices $i_1, \ldots, i_m$ are non-negative integers, and
\begin{equation}\label{eq:Dynkin2}
    \varphi_k(X_1, \ldots, X_k) = \sum_{\sigma \in \Sym_k} \frac{(-1)^{d(\sigma)}}{k^2 \binom{k-1}{d(\sigma)}} [\ldots[[X_{\sigma(1)}, X_{\sigma(2)}], X_{\sigma(3)}], \ldots, X_{\sigma(k)}].
\end{equation}
\end{lem}

Define recursively the following maps $\mu_k: \Sym_k \rightarrow \Z, k = 2, 3, \ldots$.
For $k = 2$, let $\mu_2(\id) = 1, \mu_2((12)) = -1$, where $\id$ is the constant permutation and $(12)$ is the permutation that swaps $1$ and $2$.
For $k \geq 3$, denote by $(j_1 j_2 \cdots j_m)$ the cyclic permutation that sends $j_i$ to $j_{i+1}$, $i = 1, \ldots, m-1$, and sends $j_m$ to $j_1$.
Suppose $\mu_{k-1}$ already defined, we then define
\begin{align}\label{eq:defmu}
\mu_{k}(\sigma) \coloneqq 
    \begin{cases}
    \mu_{k-1}(\sigma) \quad & k = \sigma(k) \\
    - \mu_{k-1}(\sigma \circ (12 \cdots k)) \quad & k = \sigma(1) \\
    0 \quad & k = \sigma(i), i = 2, \ldots, k-1.
    \end{cases}
\end{align}
In the first two cases, the permutation $\sigma$ or $\sigma \circ (12 \cdots k)$ fixes $k$, so they can be considered as elements in $\Sym_{k-1}$, hence $\mu_{k-1}(\sigma)$ is well defined.
For example, $\mu_3(\sigma) = 1$ when $\sigma = \id$ or $(13)$; $\mu_3(\sigma) = -1$ when $\sigma = (12)$ or $(132)$; and $\mu_3(\sigma) = 0$ otherwise.
We will show that, for this $\mu_k$, the Equation~\eqref{eq:lieperm} in Proposition~\ref{prop:lieperm} is satisfied:

\lieperm*
\begin{proof}
Take $\mu_k$ to be the function defined recursively in \eqref{eq:defmu}.
For every $j \geq 2$, there is a natural embedding $f_j: \Sym_j \hookrightarrow \Sym_{j+1}$, defined by $f_j(\sigma)(i) = \sigma(i), i = 1, \ldots, j$, $f_j(\sigma)(j+1) = j+1$.
It is easy to verify that under this natural embedding, $\mu_j$ and $\mu_{j+1}$ are identified, that is, $\mu_j = \mu_{j+1} \circ f_j$.
Therefore, we can denote by $\mu$ the map $\cup_{k \geq 2} \Sym_k \rightarrow \Z$ as $\mu(\sigma) = \mu_k(\sigma)$, where $\sigma \in \Sym_k$.
We prove Equation~\eqref{eq:lieperm} in three steps.

\begin{enumerate}[nolistsep, label = (\arabic*)]
    \item 
First, we simplify the right hand side of Equation~\eqref{eq:lieperm} by showing
\begin{equation}\label{eq:simplify}
\sum_{\sigma \in \Sym_k} \mu(\sigma) H_k(C_{\sigma(1)}, \ldots, C_{\sigma(k)}, C_{k+1}, \ldots, C_m) = \sum_{\sigma \in \Sym_k} \mu(\sigma) \varphi_k(C_{\sigma(1)}, \ldots, C_{\sigma(k)}),
\end{equation}
where $\varphi_k$ is defined in Lemma~\ref{lem:dynkin}.

Thanks to Lemma~\ref{lem:dynkin}, $H_k(C_1, \ldots, C_m)$ can be written as
\begin{multline}\label{eq:sumrewrite}
H_k(C_1, \ldots, C_m) = \\
\sum_{1\leq j_1 < j_2 < \cdots < j_k \leq m}\varphi_k(C_{j_1}, \ldots, C_{j_k})
+ \sum_{l = 2}^{k-1}\sum_{1\leq j_1 < j_2 < \cdots < j_l \leq m} H_{kl}(C_{j_1}, \ldots, C_{j_l}),
\end{multline}
where $H_{kl}(C_{j_1}, \ldots, C_{j_l})$ is some linear combination of elements in $[\{C_{j_1}, \ldots, C_{j_l}\}]_k$.
By abuse of notation, for $\sigma \in \Sym_k$ and $x > k$, we define $\sigma(x) = \sigma^{-1}(x) = x$.
For any $l = 2, \ldots, k-1$, we have
\begin{multline}\label{eq:sumordchange}
 \sum_{\sigma \in \Sym_k} \sum_{1\leq j_1 < j_2 < \cdots < j_l \leq m} \mu(\sigma) H_{kl}(C_{\sigma(j_1)}, \ldots, C_{\sigma(j_l)}) \\
=  \sum_{\overset{t_1, t_2, \ldots, t_l \in \{1, \ldots, m\}}{\text{pairwise distinct}}} H_{kl}(C_{t_1}, \ldots, C_{t_l}) \sum_{\overset{\sigma \in \Sym_k}{\sigma^{-1}(t_1) < \cdots < \sigma^{-1}(t_l)}} \mu(\sigma).
\end{multline}
We claim that, for any pairwise distinct $t_1, t_2, \ldots, t_l \in \{1, \ldots, m\}$, $l < k$, we have
\begin{equation}\label{eq:sumsigma}
    \sum_{\overset{\sigma \in \Sym_k}{\sigma^{-1}(t_1) < \cdots < \sigma^{-1}(t_l)}} \mu(\sigma) = 0.
\end{equation}
We show \eqref{eq:sumsigma} by induction on $k$.
When $k = 2$, by the definition of $\mu$, \eqref{eq:sumsigma} holds.
Suppose \eqref{eq:sumsigma} holds for $k - 1$.
Denote by $c$ the cyclic permutation $(12 \cdots k)$, then by the recursive definition of $\mu$,
\begin{equation}\label{eq:sumkminusone}
\sum_{\overset{\sigma \in \Sym_k}{\sigma^{-1}(t_1) < \cdots < \sigma^{-1}(t_l)}} \mu(\sigma) = \sum_{\overset{\sigma \in \Sym_{k-1}}{\sigma^{-1}(t_1) < \cdots < \sigma^{-1}(t_l)}} \mu(\sigma) - \sum_{\overset{\sigma \in \Sym_{k-1}}{c \circ \sigma^{-1}(t_1) < \cdots < c \circ \sigma^{-1}(t_l)}} \mu(\sigma).
\end{equation}
Without loss of generality, suppose the sum on the left hand side is not empty.
That is, there exists at least one permutation $\sigma \in \Sym_k$ such that $\sigma^{-1}(t_1) < \cdots < \sigma^{-1}(t_l)$.
Since $\sigma^{-1} \in \Sym_k$ does not permute any $t_j$ with $t_j > k$, the elements of $\{t_1, \ldots, t_l\}$ which are larger than $k$ must appear after the elements which are smaller or equal to $k$, and must appear in increasing order.
In other words, there exists some $s \geq 1$, such that $t_i \leq k$ for all $i < s$, and $k < t_{s} < \cdots < t_{l}$. ($s$ could be $l + 1$, in which case $t_i \leq k$ for all $i = 1, \ldots, l$.)
Since $\sigma \in \Sym_k$ does not change the value of $t_{s}, \cdots, t_l$, one can discard them without changing the sum.
Hence, we suppose without loss of generality $t_1, \ldots, t_l \in \{1, \ldots, k\}$.
\begin{enumerate}[nolistsep]
    \item \textbf{If $t_i = k$ for some $i = 2, \ldots, l-1$.}
    Then no permutation $\sigma \in \Sym_{k-1}$ can satisfy $\sigma^{-1}(t_1) < \sigma^{-1}(t_i) = k < \sigma^{-1}(t_l)$ or $c \circ \sigma^{-1}(t_1) < c \circ \sigma^{-1}(t_i) = 1 < c \circ \sigma^{-1}(t_l)$.
    Hence, both sums on the right hand side of Equation~\eqref{eq:sumkminusone} are empty. The claim \eqref{eq:sumsigma} follows.
    \item \textbf{If $t_1 = k$.} Then no permutation $\sigma \in \Sym_{k-1}$ can satisfy $\sigma^{-1}(t_1) < \sigma^{-1}(t_i) = k < \sigma^{-1}(t_l)$, so the first sum on the right hand side of Equation~\eqref{eq:sumkminusone} is empty. As for the second sum, because $c \circ \sigma^{-1}(t_1) = c(k) = 1$, we have $c \circ \sigma^{-1}(t_1) < \cdots < c \circ \sigma^{-1}(t_l)$ if and only if $\sigma^{-1}(t_2) < \cdots < \sigma^{-1}(t_l)$. Hence, using the induction hypothesis on $t_2, \ldots, t_l \in \{1, \ldots, k\}$ yields
    \[
    \sum_{\overset{\sigma \in \Sym_{k-1}}{c \circ \sigma^{-1}(t_1) < \cdots < c \circ \sigma^{-1}(t_l)}} \mu(\sigma) = \sum_{\overset{\sigma \in \Sym_{k-1}}{\sigma^{-1}(t_2) < \cdots < \sigma^{-1}(t_l)}} \mu(\sigma) = 0.
    \]
    Therefore both sums on the right hand side of Equation~\eqref{eq:sumkminusone} equal zero. The claim \eqref{eq:sumsigma} follows.
    \item \textbf{If $t_l = k$.} Similar to the previous case, the second sum on the right hand side of Equation~\eqref{eq:sumkminusone} is empty. As for the first sum, because $\sigma^{-1}(t_l) = k$, we have $\sigma^{-1}(t_1) < \cdots < \sigma^{-1}(t_l)$ if and only if $\sigma^{-1}(t_1) < \cdots < \sigma^{-1}(t_{l-1})$. Hence, using the induction hypothesis on $t_1, \ldots, t_{l-1} \in \{1, \ldots, k\}$ shows the sum is zero. The claim \eqref{eq:sumsigma} follows.
    \item \textbf{If $t_i \neq k$ for all $i = 1, \ldots, l$.} Then $\sigma^{-1}(t_1) < \cdots < \sigma^{-1}(t_l)$ if and only if $c \circ \sigma^{-1}(t_1) < \cdots < c \circ \sigma^{-1}(t_l)$. Hence, the two sums on the right hand side of Equation~\eqref{eq:sumkminusone} are the same. The claim \eqref{eq:sumsigma} follows.
\end{enumerate}
Using the claim \eqref{eq:sumsigma} on Equation~\eqref{eq:sumordchange} yields
\begin{equation}\label{eq:sumHkl}
\sum_{\sigma \in \Sym_k} \sum_{1\leq j_1 < j_2 < \cdots < j_l \leq m} \mu(\sigma) H_{kl}(C_{\sigma(j_1)}, \ldots, C_{\sigma(j_l)}) = 0,
\end{equation}
and this combined with Equation~\eqref{eq:sumrewrite} yields
\begin{align}\label{eq:simp1}
& \sum_{\sigma \in \Sym_k} \mu(\sigma) H_k(C_{\sigma(1)}, \ldots, C_{\sigma(k)}, C_{k+1}, \ldots, C_m) \nonumber \\
= & \sum_{\sigma \in \Sym_k} \mu(\sigma) H_k(C_{\sigma(1)}, \ldots, C_{\sigma(m)}) \quad \text{(define $\sigma(s) = s$ for $\sigma \in \Sym_k$ and $s > k$)} \nonumber \\
= & \sum_{\sigma \in \Sym_k} \mu(\sigma) \sum_{1\leq j_1 < j_2 < \cdots < j_k \leq m}\varphi_k(C_{\sigma(j_1)}, \ldots, C_{\sigma(j_k)}) \quad \text{(by \eqref{eq:sumrewrite} and \eqref{eq:sumHkl})} \nonumber \\
= & \sum_{\overset{t_1, t_2, \ldots, t_k \in \{1, \ldots, m\}}{\text{pairwise distinct}}} \varphi_k(C_{t_1}, \ldots, C_{t_k}) \sum_{\overset{\sigma \in \Sym_k}{\sigma^{-1}(t_1) < \cdots < \sigma^{-1}(t_k)}} \mu(\sigma) \nonumber \\
= & \sum_{l = 0}^{k}\sum_{\overset{\overset{t_1, t_2, \ldots, t_l \in \{1, \ldots, k\}}{\text{pairwise distinct,}}}{k < t_{l+1} < \cdots < t_k \leq m}} \varphi_k(C_{t_1}, \ldots, C_{t_k}) \sum_{\overset{\sigma \in \Sym_k}{\sigma^{-1}(t_1) < \cdots < \sigma^{-1}(t_l)}} \mu(\sigma) 
\end{align}
Because Equation~\eqref{eq:sumsigma} holds for $l < k$, that is, the sum
$
\sum_{\overset{\sigma \in \Sym_k}{\sigma^{-1}(t_1) < \cdots < \sigma^{-1}(t_l)}} \mu(\sigma)
$
vanishes whenever $l < k$,
the above expression~\eqref{eq:simp1} is equal to
\[
\sum_{\overset{t_1, t_2, \ldots, t_k \in \{1, \ldots, k\}}{\text{pairwise distinct}}} \varphi_k(C_{t_1}, \ldots, C_{t_k}) \sum_{\overset{\sigma \in \Sym_k}{\sigma^{-1}(t_1) < \cdots < \sigma^{-1}(t_k)}} \mu(\sigma) = \sum_{\sigma \in \Sym_k} \mu(\sigma) \varphi_k(C_{\sigma(1)}, \ldots, C_{\sigma(k)}).
\]
We have hence shown Equation~\eqref{eq:simplify}:
\begin{equation*}
\sum_{\sigma \in \Sym_k} \mu(\sigma) H_k(C_{\sigma(1)}, \ldots, C_{\sigma(k)}, C_{k+1}, \ldots, C_m) = \sum_{\sigma \in \Sym_k} \mu(\sigma) \varphi_k(C_{\sigma(1)}, \ldots, C_{\sigma(k)}),
\end{equation*}

\item
The second step is to show
\begin{equation}\label{eq:penult}
\sum_{\sigma \in \Sym_k} \mu(\sigma) \varphi_k(C_{\sigma(1)}, \ldots, C_{\sigma(k)}) = \sum_{T \in \Sym_k} \frac{\mu(T)}{k} [\ldots[[C_{T(1)}, C_{T(2)}], C_{T(3)}], \ldots, C_{T(k)}].
\end{equation}
Using the exact expression for $\varphi_k$ in Lemma~\ref{lem:dynkin}, we have
\begin{align}\label{eq:inter}
    & \sum_{\sigma \in \Sym_k} \mu(\sigma) \varphi_k(C_{\sigma(1)}, \ldots, C_{\sigma(k)}) \nonumber \\*
    = & \sum_{\sigma \in \Sym_k} \sum_{\tau \in \Sym_k} \frac{(-1)^{d(\tau)} \mu(\sigma)}{k^2 \binom{k-1}{d(\tau)}} [\ldots[[C_{\sigma \circ \tau(1)}, C_{\sigma \circ \tau(2)}], C_{\sigma \circ \tau(3)}], \ldots, C_{\sigma \circ \tau(k)}] \nonumber \\*
    = & \sum_{T \in \Sym_k} [\ldots[[C_{T(1)}, C_{T(2)}], C_{T(3)}], \ldots, C_{T(k)}] \sum_{\sigma \in \Sym_k} \frac{(-1)^{d(\sigma^{-1} \circ T)} \mu(\sigma)}{k^2 \binom{k-1}{d(\sigma^{-1} \circ T)}}
\end{align}
We will compute the value of $\sum_{\sigma \in \Sym_k} \frac{(-1)^{d(\sigma^{-1} \circ T)} \mu(\sigma)}{k^2 \binom{k-1}{d(\sigma^{-1} \circ T)}}$ depending on the permutation $T$.
We show by induction on $k$ that
\begin{equation}\label{eq:indgoal}
\sum_{\sigma \in \Sym_k} \frac{(-1)^{d(\sigma^{-1} \circ T)} \mu(\sigma)}{k^2 \binom{k-1}{d(\sigma^{-1} \circ T)}} = \frac{\mu(T)}{k}.
\end{equation}
When $k = 2$, by direct computation, $\sum_{\sigma \in \Sym_k} \frac{(-1)^{d(\sigma^{-1} \circ T)} \mu(\sigma)}{k^2 \binom{k-1}{d(\sigma^{-1} \circ T)}}$ is equal to $\frac{1}{2}$ if $T = \id$ and to $- \frac{1}{2}$ if $T = (12)$.
This matches the values of $\frac{\mu(T)}{k}$.
If $k \geq 3$, suppose \eqref{eq:indgoal} proven for $k-1$.
Again denote by $c$ the cyclic permutation $(12 \cdots k)$, by the recursive definition of $\mu$ we have
    \begin{equation}\label{eq:twosums}
        \sum_{\sigma \in \Sym_k} \frac{(-1)^{d(\sigma^{-1} \circ T)} \mu(\sigma)}{k^2 \binom{k-1}{d(\sigma^{-1} \circ T)}} = \sum_{\sigma \in \Sym_{k-1}} \frac{(-1)^{d(\sigma^{-1} \circ T)} \mu(\sigma)}{k^2 \binom{k-1}{d(\sigma^{-1} \circ T)}} - \sum_{\sigma \in \Sym_{k-1}} \frac{(-1)^{d(c \circ \sigma^{-1} \circ T)} \mu(\sigma)}{k^2 \binom{k-1}{d(c \circ \sigma^{-1} \circ T)}}.
    \end{equation}
\begin{enumerate}[nolistsep]
    \item \textbf{If $T(i) = k$ for some $i = 2, \ldots, k-1$.}
    We claim that $d(\sigma^{-1} \circ T) = d(c \circ \sigma^{-1} \circ T)$ for all $\sigma \in \Sym_{k-1}$.
    In fact, for $\sigma \in \Sym_{k-1}$, we have $\sigma^{-1} \circ T(i) = k$ and $c \circ \sigma^{-1} \circ T(i) = 1$.
    Therefore $\sigma^{-1} \circ T(i) > \sigma^{-1} \circ T(i+1)$, $\sigma^{-1} \circ T(i) > \sigma^{-1} \circ T(i-1)$, whereas $c \circ \sigma^{-1} \circ T(i) < c \circ \sigma^{-1} \circ T(i+1)$, $c \circ \sigma^{-1} \circ T(i) < c \circ \sigma^{-1} \circ T(i-1)$.
    And for $j \neq i-1, i$, we have $\sigma^{-1} \circ T(j) > \sigma^{-1} \circ T(j+1)$ if and only if $c \circ \sigma^{-1} \circ T(j) > c \circ \sigma^{-1} \circ T(j+1)$.
    This shows $d(\sigma^{-1} \circ T) = d(c \circ \sigma^{-1} \circ T)$.
    Hence, the two sums on the right hand side of \eqref{eq:twosums} are equal, and $\sum_{\sigma \in \Sym_k} \frac{(-1)^{d(\sigma^{-1} \circ T)} \mu(\sigma)}{k^2 \binom{k-1}{d(\sigma^{-1} \circ T)}} = 0 = \frac{\mu(T)}{k}$. %
    
    \item \textbf{If $T(1) = k$.} Similar to the above discussion, we can show that $d(\sigma^{-1} \circ T) = d(c \circ \sigma^{-1} \circ T) + 1$. 
    Hence the right hand side of \eqref{eq:twosums} is equal to
    \begin{align*}
        & - \sum_{\sigma \in \Sym_{k-1}} \left(\frac{(-1)^{d(c \circ \sigma^{-1} \circ T)} \mu(\sigma)}{k^2 \binom{k-1}{d(c \circ \sigma^{-1} \circ T) + 1}} + \frac{(-1)^{d(c \circ \sigma^{-1} \circ T)} \mu(\sigma)}{k^2 \binom{k-1}{d(c \circ \sigma^{-1} \circ T)}} \right) \\
        = & \;- \sum_{\sigma \in \Sym_{k-1}} \frac{(-1)^{d(c \circ \sigma^{-1} \circ T)} \mu(\sigma)}{k(k-1) \binom{k-2}{d(c \circ \sigma^{-1} \circ T)}} \\
        = & \; \frac{-(k-1)}{k} \sum_{\sigma \in \Sym_{k-1}} \frac{(-1)^{d(c \circ \sigma^{-1} \circ T)} \mu(\sigma)}{(k-1)^2 \binom{k-2}{d(c \circ \sigma^{-1} \circ T)}}
    \end{align*}
    We claim that $d(c \circ \sigma^{-1} \circ T) = d(\sigma^{-1} \circ T \circ c)$.
    This is because $\sigma^{-1} \circ T \circ c (k-1) < \sigma^{-1} \circ T \circ c (k) = k$, $1 = c \circ \sigma^{-1} \circ T(1) < c \circ \sigma^{-1} \circ T(2)$, and $c \circ \sigma^{-1} \circ T(i+1) > c \circ \sigma^{-1} \circ T(i)$ if and only if $\sigma^{-1} \circ T \circ c (i) > \sigma^{-1} \circ T \circ c (i-1)$, for $i = 2, 3, \ldots, k-1$.
    Hence,
    \begin{alignat*}{2}
        & \frac{-(k-1)}{k} \sum_{\sigma \in \Sym_{k-1}} \frac{(-1)^{d(c \circ \sigma^{-1} \circ T)} \mu(\sigma)}{(k-1)^2 \binom{k-2}{d(c \circ \sigma^{-1} \circ T)}} \\
        = & \; \frac{- (k-1)}{k} \sum_{\sigma' \in \Sym_{k-1}^c} \frac{(-1)^{d(\sigma^{-1} \circ T \circ c)} \mu(\sigma)}{(k-1)^2 \binom{k-2}{d(\sigma^{-1} \circ T \circ c)}} \\
        = & \; \frac{-(k-1)}{k} \frac{\mu(T \circ c)}{k-1} && \quad \text{(by induction hypothesis)}\\
        = & \; \frac{(k-1)}{k} \frac{\mu(T)}{k-1} && \quad \text{(by definition of $\mu$)}\\
        = & \; \frac{\mu(T)}{k}.
    \end{alignat*}

    \item \textbf{If $T(k) = k$.} Similar to the above discussion, we can show that $d(c \circ \sigma^{-1} \circ T) = d(\sigma^{-1} \circ T) + 1$. And hence the right hand side of \eqref{eq:twosums} is equal to 
    \[
    \frac{(k-1)}{k} \sum_{\sigma \in \Sym_{k-1}} \frac{(-1)^{d(\sigma^{-1} \circ T)} \mu(\sigma)}{(k-1)^2 \binom{k-2}{d(\sigma^{-1} \circ T)}} = \frac{(k-1)}{k} \frac{\mu(T)}{k-1} = \frac{\mu(T)}{k}
    \]
    by the induction hypothesis, where $T$ can be considered as an element in $\Sym_{k-1}$ since it stabilizes $k$.
\end{enumerate}
We have thus shown the claim \eqref{eq:indgoal}.
Putting this into Equation~\eqref{eq:inter} shows Equation~\eqref{eq:penult}:
\begin{equation*}
\sum_{\sigma \in \Sym_k} \mu(\sigma) \varphi_k(C_{\sigma(1)}, \ldots, C_{\sigma(k)}) = \sum_{T \in \Sym_k} \frac{\mu(T)}{k} [\ldots[[C_{T(1)}, C_{T(2)}], C_{T(3)}], \ldots, C_{T(k)}].
\end{equation*}

\item
The third and last step is to show\footnote{A direct way of proving Equation~\eqref{eq:final} is to use the Dynkin-Specht-Wever theorem \cite{dynkin2000calculation}, which states that if a non-commutative polynomial $f \in \Q\langle C_1, \ldots, C_k \rangle$ is \emph{Lie}, then one can replace all monomials $C_{i_1} C_{i_2} \cdots C_{i_k}$ by $[\ldots[C_{i_1}, C_{i_2}], \ldots, C_{i_k}]/k$ without changing its value.
Writing the right hand side of \eqref{eq:final} as an element in $\Q\langle C_1, \ldots, C_k \rangle$ gives $k \sum_{\sigma \in \Sym_k} \mu(\sigma) C_{\sigma(1)} C_{\sigma(2)} \cdots C_{\sigma(k)}$ (we can check this using the definition of $\mu$), which is equal to the left hand side by replacing the monomials $C_{\sigma(1)} C_{\sigma(2)} \cdots C_{\sigma(k)}$ by the Lie brackets $[\ldots[[C_{\sigma(1)}, C_{\sigma(2)}], C_{\sigma(3)}], \ldots, C_{\sigma(k)}]/k$.
Nevertheless, here we will give a self-contained proof without using the Dynkin-Specht-Wever theorem.}
\begin{equation}\label{eq:final}
    \sum_{T \in \Sym_k} \mu(T) [\ldots[[C_{T(1)}, C_{T(2)}], C_{T(3)}], \ldots, C_{T(k)}] = k [\ldots[[C_1, C_{2}], C_{3}], \ldots, C_{k}].
\end{equation}
First, using induction on $k$, we will show that
\begin{equation}\label{eq:indgoal2}
    \sum_{T \in \Sym_{k}} \mu(T) [\ldots[[C_{k+1}, C_{T(1)}], C_{T(2)}], \ldots, C_{T(k)}] = - [\ldots[[C_{1}, C_{2}], C_{3}], \ldots, C_{k+1}].
\end{equation}
The case where $k = 2$ is immediate.
Suppose Equation~\eqref{eq:indgoal2} hold for $k-1$, then
\begin{multline}\label{eq:twosums2}
    \sum_{T \in \Sym_{k}} \mu(T) [\ldots[[C_{k+1}, C_{T(1)}], C_{T(2)}], \ldots, C_{T(k)}] \\*
    = \sum_{T \in \Sym_{k-1}} \mu(T) [[\ldots[[C_{k+1}, C_{T(1)}], C_{T(2)}], \ldots, C_{T(k-1)}], C_{k}] \\*
    - \sum_{T \in \Sym_{k-1}} \mu(T) [\ldots[[C_{k+1}, C_{k}], C_{T(1)}], \ldots, C_{T(k-1)}].
\end{multline}
By the induction hypothesis, the first sum on the right hand side is equal to
\[
- [[[\ldots[[C_1, C_{2}], C_{3}], \ldots, C_{k-1}], C_{k+1}], C_k],
\]
and the second sum on the right hand side is equal to
\[
- [[\ldots[[C_1, C_{2}], C_{3}], \ldots, C_{k-1}], [C_{k+1}, C_k]].
\]
Using the Jacobi identity and the anticommutativity of Lie brackets, we have
\begin{multline*}
- [[[\ldots[[C_1, C_{2}], C_{3}], \ldots, C_{k-1}], C_{k+1}], C_k] + [[\ldots[[C_1, C_{2}], C_{3}], \ldots, C_{k-1}], [C_{k+1}, C_k]] \\
= - [[[\ldots[[C_1, C_{2}], C_{3}], \ldots, C_{k-1}], C_{k}], C_{k+1}].
\end{multline*}
Hence, Equation~\eqref{eq:twosums2} yields
\[
\sum_{T \in \Sym_{k}} \mu(T) [\ldots[[C_{k+1}, C_{T(1)}], C_{T(2)}], \ldots, C_{T(k)}] = - [[\ldots[[C_1, C_{2}], C_{3}], \ldots, C_{k}], C_{k+1}],
\]
concluding the proof by induction for Equation~\eqref{eq:indgoal2}.

Next, we will again use induction on $k$ to prove Equation~\eqref{eq:final}:
\begin{equation*}
    \sum_{T \in \Sym_k} \mu(T) [\ldots[[C_{T(1)}, C_{T(2)}], C_{T(3)}], \ldots, C_{T(k)}] = k [\ldots[[C_1, C_{2}], C_{3}], \ldots, C_{k}].
\end{equation*}
The case of $k = 2$ results from direct computation.
Suppose \eqref{eq:final} hold for $k-1$, then
\begin{alignat*}{2}
    & \sum_{T \in \Sym_k} \mu(T) [\ldots[[C_{T(1)}, C_{T(2)}], C_{T(3)}], \ldots, C_{T(k)}] \\
    = & \sum_{T \in \Sym_{k-1}} \mu(T) [[\ldots[[C_{T(1)}, C_{T(2)}], C_{T(3)}], \ldots, C_{T(k-1)}], C_k] \\
    & \quad - \sum_{T \in \Sym_{k-1}} \mu(T) [\ldots[[C_{k}, C_{T(1)}], C_{T(2)}], \ldots, C_{T(k-1)}] \\
    = & \; (k-1) [\ldots[[C_1, C_{2}], C_{3}], \ldots, C_{k}] \\
    & \quad - \sum_{T \in \Sym_{k-1}} \mu(T) [\ldots[[C_{k}, C_{T(1)}], C_{T(2)}], \ldots, C_{T(k-1)}] && \quad (\text{by induction hypothesis}) \\
    = & \; k [\ldots[[C_1, C_{2}], C_{3}], \ldots, C_{k}] && \quad (\text{by Equation~\eqref{eq:indgoal2} for } k - 1).
\end{alignat*}
We have thus shown Equation~\eqref{eq:final}.
\end{enumerate}
Combining the Equations~\eqref{eq:simplify}, \eqref{eq:penult} and \eqref{eq:final} obtained in the three steps gives us
\[
\sum_{\sigma \in \Sym_k} \mu(\sigma) H_k(C_{\sigma(1)}, \ldots, C_{\sigma(k)}, C_{k+1}, \ldots, C_m) = [\ldots[[C_1, C_2], C_3], \ldots, C_k].
\]
\end{proof}

\subsection{Proof of Proposition~\ref{prop:k}}
In this subsection we prove Proposition~\ref{prop:k}.
Again, the key is understanding the structure of the expressions for $H_k$.
For even $k$, the following lemma shows that the expression $H_k(C_1, \ldots, C_m)$ is ``antisymmetric'', and immediately yields Proposition~\ref{prop:k}.
\begin{lem}\label{lem:antisym}
When $k$ is even, we have
\[
H_k(C_1, \ldots, C_m) = - H_k(C_m, \ldots, C_1).
\]
\end{lem}
\begin{proof}
Define a new variable $t$.
Replacing $B_i$ by $\exp({t C_i})$ in the BCH formula~\eqref{eq:BCH}, we have
\begin{equation}\label{eq:BCHpos}
    \log(\exp({t C_1}) \cdots \exp({t C_m})) = t \sum_{i = 1}^m C_i + t^k \sum_{k=2}^{d-1} H_k(C_1, \ldots, C_m).
\end{equation}
Now, replace $B_i$ by $\exp({- t C_{m + 1 - i}})$, $i = 1, \ldots, m$, in the BCH formula~\eqref{eq:BCH}, we obtain
\begin{equation}\label{eq:BCHneg}
    \log(\exp({- t C_m}) \cdots \exp({- t C_1})) = - t \sum_{i = 1}^m C_i + (-t)^k \sum_{k=2}^{d-1} H_k(C_m, \ldots, C_1).
\end{equation}
Since $\log(\exp({t C_1}) \cdots \exp({t C_m})) = - \log(\exp({- t C_m}) \cdots \exp({- t C_1}))$, comparing the coefficients of $t^k$ in \eqref{eq:BCHpos} and \eqref{eq:BCHneg} yields
\[
H_k(C_1, \ldots, C_m) = - H_k(C_m, \ldots, C_1)
\]
for even $k$.
\end{proof}

Next, we need the following lemmas regarding the odd terms $H_3$, $H_5$, $H_7$ and $H_9$.
These correspond to Proposition~\ref{prop:k} for $k = 3, 5, 7, 9$.
\begin{lem}\label{lem:H3}
Let $\mH \subset \UT(n, \Q)$ be a finite set of matrices.
Given matrices $B_1, \ldots, B_m$ in $\UT(n, \Q)$ such that $\log B_i \in \Lie_{\geq 1}(\log \mH)$, $i = 1, \ldots, m$,
and $\sum_{i = 1}^m \log B_i \in \Lie_{\geq 2}(\log \mH)$,
then
\[
\sum_{\sigma \in \Sym_m} H_3(\log B_{\sigma(1)}, \ldots, \log B_{\sigma(m)}) \in \Lie_{\geq 4}(\log \mH).
\]
\end{lem}
\begin{proof}
Denote $C_i \coloneqq \log B_i, i = 1, \ldots, m$, we will show the following identity
\begin{equation}\label{eq:H3}
\sum_{\sigma \in \Sym_m} H_3(C_{\sigma(1)}, \ldots, C_{\sigma(m)}) = \frac{m!}{12} \sum_{i=1}^m \left[C_i, \left[C_i, \sum_{j=1}^m C_j\right]\right].
\end{equation}
Write
\[
H_3(C_{\sigma(1)}, \ldots, C_{\sigma(m)}) = \sum_{i < j < k} H_{33}(C_{\sigma(i)}, C_{\sigma(j)}, C_{\sigma(k)}) + \sum_{i < j} H_{32}(C_{\sigma(i)}, C_{\sigma(j)}),
\]
where
\[
H_{33}(X, Y, Z) = \frac{1}{3}[X, [Y, Z]] + \frac{1}{6}[[X, Z], Y],
\]
\[
H_{32}(X, Y) = \frac{1}{12} ([X, [X, Y]] + [[X, Y], Y]).
\]
Using the Jacobi identity, we have
\begin{align*}
& H_{33}(C_{i}, C_{j}, C_{k}) + H_{33}(C_{j}, C_{k}, C_{i}) + H_{33}(C_{k}, C_{i}, C_{j}) \\
= & \; \frac{1}{3}\left([C_i, [C_j, C_k]] + [C_j, [C_k, C_i] + [C_k, [C_i, C_j]]]\right) \\
& + \frac{1}{6}\left([[C_i, C_j], C_k] + [[C_j, C_k], C_i] + [[C_k, C_i], C_j]\right) \\
= & \; 0
\end{align*}
for any $i,j,k$.
Similarly,
\[
H_{33}(C_{k}, C_{j}, C_{i}) + H_{33}(C_{j}, C_{i}, C_{k}) + H_{33}(C_{i}, C_{k}, C_{j}) = 0.
\]
Hence,
\begin{align*}
& \sum_{\sigma \in \Sym_m} \sum_{i < j < k} H_{33}(C_{\sigma(i)}, C_{\sigma(j)}, C_{\sigma(k)}) \\
= & \; \frac{m!}{6} \sum_{i < j < k} (H_{33}(C_{i}, C_{j}, C_{k}) + H_{33}(C_{j}, C_{k}, C_{i}) + H_{33}(C_{k}, C_{i}, C_{j})) \\
& + \frac{m!}{6} \sum_{i < j < k} (H_{33}(C_{k}, C_{j}, C_{i}) + H_{33}(C_{j}, C_{i}, C_{k}) + H_{33}(C_{i}, C_{k}, C_{j})) \\
= & \; 0.
\end{align*}
Whereas 
\begin{align*}
& \sum_{\sigma \in \Sym_m} \sum_{i < j} H_{3,2}(C_{\sigma(i)}, C_{\sigma(j)}) \\
= & \; \frac{m!}{2} \sum_{i \neq j} H_{3,2}(C_{i}, C_{j}) \\
= & \; \frac{m!}{2} \sum_{i \neq j} \left(\frac{1}{12}[C_i, [C_i, C_j]] + \frac{1}{12}[[C_i, C_j], C_j]\right) \\
= & \; \frac{m!}{2} \sum_{i=1}^m\sum_{j=1}^m \left(\frac{1}{12}[C_i, [C_i, C_j]] + \frac{1}{12}[[C_i, C_j], C_j]\right) \\
= & \; \frac{m!}{2} \sum_{i=1}^m \frac{1}{12}\left[C_i, \left[C_i, \sum_{j=1}^m C_j\right]\right] + \frac{m!}{2} \sum_{j=1}^m \frac{1}{12}\left[\left[\sum_{i=1}^m C_i, C_j\right], C_j \right] \\
= & \; \frac{m!}{12} \sum_{i=1}^m \left[C_i, \left[C_i, \sum_{j=1}^m C_j\right]\right].
\end{align*}
Adding up the two above expressions yields Equation~\eqref{eq:H3}.
Since $\log B_i \in \Lie_{\geq 1}(\log \mH)$ for all $i$ and $\sum_{i = 1}^m \log B_i \in \Lie_{\geq 2}(\log \mH)$, Equation~\eqref{eq:H3} yields
\begin{align*}
    & \sum_{\sigma \in \Sym_m} H_3(\log B_{\sigma(1)}, \ldots, \log B_{\sigma(m)}) \\
    = & \; \frac{m!}{12} \sum_{i=1}^m \left[\log B_i, \left[\log B_i, \sum_{j=1}^m \log B_j\right]\right] \\
    \in & \; \frac{m!}{12} \sum_{i=1}^m \left[\log B_i, \left[\log B_i, \Lie_{\geq 2}(\log \mH)\right]\right] \\
    \in & \; \Lie_{\geq 4}(\log \mH)
\end{align*}
\end{proof}

The following Lemmas~\ref{lem:H5}, \ref{lem:H7} and \ref{lem:H9} regarding $H_5, H_7, H_9$ are proven using computer assistance from the software SageMath \cite{sagemath}.
In what follows, we give a sketch of their proof.
Details of the full proof along with the algorithm used for computer assistance are given in Section~\ref{app:proofH}.
Links to the code can be found in the respective proofs.
\begin{restatable}{lem}{Hfive}\label{lem:H5}
Let $\mH \subset \UT(n, \Q)$ be a finite set of matrices.
There exists a permutation $(j_1, j_2, \ldots, j_{12})$ of the tuple $(1, 1, 2, 2, \ldots, 6, 6)$, such that
for any given set of matrices $B_1, \ldots, B_6$ in $\UT(n, \Q)$ with $\log B_i \in \Lie_{\geq 1}(\log \mH)$
and $\sum_{i = 1}^6 \log B_i \in \Lie_{\geq 2}(\log \mH)$, we have
\begin{multline}\label{eq:H5mod}
    \sum_{\sigma \in \Sym_{6}} H_5(\log B_{\sigma(1)}, \ldots, \log B_{\sigma(6)})
    + \sum_{\sigma \in \Sym_{6}} H_5(\log B_{\sigma(j_1)}, \ldots, \log B_{\sigma(j_{12})}) \\
    \in \Lie_{\geq 6}(\log \mH) + \Lie_{\geq 2}(\Lie_{\geq 2}(\log \mH)).
\end{multline}
Namely, we can take $(j_{1}, j_{2}, \ldots, j_{12}) = (1, 2, 3, 4, 4, 5, 5, 6, 6, 1, 2, 3)$.
\end{restatable}
\begin{proof}[Sketch of proof of Lemma~\ref{lem:H5}]
For $x, y \in \mun$, denote $x \sim y$ if 
\[
x - y \in \Lie_{\geq 6}(\log \mH) + \Lie_{\geq 2}(\Lie_{\geq 2}(\log \mH)).
\]
The claim~\eqref{eq:H5mod} can be written as 
\begin{equation*}
    \sum_{\sigma \in \Sym_{6}} H_5(\log B_{\sigma(1)}, \ldots, \log B_{\sigma(6)})
    + \sum_{\sigma \in \Sym_{6}} H_5(\log B_{\sigma(j_1)}, \ldots, \log B_{\sigma(j_{12})}) \sim 0
\end{equation*}

By the Dynkin formula (Lemma~\ref{lem:dynkin}), the expressions
$
\sum_{\sigma \in \Sym_{6}} H_5(\log B_{\sigma(1)}, \ldots, \log B_{\sigma(6)})
$
and
$
\sum_{\sigma \in \Sym_{6}} H_5(\log B_{\sigma(j_1)}, \ldots, \log B_{\sigma(j_{12})})
$
can be expressed as a sum in the form of 
\begin{equation}\label{eq:def5alpha}
\sum_{\bj = (j_1, \ldots, j_5) \in \{1, \ldots, 6\}^5} \alpha_{\bj} \sum_{\sigma \in \Sym_{6}} \varphi_5(\log B_{\sigma(j_1)}, \ldots, \log B_{\sigma(j_5)}),
\end{equation}
where $\alpha_{\bj}$ are rational numbers.

Since $\sum_{i = 1}^6 \log B_i \in \Lie_{\geq 2}(\log \mH)$, for any tuple $\bj = (j_1, \ldots, j_5) \in \{1, \ldots, 6\}^5$, the expression $\sum_{\sigma \in \Sym_{6}} \varphi_5(\log B_{\sigma(j_1)}, \ldots, \log B_{\sigma(j_5)})$ is equivalent (under $\sim$) to a rational multiple of
\\
$\sum_{i \neq j} [[[[\log B_{i}, \log B_{j}], \log B_{j}], \log B_i], \log B_i]$.
(See Appendix~\ref{app:proofH} for detailed justification.)
In particular, using computer assistance, we can compute these rational multiples and show
\begin{align*}
    \sum_{\sigma \in \Sym_6} H_5(\log B_{\sigma(1)}, \ldots, \log B_{\sigma(6)})
    & \sim \sum_{i \neq j} [[[[\log B_{i}, \log B_{j}], \log B_{j}], \log B_i], \log B_i], \\
    \sum_{\sigma \in \Sym_6} H_5\big(\log B_{\sigma(j_1)}, \ldots, \log B_{\sigma(j_{12})}\big)
    & \sim - \sum_{i \neq j} [[[[\log B_{i}, \log B_{j}], \log B_{j}], \log B_i], \log B_i].
\end{align*}
This yields
\begin{equation*}
    \sum_{\sigma \in \Sym_{6}} H_5(\log B_{\sigma(1)}, \ldots, \log B_{\sigma(6)})
    + \sum_{\sigma \in \Sym_{6}} H_5(\log B_{\sigma(j_1)}, \ldots, \log B_{\sigma(j_{12})}) \sim 0.
\end{equation*}
The code for computer assistance can be found at \url{https://doi.org/10.6084/m9.figshare.20124146.v1}.
\end{proof}

\begin{rmk}
The added expression of $\Lie_{\geq 2}(\Lie_{\geq 2}(\log \mH))$ on the right hand side of Equation~\eqref{eq:H5mod} is crucial for its correctness.
In fact, we can consider Equation~\eqref{eq:H5mod} in the quotient Lie algebra $L \coloneqq \Lie_{\geq 1}(\log \mH) / \Lie_{\geq 2}(\Lie_{\geq 2}(\log \mH))$.
The Lie algebra $L$ is \emph{metabelian}, meaning $[[L,L],[L,L]] = 0$.
(Free) metabelian Lie algebras have significantly fewer dimensions compared to (free) Lie algebras having the same number of generators.
Moreover, free metabelian Lie algebras admit a relatively simple basis (sometimes called the \emph{Gr\"obner-Shirshov basis}) \cite{bokut1963basis}, making it computationally viable to find identities such as Equation~\eqref{eq:H5mod}. 
In our computer assisted proofs (see Appendix~\ref{app:proofH}), we are using a heavily modified version of this basis to compute Equation~\eqref{eq:H5mod} as well as Equations~\eqref{eq:H7mod} and \eqref{eq:H9mod} in the following lemmas.
\end{rmk}

\begin{restatable}{lem}{Hseven}\label{lem:H7}
Let $\mH \subset \UT(n, \Q)$ be a finite set of matrices.
There exist positive rational numbers $\alpha_1, \alpha_2$, as well as, for $s = 1, 2$, permutations $(j_{s,1}, j_{s,2}, \ldots, j_{s,16})$ of the tuple $(1, 1, 2, 2, \ldots, 8, 8)$, such that
for any given set of matrices $B_1, \ldots, B_{8}$ in $\UT(n, \Q)$ with $\log B_i \in \Lie_{\geq 1}(\log \mH)$
and $\sum_{i = 1}^{8} \log B_i \in \Lie_{\geq 2}(\log \mH)$, we have
\begin{multline}\label{eq:H7mod}
    \sum_{\sigma \in \Sym_{8}} H_7(\log B_{\sigma(1)}, \ldots, \log B_{\sigma(8)})
    + \sum_{s = 1}^2 \alpha_{s} \sum_{\sigma \in \Sym_{8}} H_7(\log B_{\sigma(j_{s,1})}, \ldots, \log B_{\sigma(j_{s,16})}) \\
    \in \Lie_{\geq 8}(\log \mH) + \Lie_{\geq 2}(\Lie_{\geq 2}(\log \mH)).
\end{multline}
Namely, we can take $\alpha_1 = \frac{1}{15}, \alpha_2 = \frac{8}{15}$, and 
\begin{align*}
   (j_{1,1}, j_{1,2}, \ldots, j_{1,16}) & = (1, 2, 3, 4, 5, 5, 6, 6, 7, 7, 8, 8, 1, 2, 3, 4), \\
   (j_{2,1}, j_{2,2}, \ldots, j_{2,16}) & = (1, 2, 3, 4, 5, 4, 6, 7, 1, 2, 8, 3, 5, 6, 7, 8).
\end{align*}
\end{restatable}
\begin{proof}[Sketch of proof of Lemma~\ref{lem:H7}]
Similar to Lemma~\ref{lem:H5}, define the equivalence relation 
\[
x \sim y \iff x - y \in \Lie_{\geq 8}(\log \mH) + \Lie_{\geq 2}(\Lie_{\geq 2}(\log \mH)).
\]
By the Dynkin formula (Lemma~\ref{lem:dynkin}), both the expressions
$\sum_{\sigma \in \Sym_{8}} H_7(\log B_{\sigma(1)}, \ldots, \log B_{\sigma(8)})$ and $\sum_{\sigma \in \Sym_{8}} H_7(\log B_{\sigma(j_1)}, \ldots, \log B_{\sigma(j_{16})})$
can be expressed as a sum in the form of 
\[
\sum_{\bj = (j_1, \ldots, j_7) \in \{1, \ldots, 8\}^7} \alpha_{\bj} \sum_{\sigma \in \Sym_{8}} \varphi_7(\log B_{\sigma(j_1)}, \ldots, \log B_{\sigma(j_7)}),
\]
where $\alpha_{\bj}$ are rational numbers.

Denote $C_i \coloneqq \log B_i, i = 1, \ldots, m$.
Since $\sum_{i = 1}^8 C_i \in \Lie_{\geq 2}(\log \mH)$, for any tuple $\bj = (j_1, \ldots, j_7) \in \{1, \ldots, 8\}^7$, the expression $\sum_{\sigma \in \Sym_{8}} \varphi_7(C_{\sigma(j_1)}, \ldots, C_{\sigma(j_7)})$ is equivalent to a linear combination (with rational coefficients) of 
\begin{align*}
\sum_{i \neq j} [[[[[[C_{i}, C_{j}], C_{j}], C_i], C_i], C_i], C_i], \\
\sum_{i \neq j} [[[[[[C_{i}, C_{j}], C_{j}], C_j], C_i], C_i], C_i],
\end{align*}
and 
\[
\sum_{\overset{i, j, k}{\text{ distinct}}} [[[[[[C_{i}, C_{j}], C_{j}], C_k], C_k], C_i], C_i].
\]
(See Section~\ref{app:proofH} for detailed justification.)
In fact, using computer assistance, we show that
\begin{multline*}
     \sum_{\sigma \in \Sym_8} H_7(\log B_{\sigma(1)}, \ldots, \log B_{\sigma(8)})
    \sim \frac{34}{15} \cdot \sum_{i \neq j} [[[[[[C_{i}, C_{j}], C_{j}], C_i], C_i], C_i], C_i] \\
    - \frac{34}{45} \cdot \sum_{i \neq j} [[[[[[C_{i}, C_{j}], C_{j}], C_j], C_i], C_i], C_i]
     + \frac{68}{15} \cdot \sum_{\overset{i, j, k}{\text{ distinct}}} [[[[[[C_{i}, C_{j}], C_{j}], C_k], C_k], C_i], C_i],
\end{multline*}
\begin{multline*}
     \sum_{\sigma \in \Sym_8} H_7\big(\log B_{\sigma(j_{1,1})}, \ldots, \log B_{\sigma(j_{1,16})}\big)
    \sim \frac{34}{15} \cdot\sum_{i \neq j} [[[[[[C_{i}, C_{j}], C_{j}], C_i], C_i], C_i], C_i] \\
    + \frac{238}{45} \cdot \sum_{i \neq j} [[[[[[C_{i}, C_{j}], C_{j}], C_j], C_i], C_i], C_i]
     - \frac{68}{5} \cdot\sum_{\overset{i, j, k}{\text{ distinct}}} [[[[[[C_{i}, C_{j}], C_{j}], C_k], C_k], C_i], C_i],
\end{multline*}
and
\begin{multline*}
    \sum_{\sigma \in \Sym_8} H_7\big(\log B_{\sigma(j_{2,1})}, \ldots, \log B_{\sigma(j_{2,16})}\big)
    \sim - \frac{68}{15} \cdot \sum_{i \neq j} [[[[[[C_{i}, C_{j}], C_{j}], C_i], C_i], C_i], C_i] \\
    + \frac{34}{45} \cdot \sum_{i \neq j} [[[[[[C_{i}, C_{j}], C_{j}], C_j], C_i], C_i], C_i]
    - \frac{34}{5} \cdot \sum_{\overset{i, j, k}{\text{ distinct}}} [[[[[[C_{i}, C_{j}], C_{j}], C_k], C_k], C_i], C_i].
\end{multline*}
This yields
\begin{equation*}
    \sum_{\sigma \in \Sym_{8}} H_7(\log B_{\sigma(1)}, \ldots, \log B_{\sigma(8)})
    + \sum_{s = 1}^2 \alpha_{s} \sum_{\sigma \in \Sym_{8}} H_7(\log B_{\sigma(j_{s,1})}, \ldots, \log B_{\sigma(j_{s,16})}) \sim 0,
\end{equation*}
where $\alpha_1 = \frac{1}{15}, \alpha_2 = \frac{8}{15}$.
The code can be found at \url{https://doi.org/10.6084/m9.figshare.20124113.v1}.
\end{proof}

\begin{restatable}{lem}{Hnine}\label{lem:H9}
Let $\mH \subset \UT(n, \Q)$ be a finite set of matrices.
There exist positive rational numbers $\alpha_1, \ldots, \alpha_6$, as well as, for $s = 1, \ldots, 6$, permutations $(j_{s,1}, j_{s,2}, \ldots, j_{s,20})$ of the tuple $(1, 1, 2, 2, \ldots, 10)$, such that
for any given set of matrices $B_1, \ldots, B_{10}$ in $\UT(n, \Q)$ with $\log B_i \in \Lie_{\geq 1}(\log \mH)$
and $\sum_{i = 1}^{10} \log B_i \in \Lie_{\geq 2}(\log \mH)$, we have
\begin{multline}\label{eq:H9mod}
    \sum_{\sigma \in \Sym_{10}} H_9(\log B_{\sigma(1)}, \ldots, \log B_{\sigma(10)})
    + \sum_{s = 1}^6 \alpha_{s} \sum_{\sigma \in \Sym_{10}} H_9(\log B_{\sigma(j_{s,1})}, \ldots, \log B_{\sigma(j_{s,20})}) \\
    \in \Lie_{\geq 10}(\log \mH) + \Lie_{\geq 2}(\Lie_{\geq 2}(\log \mH)).
\end{multline}
Namely, we can take $\alpha_1 = \frac{44566633}{13702661}, \alpha_2 = \frac{557040}{13702661}, \alpha_3 = \frac{205175}{3915046}, \alpha_4 = \frac{1307207}{13702661}, \alpha_5 = \frac{86275275}{27405322}$, $\alpha_6 = \frac{4105194}{1957523}$, and 
\begin{align*}
   (j_{1,1}, j_{1,2}, \ldots, j_{1,20}) & = (5, 4, 7, 10, 2, 8, 3, 8, 1, 9, 7, 6, 5, 6, 2, 3, 9, 10, 1, 4), \\
   (j_{2,1}, j_{2,2}, \ldots, j_{2,20}) & = (8, 3, 5, 7, 10, 6, 8, 2, 1, 10, 2, 4, 9, 1, 5, 9, 3, 6, 7, 4), \\
   (j_{3,1}, j_{3,2}, \ldots, j_{3,20}) & = (7, 10, 2, 6, 4, 9, 6, 4, 1, 5, 3, 5, 1, 9, 3, 7, 10, 2, 8, 8), \\
   (j_{4,1}, j_{4,2}, \ldots, j_{4,20}) & = (10, 2, 2, 6, 7, 1, 9, 3, 9, 4, 8, 7, 8, 5, 5, 1, 4, 10, 6, 3), \\
   (j_{5,1}, j_{5,2}, \ldots, j_{5,20}) & = (3, 5, 10, 1, 4, 8, 6, 9, 3, 2, 7, 6, 1, 10, 9, 7, 2, 4, 5, 8), \\
   (j_{6,1}, j_{6,2}, \ldots, j_{6,20}) & = (4, 7, 2, 10, 2, 1, 3, 5, 8, 1, 6, 9, 10, 7, 6, 8, 3, 5, 9, 4).
\end{align*}
\end{restatable}
\begin{proof}[Sketch of proof of Lemma~\ref{lem:H9}]
Similar to Lemma~\ref{lem:H5}, and Lemma~\ref{lem:H7}, denote $C_i = \log B_i$ for $i = 1, \ldots, m$.
Since $\sum_{i = 1}^{10} C_i \in \Lie_{\geq 2}(\log \mH)$, for any tuple $\bj = (j_1, \ldots, j_{9}) \in \{1, \ldots, 10\}^{9}$, the expression $\sum_{\sigma \in \Sym_{10}} \varphi_9(C_{\sigma(j_1)}, \ldots, C_{\sigma(j_9)})$ is equivalent to a linear combination (with rational coefficient) of 
\begin{align*}
& \sum_{i \neq j} [[[[[[[[C_{i}, C_{j}], C_{j}], C_i], C_i], C_i], C_i], C_i], C_i], \\
& \sum_{i \neq j} [[[[[[[[C_{i}, C_{j}], C_{j}], C_j], C_i], C_i], C_i], C_i], C_i], \\
& \sum_{i \neq j} [[[[[[[[C_{i}, C_{j}], C_{j}], C_j], C_j], C_i], C_i], C_i], C_i], \\
& \sum_{\overset{i, j, k}{\text{ distinct}}} [[[[[[[[C_{i}, C_{j}], C_{j}], C_k], C_k], C_i], C_i], C_i], C_i], \\
& \sum_{\overset{i, j, k}{\text{ distinct}}} [[[[[[[[C_{i}, C_{j}], C_{j}], C_j], C_k], C_k], C_i], C_i], C_i],
\end{align*}
and
\begin{align*}
    & \sum_{\overset{i, j, k}{\text{ distinct}}} [[[[[[[[C_{i}, C_{j}], C_{j}], C_k], C_k], C_k], C_i], C_i], C_i].
\end{align*}
Similar to the previous lemmas, the rest of the proof can be done by computer assistance.
The code can be found at \url{https://doi.org/10.6084/m9.figshare.20122979.v1}.
\end{proof}

Combining Lemma~\ref{lem:antisym}-\ref{lem:H9}, we obtain Proposition~\ref{prop:k}.

\propk*
\begin{proof}
For even $k$, Equation~\eqref{eq:propk} is satisfied by Lemma~\ref{lem:antisym} by taking $r = 0$ and pairing each permutation $\sigma$ with its \emph{reversal}:
\begin{align*}
& \sum_{\sigma \in \Sym_{k+1}} H_{k}(\log B_{\sigma(1)}, \ldots, \log B_{\sigma(k+1)}) \\
= & \; \frac{1}{2} \left(\sum_{\sigma \in \Sym_{k+1}} H_{k}(\log B_{\sigma(1)}, \ldots, \log B_{\sigma(k+1)}) + \sum_{\sigma \in \Sym_{k+1}} H_{k}(\log B_{\operatorname{rev}(\sigma)(1)}, \ldots, \log B_{\operatorname{rev}(\sigma)(k+1)})\right) \\
= & \; \frac{1}{2} \sum_{\sigma \in \Sym_{k+1}} \left(H_{k}(\log B_{\sigma(1)}, \ldots, \log B_{\sigma(k+1)}) + H_{k}(\log B_{\operatorname{rev}(\sigma)(1)}, \ldots, \log B_{\operatorname{rev}(\sigma)(k+1)})\right) \\
= & \; 0.
\end{align*}
Here $\operatorname{rev}(\sigma) \in \Sym_{k+1}$ is the reversal of $\sigma$, meaning $\operatorname{rev}(\sigma)(i) = \sigma(k+2 - i), i = 1, \ldots, k+1$.
For $k = 3, 5, 7, 9$, Equation~\eqref{eq:propk} is satisfied by Lemma~\ref{lem:H3}, \ref{lem:H5}, \ref{lem:H7} and \ref{lem:H9} respectively.
\end{proof}

\subsection{Proof of Proposition~\ref{prop:cone}}
In this subsection, we give the proof of Proposition~\ref{prop:cone}.

\cone*
\begin{proof}
For convenience, define $\mC_{d+1} \coloneqq \{0\}$. 
We will prove that, for all $k = 2, \ldots, d+1$, the cone $\langle P_i(t) \mid i \in \mI, t \in \Zpp \rangle_{\Qp} + \mC_k$ is equal to $\mC_1$.
Notice that the claim in the proposition is the case where $k = d+1$. We use induction on $k$.

For $k = 2$, since $\ba_{ki} \in \mC_2$ for $k \geq 2$, we have $P_i(t) + \mC_2 = t \ba_{1i} + \mC_2$, so
\[
\langle P_i(t) \mid i \in \mI, t \in \Zpp \rangle_{\Qp} + \mC_2 = \langle t \ba_{1i} \mid i \in \mI, t \in \Zpp \rangle_{\Qp} + \mC_2 \overset{\text{(ii)}}{=} \mC_1.
\]

For the induction step, suppose now that the cone $\langle P_i(t) \mid i \in \mI, t \in \Zpp \rangle_{\Qp} + \mC_k$ is equal to $\mC_1$, we want to prove that  $\langle P_i(t) \mid i \in \mI, t \in \Zpp \rangle_{\Qp} + \mC_{k+1}$ is equal to $\mC_1$.

By the induction hypothesis, there exist indices $i_1, \ldots, i_m \in \mI$ as well as positive integers $t_1, \ldots, t_m \in \Zpp$, such that 
\[
\langle P_{i_j}(t_j) \mid j = 1, \ldots, m \rangle_{\Qp} + \mC_k = \mC_1.
\]
Condition (ii) of the proposition shows that there exist indices $i'_1, \ldots, i'_{m'} \in \mI$ such that
\[
\langle \ba_{ki'_j} \mid j = 1, \ldots, m' \rangle_{\Qp} + \mC_{k+1} = \mC_{k}.
\]
Hence
\begin{equation}\label{eq:aC}
\langle P_{i_j}(t_j) \mid j = 1, \ldots, m \rangle_{\Qp} + \langle \ba_{ki'_j} \mid j = 1, \ldots, m' \rangle_{\Qp} + \mC_{k+1} = \mC_1.
\end{equation}
We show that there exists $t \in \Zpp$ such that
\[
\langle P_{i_j}(t_j) \mid j = 1, \ldots, m \rangle_{\Qp} + \langle P_{i'_j}(t) \mid j = 1, \ldots, m' \rangle_{\Qp} + \mC_{k+1} = \mC_1.
\]
Suppose the contrary, that for every $t \in \Zpp$, 
\[
\langle P_{i_j}(t_j) \mid j = 1, \ldots, m \rangle_{\Qp} + \langle P_{i'_j}(t) \mid j = 1, \ldots, m' \rangle_{\Qp} + \mC_{k+1} \subsetneq \mC_{1}.
\]
For any $\Qp$-cone $\mC$, define the \emph{normal cone} of $\mC$ as the set of vectors $\bv \in V$ such that $\bv^{\top} \bc \leq 0$ for all $\bc \in \mC$.
For every $t$, take a normalized vector $\bv_t \in \mC_1$ (meaning the norm of $\bv_t$ is 1) in the normal cone of $\langle P_{i_j}(t_j) \mid j = 1, \ldots, m \rangle_{\Qp} + \langle P_{i'_j}(t) \mid j = 1, \ldots, m' \rangle_{\Qp} + \mC_{k+1}$.
That is,
\begin{equation}\label{eq:innerprod}
\bv_t^{\top} P_{i_j}(t_j) \leq 0 \text{ for all } j, \quad \bv_t^{\top} P_{i'_j}(t) \leq 0 \text{ for all } j, \quad \bv_t \perp \mC_{k+1}.
\end{equation}
Such a vector must exist because $\langle P_{i_j}(t_j) \mid j = 1, \ldots, m \rangle_{\Qp} + \langle P_{i'_j}(t) \mid j = 1, \ldots, m' \rangle_{\Qp} + \mC_{k+1}$ is a strict sub-cone of the linear space $\mC_1$.
The $\R$-linear space $V_{\R} = V \otimes_{\Q} \R$ is finite dimensional and hence compact.
Embed $V$ into $V_{\R}$ canonically, then the sequence $\{\bv_t\}_{t \in \Zpp}$ has a limit point in $V_{\R}$.
Denote by $\bv_{lim}$ this limit point.
As all the vectors $\bv_t$ are in $\mC_1$, $\bv_{lim}$ must be in $\mC_1 \otimes_{\Q} \R$.
Since the inner product of $V$ canonically extends to the inner product of $V_{\R}$, taking the limit of \eqref{eq:innerprod}, we have
\begin{equation}\label{eq:innerprodlim1}
\bv_{lim}^{\top} \cdot P_{i_j}(t_j) \leq 0 \text{ for all } j, \quad \bv_{lim} \perp \mC_{k+1},
\end{equation}
and
\begin{align}\label{eq:innerprodlim2}
\bv_{lim}^{\top} \cdot \ba_{ki'_j} = \; & \bv_{lim}^{\top} \cdot \lim_{t \rightarrow \infty} \left(\frac{P_{i'_j}(t)}{t^k} - t \ba_{k+1, i'_j} - \cdots - t^{d-k} \ba_{d, i'_j} \right) \nonumber \\
= \; & \bv_{lim}^{\top} \cdot \lim_{t \rightarrow \infty} \frac{P_{i'_j}(t)}{t^k} 
\leq 0, \quad j = 1, \ldots, m'.
\end{align}
The second equality is due to $\ba_{k+1, i'_j}, \ldots, \ba_{d, i'_j} \in \mC_{k+1} \perp \bv_{lim}$.
Hence, \eqref{eq:innerprodlim1} and \eqref{eq:innerprodlim2} show that $\bv_{lim}^{\top} \cdot \bv \leq 0$ for all $\bv$ in the $\Rp$-cone
\[
\langle P_{i_j}(t_j) \mid j = 1, \ldots, m \rangle_{\Qp} + \langle \ba_{ki'_j} \mid j = 1, \ldots, m' \rangle_{\Qp} + \mC_{k+1} \overset{\text{Eq.}~\eqref{eq:aC}}{=} \mC_1.
\]
Since $\mC_1$ is a linear space, $\bv_{lim}$ is non-zero (it has norm one) and is in $\mC_1 \otimes_{\Q} \R$, this yields a contradiction.
We have thus shown that there exists $t \in \Zpp$ such that
\[
\langle P_{i_j}(t_j) \mid j = 1, \ldots, m \rangle_{\Qp} + \langle P_{i'_j}(t) \mid j = 1, \ldots, m' \rangle_{\Qp} + \mC_{k+1} = \mC_1.
\]
Since $P_i(t) \in \mC_1, i \in \mI, t \in \Zpp$, this means
\[
\langle P_i(t) \mid i \in \mI, t \in \Zpp \rangle_{\Qp} + \mC_{k+1} = \mC_1,
\]
concluding the induction.

Finally, take $k = d+1$.
This yields $\langle P_i(t) \mid i \in \mI, t \in \Zpp \rangle_{\Qp} = \mC_1$.
\end{proof}

\subsection{Full proof of Theorem~\ref{thm:equivn}}\label{subsec:fullproof}
In this subsection, with Propositions~\ref{prop:lieperm}~-~\ref{prop:cone} at our disposal, we will show the proof of Theorem~\ref{thm:equivn}.
First, we need the following lemma.
\begin{lem}\label{lem:filalg}
Let $\mH$ be a finite subset of the Lie algebra $\mun$.
Let $W, V$ be linear subspaces of $\Lie_{\geq 1}(\mH)$ such that $W + \Lie_{\geq 2}(V) = V$, then $\Lie_{\geq 2}(W) = \Lie_{\geq 2}(V)$.
\end{lem}
\begin{proof}
Since $W + \Lie_{\geq 2}(V) = V$, we have $W \subseteq V$, and thus $\Lie_{\geq 2}(W) \subseteq \Lie_{\geq 2}(V)$.
Therefore, it suffices to prove the opposite inclusion $\Lie_{\geq 2}(W) \supseteq \Lie_{\geq 2}(V)$.

Note that since $W, V$ are linear spaces, the sets $[V]_k, [W]_k$ are also linear spaces for all $k = 1, \ldots, n$.
We use induction on $k$ to show that
\begin{equation}\label{eq:filalgind}
    [V]_k \subseteq [W]_k + \Lie_{\geq k+1}(V).
\end{equation}
For $k = 1$ this immediately results from the equation $W + \Lie_{\geq 2}(V) = V$.
Suppose Equation~\eqref{eq:filalgind} hold for $k-1$. 
Then, take any elements $x \in V, y \in [V]_{k-1}$, by the induction hypothesis and by $W + \Lie_{\geq 2}(V) = V$, there exist $x' \in W, y' \in [W]_{k-1}$, such that $x - x' \in \Lie_{\geq 2}(V), y - y' \in \Lie_{\geq k}(V)$.
Then, 
\begin{align*}
    [y, x] & = [y', x'] + [y - y', x'] + [y, x - x'] \\
    & \in [[W]_{k-1}, W] + [\Lie_{\geq k}(V), W] + [[V]_{k-1}, \Lie_{\geq 2}(V)] \\
    & \subseteq [W]_k + [\Lie_{\geq k}(V), V] + [\Lie_{\geq k-1}(V), \Lie_{\geq 2}(V)] \\
    & \subseteq [W]_k + \Lie_{\geq k+1}(V).
\end{align*}
Taking the linear span for all $x \in V, y \in [V]_{k-1}$ shows $[V]_k \subseteq [W]_k + \Lie_{\geq k+1}(V)$, concluding the induction.

Now, for any $l = 2, \ldots, d$, take the sum of Equation~\eqref{eq:filalgind} for $k = l, \ldots, d$, we have
\[
\Lie_{\geq l}(V) = \sum_{k \geq l} [V]_k \subseteq \sum_{k \geq l} [W]_k + \sum_{k \geq l} \Lie_{\geq k+1}(V) = \Lie_{\geq l}(W) + \Lie_{\geq l+1}(V).
\]
Therefore,
\begin{align*}
    & \Lie_{\geq 2}(V) \\
    \subseteq & \; \Lie_{\geq 2}(W) + \Lie_{\geq 3}(V) \\
    \subseteq & \; \Lie_{\geq 2}(W) + \Lie_{\geq 3}(W) + \Lie_{\geq 4}(V) \\
    & \vdots \\
    \subseteq & \; \Lie_{\geq 2}(W) + \Lie_{\geq 3}(W) + \cdots + \Lie_{\geq n}(W) \\
    = & \; \Lie_{\geq 2}(W).
\end{align*}
This shows the inclusion $\Lie_{\geq 2}(W) \supseteq \Lie_{\geq 2}(V)$.
\end{proof}

Let $\mG = \{A_1, \ldots, A_K\}$ be a finite alphabet of elements in $\UT(n, \Q)$.
For any vector $\bl = (\ell_1, \ldots, \ell_K) \in \Zp^K$, define inductively the following $\Q$-cones $\mR_k(\bl)$ for $k = 11, 10, \ldots, 2$:
    \begin{align}\label{eq:defR1}
    & \mR_{11}(\bl) \coloneqq \{0\}, \\
    & \mR_k(\bl) \coloneqq \mR_{k+1}(\bl) +
    \bigg\langle H_k(\log B_1, \ldots, \log B_m) \;\bigg|\;  B_i \in \mG^*, \sum_{i = 1}^m \PI^{\mG}(B_i) \in \{\bl, 2 \bl\} \bigg\rangle_{\Qp}.
    \end{align}
That is, $\mR_k(\bl)$ is the $\Qp$-cone generated by the elements $H_j(\log B_1, \ldots, \log B_m), j \geq k$, where $B_1, \ldots, B_m$ are \emph{words} in $\mG^*$, and the Parikh images of $B_i$ sum up to $\bl$ or $2 \bl$.
Recall the definition of
\[
\log \mG_{\supp(\bl)} \coloneqq \{\log A_i \mid i \in \supp(\bl)\}
\]
as the set of logarithm of matrices in $\mG$ whose index appears in $\supp(\bl)$.
Combining Proposition~\ref{prop:lieperm} and \ref{prop:k}, we can show the following proposition that characterizes the cones $\mR_{k}(\log \mG_{\supp(\bl)})$ up to the quotient by $\Lie_{\geq 2}(\Lie_{\geq 2}(\log \mG_{\supp(\bl)}))$.

\begin{restatable}{prop}{linspacen}\label{prop:linspacen}
Let $\mG = \{A_1, \ldots, A_K\}$ be a finite set of matrices in $\UT(n, \Q)$ that satisfies $[\log \mG]_{11} = \{0\}$.
Let $\bl = (\ell_1, \ldots, \ell_K) \in \Zp^K$ be a non-zero vector that satisfies $\sum_{i = 1}^K \ell_i \log A_i \in \Lie_{\geq 2}(\log \mG_{\supp(\bl)})$ as well as $\ell_i \geq 10$ for all $i \in \supp(\bl)$.
Consider the quotient linear space $\mun / \Lie_{\geq 2}(\Lie_{\geq 2}(\log \mG_{\supp(\bl)}))$. 

For any set $\mC \subseteq \mun$, denote by $\overline{\mC}$ the subset of $\mun / \Lie_{\geq 2}(\Lie_{\geq 2}(\log \mG_{\supp(\bl)}))$ consisting of the equivalence classes $c + \Lie_{\geq 2}(\Lie_{\geq 2}(\log \mG_{\supp(\bl)})), c \in \mC$.
Then for all $k \leq 11$, the cone $\overline{\mR_k(\bl)}$
is equal to the linear space
$\overline{\Lie_{\geq k}(\log \mG_{\supp(\bl)})}$.
\end{restatable}
\begin{proof}

We show that the claim is true for $k = 11, 10, \ldots, 2,$ using induction with reverse order on $k$.
For $k = 11$, we have $\overline{\mR_{11}(\bl)} = \overline{\Lie_{\geq 11}(\log \mG_{\supp(\bl)})} = \{0\}$ because $[\log \mG]_{11} = \{0\}$.
Now for some $10 \geq k \geq 2$, suppose $\overline{\mR_{k+1}(\bl)} = \overline{\Lie_{\geq {k+1}}(\log \mG_{\supp(\bl)})}$ by induction hypothesis.
We will show that $\overline{\mR_k(\bl)} = \overline{\Lie_{\geq k}(\log \mG_{\supp(\bl)})}$.

    First, we show that for any $i_1, i_2, \ldots, i_k \in \supp(\bl)$, we have
    \[
    [\ldots[[\log A_{i_1}, \log A_{i_2}], \log A_{i_3}], \ldots , \log A_{i_k}] \in \mR_k(\bl) + \Lie_{\geq 2}(\Lie_{\geq 2}(\log \mG_{\supp(\bl)})).
    \]
    Take a tuple of words $(B'_1, \ldots, B'_{k+1})$ with $B'_1 = A_{i_1}, B'_2 = A_{i_2}, \ldots, B'_k = A_{i_k}, B'_{k+1} \in \mG^*$, such that $\sum_{i=1}^{k+1} \PI^{\mG}(B'_i) = \bl$.
    Such a tuple can always be found because $\bl$ satisfies $\ell_i \geq 10 \geq k, i \in \supp(\bl)$.
    For this tuple, the BCH formula gives us
    \[
    \sum_{i = 1}^{k+1} \log B'_i \in \sum_{i = 1}^K \ell_i \log A_i + \Lie_{\geq 2}(\log \mG_{\supp(\bl)}) \subseteq \Lie_{\geq 2}(\log \mG_{\supp(\bl)}).
    \]
    Hence, for any $\sigma \in \Sym_k$, Proposition~\ref{prop:k} shows that 
    \begin{align}\label{eq:minusHin}
    & - H_k\left(\log B'_{\sigma(1)}, \log B'_{\sigma(2)}, \ldots, \log B'_{\sigma(k)}, \log B'_{k+1}\right) \nonumber \\
    \in & \bigg\langle H_k(\log B_1, \ldots, \log B_{k+1}) \;\bigg|\; B_i \in \mG^*, \sum_{i=1}^{k+1} \PI^{\mG}(B_i) = \bl \bigg\rangle_{\Qp} \nonumber \\
    & \quad + \bigg\langle H_k(\log B_1, \ldots, \log B_{2k+2}) \;\bigg|\;  B_i \in \mG^*, \sum_{i=1}^{2k+2} \PI^{\mG}(B_i) = 2 \bl \bigg\rangle_{\Qp} \nonumber \\
    & \quad + \Lie_{\geq k+1}(\log \mG_{\supp(\bl)}) + \Lie_{\geq 2}(\Lie_{\geq 2}(\log \mG_{\supp(\bl)})) \nonumber \\
    \subseteq & \; \mR_k(\bl) + \Lie_{\geq k+1}(\log \mG_{\supp(\bl)}) + \Lie_{\geq 2}(\Lie_{\geq 2}(\log \mG_{\supp(\bl)})) \nonumber \\
    = & \; \mR_k(\bl) + \Lie_{\geq 2}(\Lie_{\geq 2}(\log \mG_{\supp(\bl)})).
    \end{align}
    The last equality come from $\overline{\Lie_{\geq k+1}(\log \mG_{\supp(\bl)})} = \overline{\mR_{k+1}(\bl)} \subseteq \overline{\mR_k(\bl)}$ by the induction hypothesis.
    
    Hence, by Proposition~\ref{prop:lieperm},
    \begin{align*}
    & [\ldots [[\log A_{i_1}, \log A_{i_2}], \log A_{i_3}], \ldots, \log A_{i_k}] \\
    = & \; [\ldots [[\log B'_1, \log B'_2], \log B'_3], \ldots, \log B'_k] \\
    = & \; \sum_{\sigma \in \Sym_k} \mu(\sigma) H_{k}(\log B'_{\sigma(1)}, \log B'_{\sigma(2)}, \ldots, \log B'_{\sigma(k)}, \log B'_{k+1}) \\
    \in & \; \mR_k(\bl) + \Lie_{\geq 2}(\Lie_{\geq 2}(\log \mG_{\supp(\bl)})).
    \end{align*}
    The last inclusion comes from the fact that both the expressions $H_{k}(\log B'_{\sigma(1)}, \ldots, \log B'_{\sigma(k)}, \log B'_{k+1})$ and $- H_{k}(\log B'_{\sigma(1)}, \ldots, \log B'_{\sigma(k)}, \log B'_{k+1})$ are in the cone $\mR_k(\bl) + \Lie_{\geq 2}(\Lie_{\geq 2}(\log \mG_{\supp(\bl)}))$ (by Equation~\eqref{eq:minusHin}).
    Therefore for every $\sigma \in \Sym_k$, regardless of the sign which $\mu(\sigma)$ takes, the summand $\mu(\sigma) H_{k}(\log B'_{\sigma(1)}, \ldots, \log B'_{\sigma(k)}, \log B'_{k+1})$ is in the cone $\mR_k(\bl) + \Lie_{\geq 2}(\Lie_{\geq 2}(\log \mG_{\supp(\bl)}))$.
    
    Therefore, $[\log \mG_{\supp(\bl)}]_k \subseteq \mR_k(\bl) + \Lie_{\geq 2}(\Lie_{\geq 2}(\log \mG_{\supp(\bl)}))$; that is, $\overline{[\log \mG_{\supp(\bl)}]_k} \subseteq \overline{\mR_k(\bl)}$.
    And since $\overline{\Lie_{\geq k+1} (\log \mG_{\supp(\bl)})} \subseteq \overline{\mR_{k+1}(\bl)} \subseteq \overline{\mR_k(\bl)}$ by the induction hypothesis, we have
    \begin{align}\label{eq:Toddsubset1}
    \overline{\Lie_{\geq k} (\log \mG_{\supp(\bl)})} = \overline{\left\langle [\log \mG_{\supp(\bl)}]_k \right\rangle_{\Q}} + \overline{\Lie_{\geq k+1} (\log \mG_{\supp(\bl)})}
    \subseteq \overline{\mR_k(\bl)}.
    \end{align}
    
    Next, take any tuple $(B_1, \ldots, B_m) \in \left(\mG^*\right)^m$, $\sum_{i=1}^m \PI^{\mG}(B_i) = \bl$ or $2 \bl$.
    Note that $\log B_i \in \Lie_{\geq 1}(\log \mG_{\supp(\bl)}), i = 1, \ldots, m,$ by the BCH formula.
    Hence, the expression $H_k(\log B_1, \ldots, \log B_m)$ can be written as a linear combination of elements in $\left[\Lie_{\geq 1}(\log \mG_{\supp(\bl)})\right]_k$. 
    That is,
    \begin{multline}\label{eq:Toddsubset2}
    \mR_k(\bl) \subseteq \left\langle\left[\Lie_{\geq 1}(\log \mG_{\supp(\bl)})\right]_k\right\rangle_{\Q} + \mR_{k+1}(\bl) \\
    \subseteq \Lie_{\geq k}(\log \mG_{\supp(\bl)}) + \mR_{k+1}(\bl) = \Lie_{\geq {k+1}}(\log \mG_{\supp(\bl)}).
    \end{multline}
    Combining \eqref{eq:Toddsubset1} and \eqref{eq:Toddsubset2} we have the desired equality.
    This concludes the induction and thus the whole proof.
\end{proof}

We now prove Theorem~\ref{thm:equivn}.
Although part (i) has already been proven when the theorem is first stated, we will restate it for the sake of completeness.

\equivn*
\begin{proof}
(i)
Let $w$ be a word with $\PI^{\mG}(w) = \bl$. 
Write $w = B_1 B_2 \cdots B_m$ $B_i \in \mG, i = 1, \ldots, m$.
Regrouping by letters, we have $\sum_{i = 1}^K \ell_i \log A_i = \sum_{i = 1}^m \log B_i$.

If $\log w = 0$, then by the BCH formula, we have
\[
\sum_{i = 1}^m \log B_i + \sum_{k=2}^{n-1} H_k(\log B_1, \ldots, \log B_m) = \log(B_1 B_2 \cdots B_m) = 0.
\]
The higher order terms $H_k, k \geq n$ vanish because $[\log \mG]_{n} = \{0\}$ (a consequence of $\mG \subseteq \UT(n, \Q)$).
Therefore, $\sum_{i = 1}^K \ell_i \log A_i = \sum_{i = 1}^m \log B_i = - \sum_{k=2}^{n-1} H_k(\log B_1, \ldots, \log B_m)$.

Since the Parikh image of the word $B_1 \cdots B_m$ is $\bl$, the matrices $B_i$ all lie in the subset $\{A_i \mid i \in \supp(\bl)\}$ of $\mG$.
Therefore, $\log B_i \in \log \mG_{\supp(\bl)}$ for all $i$.
By Theorem~\ref{thm:BCH}, for all $k \geq 2$ we have $- H_k(\log B_1, \ldots, \log B_m) \in \big\langle[\{\log B_i \mid i = 1, \ldots, m\}]_k \big\rangle_{\Q} \subseteq \Lie_{\geq k}(\log \mG_{\supp(\bl)}) \subseteq \Lie_{\geq 2}(\log \mG_{\supp(\bl)})$.
Therefore, we have
$
\sum_{i = 1}^K \ell_i \log A_i = - \sum_{k=2}^{n-1} H_k(\log B_1, \ldots, \log B_m) \in \Lie_{\geq 2}(\log \mG_{\supp(\bl)}).
$

(ii) Suppose condition~\eqref{eq:condn} hold for the vector $\bl$.
Resonating Example~\ref{example:u4}, our proof for (ii) proceeds in four steps.
Now we give an overview of each step.
As the first step, we want to construct some matrices $A'_1, \ldots, A'_{K'} \in \sgmG$, such that
\begin{equation}\label{eq:UTnstep1goal}
\langle \log A'_i \mid i = 1, \ldots, K' \rangle_{\Qp} + \Lie_{\geq 2}(\Lie_{\geq 2}(\log \mG_{\supp(\bl)})) = \Lie_{\geq 2}(\log \mG_{\supp(\bl)}) + \Lie_{\geq 2}(\Lie_{\geq 2}(\log \mG_{\supp(\bl)})).
\end{equation}
The candidates for the matrices $A'_1, \ldots, A'_{K'}$ are of the form $B_1^t \cdots B_m^t$, where $m \geq 1$, $t \in \Zpp$, $B_i \in \mG^{*}, i = 1, \ldots, m$ and $\sum_{i = 1}^m \PI^{\mG}(B_i) = \bl$ or $2 \bl$.
The general strategy is to invoke Proposition~\ref{prop:cone} 
while using Proposition~\ref{prop:linspacen} to guarantee that the conditions (i) and (ii) of Proposition~\ref{prop:cone} are satisfied.

As the second step, we work in the new alphabet $\mG' = \{A'_1, \ldots, A'_{K'}\}$ of matrices found in the previous step.
We want to fabricate some matrices $A''_1, \ldots, A''_{K''} \in \langle\mG'\rangle$, such that
\begin{multline}\label{eq:UTnstep2goal}
\langle \log A''_i \mid i = 1, \ldots, K'' \rangle_{\Qp} + \Lie_{\geq 2}(\Lie_{\geq 2}(\Lie_{\geq 2}(\log \mG_{\supp(\bl)}))) \\
= \Lie_{\geq 2}\left(\Lie_{\geq 2}(\log \mG_{\supp(\bl)})\right) + \Lie_{\geq 2}(\Lie_{\geq 2}(\Lie_{\geq 2}(\log \mG_{\supp(\bl)}))).
\end{multline}
The candidates for the matrices $A''_1, \ldots, A''_{K''}$ are of the form $B_1^t \cdots B_m^t$, where $B_i \in \left(\mG'\right)^{*}, i = 1, \ldots, m$. The idea is to again invoke Proposition~\ref{prop:cone} and to use Proposition~\ref{prop:linspacen} for the new alphabet $\mG'$ and a suitable vector $\bl'$.

As the third step, we work in the new alphabet $\mG'' = \{A''_1, \ldots, A''_{K''}\}$ of matrices found in the previous step.
We want to fabricate some matrices $A'''_1, \ldots, A'''_{K'''} \in \langle\mG''\rangle$, such that
\begin{equation}\label{eq:UTnstep3goal}
\langle \log A'''_i \mid i = 1, \ldots, K'' \rangle_{\Qp}
= \Lie_{\geq 2}(\Lie_{\geq 2}(\Lie_{\geq 2}(\log \mG_{\supp(\bl)}))).
\end{equation}
(Note that $\Lie_{\geq 2}(\Lie_{\geq 2}(\Lie_{\geq 2}(\Lie_{\geq 2}(\log \mG_{\supp(\bl)})))) = \{0\}$.) 
The candidates for $A'''_1, \ldots, A'''_{K'''}$ are of the form $B_1^t \cdots B_m^t$, where $B_i \in \left(\mG''\right)^{*}, i = 1, \ldots, m$. The idea is to again invoke Proposition~\ref{prop:cone} and to use Proposition~\ref{prop:linspacen} for the new alphabet $\mG''$ and a suitable vector $\bl''$.

As the fourth and last step, we work in the new alphabet $\mG''' = \{A'''_1, \ldots, A'''_{K'''}\}$ of matrices found in the previous step.
We then observe that the matrices $A'''_1, \ldots, A'''_{K'''}$ commute with each other, because $\Lie_{\geq 2}(\Lie_{\geq 2}(\Lie_{\geq 2}(\Lie_{\geq 2}(\log \mG_{\supp(\bl)})))) = \{0\}$.
Hence, it is very easy to search for the desired non-empty word $w \in \left(\mG'''\right)^{*}$ with $\log w = 0$.

We now give the detailed account of each step.
\begin{enumerate}[noitemsep, label = (\arabic*)]
    \item \textbf{Find matrices $A'_1, \ldots, A'_{K'} \in \sgmG$ satisfying condition \eqref{eq:UTnstep1goal}.}
    Since the right hand side of Equation~\eqref{eq:condn} is a linear space, we can replace $\bl$ by $10\bl$, and thus suppose $\bl$ satisfy $\ell_i \geq 10, i \in \supp(\bl)$.
    Since $\Zpp \cdot 10\bl \subseteq \Zpp \cdot \bl$, the resulting word $w$ will still satisfy $\PI^{\mG}(w) \in \Zpp \cdot \bl$.
    
    Since $\bl$ satisfies $\sum_{i = 1}^K \ell_i \log A_i \in \Lie_{\geq 2}(\log \mG_{\supp(\bl)})$,
    we are able to use Proposition~\ref{prop:linspacen} for the vector $\bl$.
    Our aim is to apply Proposition~\ref{prop:cone} in the quotient space
    \[
    V \coloneqq \mun/\Lie_{\geq 2}(\Lie_{\geq 2}(\log \mG_{\supp(\bl)})),
    \]
    for the index set 
    \[
    \mI \coloneqq \bigg\{(B_1, \ldots, B_m) \;\bigg|\; m \geq 1, B_i \in \mG^{*}, \sum_{i = 1}^m \PI^{\mG}(B_i) \in \{\bl, 2 \bl\}\bigg\},
    \]
    that is, the set of tuples of words whose concatenation has Parikh image $\bl$ or $2 \bl$.
    For any element $\bx \in \mun$, denote by $\overline{\bx} \coloneqq \bx + \Lie_{\geq 2}(\Lie_{\geq 2}(\log \mG_{\supp(\bl)})$ its equivalence class in 
    \[
    V = \mun/\Lie_{\geq 2}(\Lie_{\geq 2}(\log \mG_{\supp(\bl)})).
    \]
    For any tuple $b = (B_1, \ldots, B_m) \in \mI$, consider the vectors in $V$:
    \begin{align*}
    \ba_{1b} & \coloneqq \overline{\sum_{i = 1}^m \log B_i}, \\
    \ba_{kb} & \coloneqq \overline{H_{k}(\log B_1, \ldots, \log B_m)}, \quad k = 2, \ldots, 10,
    \end{align*}
    and
    \[
    P_{b}(t) \coloneqq \overline{\log(B_1^t \cdots B_m^t)} = t \ba_{1b} + \sum_{k = 2}^{10} t^k \ba_{kb},
    \]
    coming from the BCH formula for $B_1^t, \ldots, B_m^t$.
    We now apply Proposition~\ref{prop:cone} to these vectors:
    we need to verify that the cones $\mC_k, k = 10, \ldots, 1$ as defined in Proposition~\ref{prop:cone} are indeed linear spaces.
    Proposition~\ref{prop:linspacen} shows that 
    \[
    \mC_{10} = \langle \ba_{{10}b} \mid b \in \mI \rangle_{\Qp} = \overline{\mR_{10}(\bl)}
    \]
    and
    \[
    \mC_k = \langle \ba_{kb} \mid b \in \mI \rangle_{\Qp} + \mC_{k+1} 
    = \overline{\mR_k(\bl)} , \quad k = 9, \ldots, 2,
    \]
    are linear subspaces of $\mun/\Lie_{\geq 2}(\Lie_{\geq 2}(\log \mG_{\supp(\bl)})$. 
    Furthermore, by the condition \\
    $\sum_{i = 1}^K \ell_i \log A_i \in \Lie_{\geq 2}(\log \mG_{\supp(\bl)})$, we have 
    \[
    \ba_{1b} \in \left\{\overline{\sum_{i = 1}^K \ell_i \log A_i}, 2 \cdot \overline{\sum_{i = 1}^K \ell_i \log A_i} \right\} \subseteq \overline{\Lie_{\geq 2}(\log \mG_{\supp(\bl)})} = \overline{\mR_2(\bl)}
    \]
    for all $b \in \mI$. 
    Hence, 
    \[
    \mC_1 = \langle \ba_{1b} \mid b \in \mI \rangle_{\Qp} + \mC_{2} = \overline{\mR_2(\bl)}
    = \overline{\Lie_{\geq 2}(\log \mG_{\supp(\bl)})}
    \]
    is also a linear space.
    The conditions (i) and (ii) in Proposition~\ref{prop:cone} are thus satisfied. 
    We can thus apply Proposition~\ref{prop:cone}, which yields
    \[
    \left\langle P_b(t) \mid b \in \mI, t \in \Zp \right\rangle_{\Qp} = \mC_1 = \overline{\Lie_{\geq 2}(\log \mG_{\supp(\bl)})}.
    \]
    In other words,
    \begin{equation*}
    \bigg\langle \overline{\log(B_1^t \cdots B_m^t)} \;\bigg|\; t \in \Zp,  m \geq 1, B_i \in \mG^{*}, \sum_{i = 1}^m \PI^{\mG}(B_i) \in \{\bl, 2 \bl\} \bigg\rangle_{\Qp} \\
    = \overline{\Lie_{\geq 2}(\log \mG_{\supp(\bl)})}.
    \end{equation*}
    Since $\mun/\Lie_{\geq 2}(\Lie_{\geq 2}(\log \mG_{\supp(\bl)}))$ is of finite dimension, this shows that there exist $K' > 0$ tuples of words $(B_{11}, \ldots, B_{1m})$, $\ldots$, $(B_{K'1}, \ldots, B_{K'm})$ with $\sum_{i = 1}^m \PI^{\mG}(B_{ji}) = \bl$ or $2 \bl$ for all $j \in \{1, \ldots, K'\}$, as well as positive integers $t_1, \ldots, t_{K'} \in \Zpp$, such that
    \begin{multline*}
    \langle \log(B_{i1}^{t_i} \cdots B_{im}^{t_i}) \mid i = 1, \ldots, K' \rangle_{\Qp} + \Lie_{\geq 2}(\Lie_{\geq 2}(\log \mG_{\supp(\bl)})) \\
    = \Lie_{\geq 2}(\log \mG_{\supp(\bl)}) + \Lie_{\geq 2}(\Lie_{\geq 2}(\log \mG_{\supp(\bl)})).
    \end{multline*}
    Hence, the matrices $A'_i = B_{i1}^{t_i} \cdots B_{im}^{t_i}, i = 1, \ldots, K'$ satisfy the Equation~\eqref{eq:UTnstep1goal}.
    Define a new alphabet $\mG' = \{A'_1, \ldots, A'_{K'}\} \subseteq G$.
    
    \item \textbf{Find matrices $A''_1, \ldots, A''_{K''} \in \langle\mG'\rangle$ satisfying condition \eqref{eq:UTnstep2goal}.}
    Since the right hand side of Equation~\eqref{eq:UTnstep1goal} is a linear space, we have $- \log A'_j \in \langle \log A'_i \mid i = 1, \ldots, K' \rangle_{\Qp} + \Lie_{\geq 2}(\Lie_{\geq 2}(\log \mG_{\supp(\bl)}))$ for $j = 1, \ldots, K'$.
    Hence, there exists a non-zero vector $\bl' = (\ell'_1, \ldots, \ell'_{K'})$ in $\Zp^{K'}$, satisfying $\supp(\bl') = \{1, \ldots, K'\}$, $\ell'_i \geq 10$ for all $i \in \{1, \ldots, K'\}$, and 
    \begin{equation}\label{eq:sumprime}
    \sum_{i = 1}^{K'} \ell'_i \log A'_i \in \Lie_{\geq 2}(\Lie_{\geq 2}(\log \mG_{\supp(\bl)})).
    \end{equation}
    
    Define $\log \mG'_{\supp(\bl')} \coloneqq \{\log A'_i \mid i \in \supp(\bl')\} = \log \mG'$, because $\supp(\bl') = \{1, \ldots, K'\}$.
    First, we claim that
    \begin{equation}\label{eq:LeqLL}
    \Lie_{\geq 2}(\Lie_{\geq 2}(\log \mG_{\supp(\bl)})) = \Lie_{\geq 2}(\log \mG'_{\supp(\bl')}).
    \end{equation}
    Indeed, Equation~\eqref{eq:UTnstep1goal} shows that 
    \begin{multline*}
    \left\langle \log \mG'_{\supp(\bl')} \right\rangle_{\Q} + \Lie_{\geq 2}(\Lie_{\geq 2}(\log \mG_{\supp(\bl)}))
    \\ = \Lie_{\geq 2}(\log \mG_{\supp(\bl)}) + \Lie_{\geq 2}(\Lie_{\geq 2}(\log \mG_{\supp(\bl)})) = \Lie_{\geq 2}(\log \mG_{\supp(\bl)}).
    \end{multline*}
    Applying Lemma~\ref{lem:filalg} with $W = \left\langle \log \mG'_{\supp(\bl')} \right\rangle_{\Q}, V = \Lie_{\geq 2}(\log \mG_{\supp(\bl)})$ to the above equation yields the equality~\eqref{eq:LeqLL}.
    Consequently, we have 
    \[
    \sum_{i = 1}^{K'} \ell'_i \log A'_i \in \Lie_{\geq 2}(\log \mG'_{\supp(\bl')})
    \]
    by \eqref{eq:sumprime}.
    Apply Proposition~\ref{prop:linspacen} for the alphabet $\mG'$ and the vector $\bl'$, then we have that, in the quotient space
    \[
    \mun / \Lie_{\geq 2}(\Lie_{\geq 2}(\log \mG'_{\supp(\bl')})),
    \]
    the equations $\overline{\mR_k(\bl')} = \overline{\Lie_{\geq k}(\log \mG'_{\supp(\bl')})}, k = 10, \ldots, 2$, hold.
    Then, applying Proposition~\ref{prop:cone} in the quotient linear space 
    \[
    V \coloneqq \mun/\Lie_{\geq 2}(\Lie_{\geq 2}(\log \mG'_{\supp(\bl')}))
    \]
    as in the previous step, we have
    \begin{multline*}
    \bigg\langle \overline{\log(B_1^t \cdots B_m^t)} \;\bigg|\; t \in \Zp,  m \geq 1, B_i \in \left(\mG'\right)^{*}, \sum_{i=1}^m \PI^{\mG'}(B_i) \in \{\bl', 2 \bl'\} \bigg\rangle_{\Qp} \\
    = \overline{\Lie_{\geq 2}(\log \mG'_{\supp(\bl')})}.
    \end{multline*}
    Hence, there exist $K'' > 0$ tuples of words $(B'_{11}, \ldots, B'_{1m})$, $\ldots$, $(B'_{K''1}, \ldots, B'_{K''m})$ in $\left(\mG'\right)^*$ with $\sum_{i=1}^m \PI^{\mG'}(B'_{ji}) = \bl'$ or $2 \bl'$ for all $j$, as well as positive integers $t'_1, \ldots, t'_{K''} \in \Zpp$, such that
    \begin{multline}\label{eq:step2res}
    \langle \log({B'}_{i1}^{t'_i} \cdots {B'}_{im}^{t'_i}) \mid i = 1, \ldots, K'' \rangle_{\Qp} + \Lie_{\geq 2}(\Lie_{\geq 2}(\log \mG'_{\supp(\bl')})) \\
    = \Lie_{\geq 2}(\log \mG'_{\supp(\bl')}) + \Lie_{\geq 2}(\Lie_{\geq 2}(\log \mG'_{\supp(\bl')})).
    \end{multline}
    Substituting with $\Lie_{\geq 2}(\log \mG'_{\supp(\bl')}) = \Lie_{\geq 2}(\Lie_{\geq 2}(\log \mG_{\supp(\bl)}))$, Equation~\eqref{eq:step2res} can be rewritten as
    \begin{multline*}
    \langle \log(B_{i1}^{t'_i} \cdots B_{im}^{t'_i}) \mid i = 1, \ldots, K'' \rangle_{\Qp} + \Lie_{\geq 2}(\Lie_{\geq 2}(\Lie_{\geq 2}(\log \mG_{\supp(\bl)}))) \\
    = \Lie_{\geq 2}(\Lie_{\geq 2}(\log \mG_{\supp(\bl)})) + \Lie_{\geq 2}(\Lie_{\geq 2}(\Lie_{\geq 2}(\log \mG_{\supp(\bl)}))).
    \end{multline*}
    Hence, the matrices $A''_i = {B'}_{i1}^{t'_i} \cdots {B'}_{im}^{t'_i}, i = 1, \ldots, K''$ satisfy the Equation~\eqref{eq:UTnstep2goal}.
    Define the new alphabet $\mG'' = \{A''_1, \ldots, A''_{K''}\}$.
    
    \item \textbf{Find matrices $A'''_1, \ldots, A'''_{K'''} \in \langle\mG''\rangle$ satisfying condition \eqref{eq:UTnstep3goal}.}
    Similar to the previous step, one can find a vector $\bl'' = (\ell''_1, \ldots, \ell''_{K''}) \in \Zp^{K''}$, satisfying $\supp(\bl'') = \{1, \ldots, K''\}$, $\ell''_i \geq 10, i = 1, \ldots, K''$, and 
    \[
    \sum_{i = 1}^{K''} \ell''_i \log A''_i \in \Lie_{\geq 2}(\Lie_{\geq 2}(\log \mG'_{\supp(\bl')})).
    \]
    Define $\log \mG''_{\supp(\bl'')} \coloneqq \{\log A''_i \mid i \in \supp(\bl'')\} = \log \mG''$.
    As in the previous step, we have
    \[
    \Lie_{\geq 2}(\log \mG''_{\supp(\bl'')}) = \Lie_{\geq 2}(\Lie_{\geq 2}(\log \mG'_{\supp(\bl')})).
    \]
    Combining it with $\Lie_{\geq 2}(\log \mG'_{\supp(\bl')}) = \Lie_{\geq 2}(\Lie_{\geq 2}(\log \mG_{\supp(\bl)}))$, we have
    \[
    \Lie_{\geq 2}(\log \mG''_{\supp(\bl'')}) = \Lie_{\geq 2}(\Lie_{\geq 2}(\Lie_{\geq 2}(\log \mG_{\supp(\bl)}))).
    \]
    Apply Proposition~\ref{prop:linspacen} for the alphabet $\mG''$ and the vector $\bl''$, then we have that, in the quotient space $\mun / \Lie_{\geq 2}(\Lie_{\geq 2}(\log \mG''_{\supp(\bl'')}))$, the equations $\overline{\mR_k(\bl'')} = \overline{\Lie_{\geq k}(\log \mG''_{\supp(\bl'')})}$, $k = 10, \ldots, 2$, hold.
    
    Then, applying Proposition~\ref{prop:cone} in the quotient linear space 
    \[
    V \coloneqq \mun/\Lie_{\geq 2}(\Lie_{\geq 2}(\log \mG''_{\supp(\bl'')}))
    \]
    as in the previous steps, we have
    \begin{multline*}
    \bigg\langle \overline{\log(B_1^t \cdots B_m^t)} \;\bigg|\; t \in \Zp,  m \geq 1, B_i \in \left(\mG'\right)^{*}, \sum_{i=1}^m \PI^{\mG''}(B_i) \in \{\bl'', 2 \bl''\} \bigg\rangle_{\Qp} \\
    = \overline{\Lie_{\geq 2}(\log \mG''_{\supp(\bl'')})}.
    \end{multline*}
    Hence, there exist $K''' > 0$ tuples of words $(B''_{11}, \ldots, B''_{1m})$, $\ldots$, $(B''_{K'''1}, \ldots, B''_{K'''m})$ in $\left(\mG''\right)^*$ with $\sum_{i=1}^m \PI^{\mG''}(B''_{ji}) = \bl''$ or $2 \bl''$ for all $j$, as well as positive integers $t''_1, \ldots, t''_{K'''} \in \Zpp$, such that
    \begin{multline}\label{eq:step3res}
    \langle \log({B''}_{i1}^{t''_i} \cdots {B''}_{im}^{t''_i}) \mid i = 1, \ldots, K''' \rangle_{\Qp} + \Lie_{\geq 2}(\Lie_{\geq 2}(\log \mG''_{\supp(\bl'')})) \\
    = \Lie_{\geq 2}(\log \mG''_{\supp(\bl'')}) + \Lie_{\geq 2}(\Lie_{\geq 2}(\log \mG''_{\supp(\bl'')})).
    \end{multline}
    Since $\Lie_{\geq 2}(\log \mG''_{\supp(\bl'')}) = \Lie_{\geq 2}(\Lie_{\geq 2}(\Lie_{\geq 2}(\log \mG_{\supp(\bl)})))$, we have
    \[
    \Lie_{\geq 2}(\Lie_{\geq 2}(\log \mG''_{\supp(\bl'')})) \subseteq \Lie_{\geq 16}(\log \mG_{\supp(\bl)}) = \{0\}.
    \]
    Thus, Equation~\eqref{eq:step3res} can be rewritten as
    \begin{equation*}
    \langle \log({B''}_{i1}^{t''_i} \cdots {B''}_{im}^{t''_i}) \mid i = 1, \ldots, K''' \rangle_{\Qp} 
    = \Lie_{\geq 2}(\Lie_{\geq 2}(\Lie_{\geq 2}(\log \mG_{\supp(\bl)}))).
    \end{equation*}
    Hence, the matrices $A'''_i = {B''}_{i1}^{t''_i} \cdots {B''}_{im}^{t''_i}, i = 1, \ldots, K'''$ satisfy the Equation~\eqref{eq:UTnstep3goal}.
    Define the new alphabet $\mG''' = \{A'''_1, \ldots, A'''_{K'''}\}$.
    
    \item \textbf{Find a word $w \in \langle\mG'''\rangle$ with $\log w = 0$.}
    Since the right hand side of Equation~\eqref{eq:UTnstep3goal} is a linear space, we have $- \log A'''_j \in \langle \log A'''_i \mid i = 1, \ldots, K''' \rangle_{\Qp}$ for $j = 1, \ldots, K'''$.
    Hence, there exists a non-zero vector $\bl''' = (\ell'''_1, \ldots, \ell'''_{K'}) \in \Zp^{K'''}$, satisfying
    \begin{equation*}
    \sum_{i = 1}^{K'''} \ell'''_i \log A'''_i = 0.
    \end{equation*}
    Since $\log \mG''' \in \Lie_{\geq 2}(\Lie_{\geq 2}(\Lie_{\geq 2}(\log \mG_{\supp(\bl)}))) \subseteq \Lie_{\geq 8}(\log \mG_{\supp(\bl)})$, we have 
    \[
    \Lie_{\geq 2}(\log \mG''') \subseteq \Lie_{\geq 16}(\log \mG_{\supp(\bl)}) = \{0\}.
    \]
    Hence, by the BCH formula,
    \[
    \log ({A'''}_1^{\ell'''_1} \cdots {A'''}_{K'''}^{\ell'''_{K'''}}) = \sum_{i = 1}^{K'''} \ell'''_i \log A'''_i = 0,
    \]
    because the terms $H_k, k \geq 2$ are in $\Lie_{\geq 2}(\log \mG''')$, which vanishes.
    Therefore, we have found the non-empty word $w = {A'''}_1^{\ell'''_1} \cdots {A'''}_{K'''}^{\ell'''_{K'''}} \in \left(\mG'''\right)^{*}$ satisfying $\log w = 0$.
    By replacing $A'''_i$ with their corresponding words ${B''}_{i1}^{t''_i} \cdots {B''}_{im}^{t''_i}$ in $\left(\mG''\right)^*$, then replacing $A''_i$ with corresponding words in $\left(\mG'\right)^*$, then replacing $A'_i$ with corresponding words in $\mG^*$, we see that $w$ considered as a word in $\mG^*$ has Parikh image in $\Zpp \cdot \bl$, because the words $B_{i1}^{t_i} \cdots B_{im}^{t_i}$ corresponding to $A'_i$ all have Parikh image in $\Zpp \cdot \bl$.
\end{enumerate}
\end{proof}

\section{Conjecture for higher nilpotency class}\label{sec:conj}
In Sections~\ref{sec:alg} and \ref{sec:techthm}, we showed that the invertible subset of any finite set $\mG \subseteq G$ is computable in polynomial time, where $G$ is a subgroup of $\UT(n, \Q)$ of nilpotency class at most ten.
The only obstacle for generalizing this result to higher nilpotency class is to prove Proposition~\ref{prop:k} for $k \geq 11$.
If the identities \eqref{eq:propk} exist for $k \geq 11$, then they can be found with the same computer aided procedure as the one used in the proof of Lemma~\ref{lem:H5}-\ref{lem:H7} (see Section~\ref{app:proofH}).
Following this idea, given $k \geq 11$, we propose the following conjecture, which generalizes Proposition~\ref{prop:k}:

\begin{conj}\label{conj:k}
Let $\mH \subset \UT(n, \Q)$ be any finite set of matrices.
There exist an integer $r \geq 0$, positive rational numbers $\alpha_1, \ldots, \alpha_r$, as well as, for $s = 1, \ldots, r$, words $\bj_s = j_{s,1} j_{s,2} \cdots j_{s,m_s}$ in the alphabet $\mI = \{1, 2, \ldots, k+1\}$, such that $\PI^{\mI}(\bj_s) \in \Zpp \cdot (1, 1, \ldots, 1)$ and
\begin{multline}\label{eq:conj}
    \sum_{\sigma \in \Sym_{k+1}} H_{k}(\log B_{\sigma(1)}, \ldots, \log B_{\sigma(k+1)})
    + \sum_{s = 1}^r \alpha_{s} \sum_{\sigma \in \Sym_{k+1}} H_k(\log B_{\sigma(j_{s,1})}, \ldots, \log B_{\sigma(j_{s,m_s})}) \\
    \in \Lie_{\geq k+1}(\log \mH) + \Lie_{\geq 2}(\Lie_{\geq 2}(\log \mH))
\end{multline}
for all matrices $B_1, \ldots, B_{k+1}$ in $\UT(n, \Q)$ satisfying $\log B_i \in \Lie_{\geq 1}(\log \mH)$
and $\sum_{i = 1}^{k+1} \log B_i \in \Lie_{\geq 2}(\log \mH)$.
\end{conj}

For even $k$, Conjecture~\ref{conj:k} is correct by the antisymmetry of $H_k$ (Lemma~\ref{lem:antisym}).
For $k = 3$, it is correct by taking $r = 0$ and using Lemma~\ref{lem:H3}.
For $k = 5, 7, 9$, it is verified by Lemma~\ref{lem:H5}, \ref{lem:H7}, \ref{lem:H9}, where the words $\bj_s, s = 1, \ldots, r,$ all satisfy $\PI^{\mI}(\bj_s) = (2, 2, \ldots, 2)$.

For odd $k$ larger than $10$, using Algorithm~\ref{alg:lincombrewrite} in Section~\ref{app:proofH}, we can search for words $\bj_s$ that potentially verify Conjecture~\ref{conj:k}.
Namely, starting with $q = 2$, take all the words $\bj_s$ satisfying $\PI^{\mI}(\bj_s) = (p, p, \ldots, p), 2 \leq p \leq q$.
Under the equivalence relation $\sim$ (defined in the proof of Lemma~\ref{lem:H5}), we can write each expression 
\[
h_k(\bj_s) \coloneqq \sum_{\sigma \in \Sym_{k+1}} H_k(\log B_{\sigma(j_{s,1})}, \ldots, \log B_{\sigma(j_{s,m_s})})
\]
as a linear combination of expressions $\tM(P, c)$ (see Section~\ref{app:proofH}) using Algorithm~\ref{alg:lincombrewrite}.
Then, writing $- \sum_{\sigma \in \Sym_{k+1}} H_{k}(\log B_{\sigma(1)}, \ldots, \log B_{\sigma(k+1)})$ also as a linear combination of $\tM(P, c)$, we can verify whether it is in the $\Qp$-cone generated by the elements $h_k(\bj_s)$.
If this is the case, then there exist positive rational numbers $\alpha_s, s = 1, 2, \ldots, $ satisfying Equation~\eqref{eq:conj}.
If this is not the case, we can increase $q$ and repeat the above procedure.

If there exists a relation of the form \eqref{eq:conj}, then the above procedure terminates for some $q$ and returns this relation. Otherwise it does not terminate.
In practice, it is more computationally viable to not take all the words $\bj_s$ satisfying $\PI^{\mI}(\bj_s) = (p, p, \ldots, p)$, but only a small amount of them chosen randomly.

Due to the restraint on computational power, we have only verified Conjecture~\ref{conj:k} for all $k \leq 10$, this is the reason why the main result of this paper stops at nilpotency class ten.
However, if we can verify Conjecture~\ref{conj:k} for larger $k$ (it suffices to verify for odd $k$), then we can extend the result of this paper to higher nilpotency class.
This is formalized by the following theorem.

\begin{thrm}\label{thm:conj}
Let $G$ be a subgroup of $\UT(n, \Q)$ whose nilpotency class is at most $d$.
If Conjecture~\ref{conj:k} holds for all $k \leq d$, then Algorithm~\ref{alg:invn} correctly computes the invertible subset of any finite set $\mG \subseteq G$ in polynomial time.
\end{thrm}
\begin{proof}
For any $\bl \in \Zp^K$,
similar to Equation~\eqref{eq:defR1}, define recursively the cones
\begin{multline*}
    \mR_{d+1}(\bl) \coloneqq \{0\}, \\
    \mR_k(\bl) \coloneqq \mR_{k+1}(\bl) +
    \bigg\langle H_k(\log B_1, \ldots, \log B_m) \;\bigg|\; m \geq 1, B_i \in \mG^*, \sum_{i = 1}^m \PI^{\mG}(B_i) \in \Zpp \cdot \bl \bigg\rangle_{\Qp}, \\
    k = d, d-1, \ldots, 3, 2,
\end{multline*}
and the set 
\[
\log \mG_{\supp(\bl)} \coloneqq \{\log A_i \mid A_i \in \mG, i \in \supp(\bl)\}.
\]
Suppose $\bl$ satisfies $\ell_i \geq d, i \in \supp(\bl)$.
Consider the quotient space $\mun / \Lie_{\geq 2}(\Lie_{\geq 2}(\log \mG_{\supp(\bl)}))$.
Following the pattern in the proof of Proposition~\ref{prop:linspacen}, we can show that for all $k \leq d+1$, the cone $\overline{\mR_k(\bl)}$
is equal to the linear space
$\overline{\Lie_{\geq k}(\log \mG_{\supp(\bl)})}$.

Then, using the same arguments as Theorem~\ref{thm:equivn}, we can show the following generalization of Theorem~\ref{thm:equivn}:
\begin{enumerate}[nolistsep, label=(\roman*)]
    \item If there exists a word $w \in \mG^*$ with $\PI^{\mG}(w) = \bl$ and $\log w = 0$, then
    \begin{equation}\label{eq:condconj}
        \sum_{i = 1}^K \ell_i \log A_i \in \Lie_{\geq 2}(\log \mG_{\supp(\bl)}).
    \end{equation}
    \item If $\bl$ satisfies \eqref{eq:condconj}, then there exists a non-empty word $w \in \mG^*$, with $\PI^{\mG}(w) \in \Zpp \cdot \bl$, such that $\log w = 0$.
\end{enumerate}
From here, the proof of correctness of Algorithm~\ref{alg:invn} and its complexity analysis is identical to the proof of Theorem~\ref{thm:invn}, replacing the property $[\log \mG]_{11} = \{0\}$ by $[\log \mG]_{d+1} = \{0\}$.
\end{proof}
A natural question is whether our result can be extended to arbitrary nilpotency class $d$.
This can either be done by proving Conjecture~\ref{conj:k} for higher $k$ or by finding another way to approach this problem.
In particular, similar to Corllary~\ref{cor:decnilp}, this would yield the decidability for the Identity Problem and the Group Problem for arbitrary finitely generated nilpotent groups of class at most $d$.

\bibliography{UTFive}

\begin{thebibliography}{10}

\bibitem{babai1985trading}
L.~Babai.
\newblock Trading group theory for randomness.
\newblock In {\em Proceedings of the seventeenth annual ACM symposium on Theory
  of computing}, pages 421--429, 1985.

\bibitem{babai1996multiplicative}
L.~Babai, R.~Beals, J.-y. Cai, G.~Ivanyos, and E.~M. Luks.
\newblock Multiplicative equations over commuting matrices.
\newblock In {\em Proceedings of the Seventh Annual ACM-SIAM Symposium on
  Discrete Algorithms}, pages 498--507, 1996.

\bibitem{babai2011code}
L.~Babai, P.~Codenotti, J.~A. Grochow, and Y.~Qiao.
\newblock Code equivalence and group isomorphism.
\newblock In {\em Proceedings of the twenty-second annual ACM-SIAM symposium on
  Discrete Algorithms}, pages 1395--1408. SIAM, 2011.

\bibitem{baker1905alternants}
H.~F. Baker.
\newblock Alternants and continuous groups.
\newblock {\em Proceedings of the London Mathematical Society}, 2(1):24--47,
  1905.

\bibitem{Baumslag2007LectureNO}
G.~Baumslag.
\newblock {\em Lecture notes on nilpotent groups}.
\newblock American Mathematical Society, 2007.

\bibitem{beals1993vegas}
R.~Beals and L.~Babai.
\newblock Las vegas algorithms for matrix groups.
\newblock In {\em Proceedings of 1993 IEEE 34th Annual Foundations of Computer
  Science}, pages 427--436. IEEE, 1993.

\bibitem{bell2017identity}
P.~C. Bell, M.~Hirvensalo, and I.~Potapov.
\newblock The {I}dentity {P}roblem for {M}atrix {S}emigroups in
  $\operatorname{SL}_2(\mathbb{Z})$ is {NP}-complete.
\newblock In {\em Proceedings of the Twenty-Eighth Annual ACM-SIAM Symposium on
  Discrete Algorithms}, pages 187--206. SIAM, 2017.

\bibitem{bell2010undecidability}
P.~C. Bell and I.~Potapov.
\newblock On the undecidability of the identity correspondence problem and its
  applications for word and matrix semigroups.
\newblock {\em International Journal of Foundations of Computer Science},
  21(06):963--978, 2010.

\bibitem{blondel2005decidable}
V.~D. Blondel, E.~Jeandel, P.~Koiran, and N.~Portier.
\newblock Decidable and undecidable problems about quantum automata.
\newblock {\em SIAM Journal on Computing}, 34(6):1464--1473, 2005.

\bibitem{bokut1963basis}
L.~A. Bokut'.
\newblock A basis for free polynilpotent lie algebras.
\newblock {\em Algebra i logika}, 2(4):13--19, 1963.

\bibitem{campbell1897law}
J.~E. Campbell.
\newblock On a law of combination of operators (second paper).
\newblock {\em Proceedings of the London Mathematical Society}, 1(1):14--32,
  1897.

\bibitem{casas2009efficient}
F.~Casas and A.~Murua.
\newblock An efficient algorithm for computing the
  {B}aker-{C}ampbell-{H}ausdorff series and some of its applications.
\newblock {\em Journal of Mathematical Physics}, 50(3):033513, 2009.

\bibitem{choffrut2005some}
C.~Choffrut and J.~Karhum{\"a}ki.
\newblock Some decision problems on integer matrices.
\newblock {\em RAIRO-Theoretical Informatics and Applications-Informatique
  Th{\'e}orique et Applications}, 39(1):125--131, 2005.

\bibitem{colcombet2019reachability}
T.~Colcombet, J.~Ouaknine, P.~Semukhin, and J.~Worrell.
\newblock On reachability problems for low-dimensional matrix semigroups.
\newblock In C.~Baier, I.~Chatzigiannakis, P.~Flocchini, and S.~Leonardi,
  editors, {\em 46th International Colloquium on Automata, Languages, and
  Programming, {ICALP} 2019, July 9-12, 2019, Patras, Greece}, volume 132 of
  {\em LIPIcs}, pages 44:1--44:15. Schloss Dagstuhl - Leibniz-Zentrum f{\"{u}}r
  Informatik, 2019.

\bibitem{DEGRAAF200231}
W.~A. {de Graaf} and W.~Nickel.
\newblock Constructing faithful representations of finitely-generated
  torsion-free nilpotent groups.
\newblock {\em Journal of Symbolic Computation}, 33(1):31--41, 2002.

\bibitem{derksen2005quantum}
H.~Derksen, E.~Jeandel, and P.~Koiran.
\newblock Quantum automata and algebraic groups.
\newblock {\em Journal of Symbolic Computation}, 39(3-4):357--371, 2005.

\bibitem{dong2022identity}
R.~Dong.
\newblock On the identity problem for unitriangular matrices of dimension four.
\newblock In S.~Szeider, R.~Ganian, and A.~Silva, editors, {\em 47th
  International Symposium on Mathematical Foundations of Computer Science,
  {MFCS} 2022, August 22-26, 2022, Vienna, Austria}, volume 241 of {\em
  LIPIcs}, pages 43:1--43:14. Schloss Dagstuhl - Leibniz-Zentrum f{\"{u}}r
  Informatik, 2022.

\bibitem{dynkin2000calculation}
E.~B. Dynkin.
\newblock Calculation of the coefficients in the {C}ampbell--{H}ausdorff
  formula.
\newblock {\em Selected Papers of E. B. Dynkin with Commentary. Ed. by
  Yushkenich, A. A.}, pages 31--35, 2000.

\bibitem{erdmann2006introduction}
K.~Erdmann and M.~J. Wildon.
\newblock {\em Introduction to Lie algebras}, volume 122.
\newblock Springer, 2006.

\bibitem{garzon1991isomorphism}
M.~Garzon and Y.~Zalcstein.
\newblock On isomorphism testing of a class of 2-nilpotent groups.
\newblock {\em Journal of Computer and System Sciences}, 42(2):237--248, 1991.

\bibitem{gong1998classification}
M.-P. Gong.
\newblock {\em Classification of Nilpotent Lie Algebras of Dimension 7 (over
  Algebraically Closed Fields and $\mathbb{R}$)}.
\newblock PhD thesis, University of Waterloo, 1998.

\bibitem{grunewald1980some}
F.~Grunewald and D.~Segal.
\newblock Some general algorithms. {II}: Nilpotent groups.
\newblock {\em Annals of Mathematics}, 112(3):585--617, 1980.

\bibitem{hausdorff1906symbolische}
F.~Hausdorff.
\newblock Die symbolische {E}xponentialformel in der {G}ruppentheorie.
\newblock {\em Berichte \"{u}ber die Verhandlungen der
  K\"{o}niglich-S\"{a}chsischen Gesellschaft der Wissenschaften zu Leipzig,
  Mathematisch-Physische Klasse}, 58:19--48, 1906.

\bibitem{holt2005handbook}
D.~F. Holt, B.~Eick, and E.~A. O'Brien.
\newblock {\em Handbook of Computational Group Theory}.
\newblock Chapman and Hall/CRC, 2005.

\bibitem{hrushovski2018polynomial}
E.~Hrushovski, J.~Ouaknine, A.~Pouly, and J.~Worrell.
\newblock Polynomial invariants for affine programs.
\newblock In {\em Proceedings of the 33rd Annual ACM/IEEE Symposium on Logic in
  Computer Science}, pages 530--539, 2018.

\bibitem{kargapolov1979fundamentals}
M.~I. Kargapolov and J.~I. Merzljakov.
\newblock {\em Fundamentals of the Theory of Groups}, volume~62.
\newblock Springer, 1979.

\bibitem{khukhro1998p}
E.~I. Khukhro.
\newblock {\em p-Automorphisms of Finite p-Groups}, volume 246.
\newblock Cambridge University Press, 1998.

\bibitem{ko2017identity}
S.~Ko, R.~Niskanen, and I.~Potapov.
\newblock On the identity problem for the special linear group and the
  {H}eisenberg group.
\newblock In I.~Chatzigiannakis, C.~Kaklamanis, D.~Marx, and D.~Sannella,
  editors, {\em 45th International Colloquium on Automata, Languages, and
  Programming, {ICALP} 2018, July 9-13, 2018, Prague, Czech Republic}, volume
  107 of {\em LIPIcs}, pages 132:1--132:15. Schloss Dagstuhl - Leibniz-Zentrum
  f{\"{u}}r Informatik, 2018.

\bibitem{kopytov1968solvability}
V.~M. Kopytov.
\newblock Solvability of the problem of occurrence in finitely generated
  soluble groups of matrices over the field of algebraic numbers.
\newblock {\em Algebra and Logic}, 7(6):388--393, 1968.

\bibitem{lo1999practical}
E.~H. Lo and G.~Ostheimer.
\newblock A practical algorithm for finding matrix representations for
  polycyclic groups.
\newblock {\em Journal of Symbolic Computation}, 28(3):339--360, 1999.

\bibitem{loday1992serie}
J.-L. Loday.
\newblock S{\'e}rie de {H}ausdorff, idempotents {E}ul{\'e}riens et algebres de
  {H}opf.
\newblock {\em Expositiones Mathematicae}, 12, 1994.

\bibitem{lohrey2021subgroup}
M.~Lohrey.
\newblock Subgroup membership in {GL}(2, {Z}).
\newblock In {\em 38th International Symposium on Theoretical Aspects of
  Computer Science (STACS 2021)}. Schloss Dagstuhl-Leibniz-Zentrum f{\"u}r
  Informatik, 2021.

\bibitem{macdonald2019low}
J.~Macdonald, A.~Miasnikov, and D.~Ovchinnikov.
\newblock Low-complexity computations for nilpotent subgroup problems.
\newblock {\em International Journal of Algebra and Computation},
  29(04):639--661, 2019.

\bibitem{mal1951some}
A.~Mal'cev.
\newblock On some classes of infinite soluble groups.
\newblock {\em Mat. Sb. (N.S.)}, 28(70):567--588, 1951.

\bibitem{markov1947certain}
A.~Markov.
\newblock On certain insoluble problems concerning matrices.
\newblock {\em Doklady Akad. Nauk SSSR}, 57(6):539--542, 1947.

\bibitem{mikhailova1966occurrence}
K.~A. Mikhailova.
\newblock The occurrence problem for direct products of groups.
\newblock {\em Matematicheskii Sbornik}, 112(2):241--251, 1966.

\bibitem{nickel1994computing}
W.~Nickel.
\newblock Computing nilpotent quotients of finitely presented groups.
\newblock {\em Geometric and computational perspectives on infinite groups},
  25:175--191, 1994.

\bibitem{roman2022undecidability}
V.~Roman'kov.
\newblock Undecidability of the submonoid membership problem for a sufficiently
  large finite direct power of the heisenberg group.
\newblock {\em arXiv preprint arXiv:2209.14786}, 2022.

\bibitem{sun2023faster}
X.~Sun.
\newblock Faster isomorphism for $p$-groups of class 2 and exponent $p$.
\newblock In {\em Proceedings of the 55th Annual ACM Symposium on Theory of
  Computing}, pages 433--440, 2023.

\bibitem{sagemath}
{The Sage Developers}.
\newblock {\em {S}ageMath, the {S}age {M}athematics {S}oftware {S}ystem
  ({V}ersion 9.0)}, 2020.
\newblock {\tt https://www.sagemath.org}.

\bibitem{vera2003conjugacy}
A.~Vera-L{\'o}pez and J.~M. Arregi.
\newblock Conjugacy classes in unitriangular matrices.
\newblock {\em Linear algebra and its applications}, 370:85--124, 2003.

\end{thebibliography}

\begin{appendices}

\section{Omitted proofs from Sections~\ref{sec:intro}-\ref{sec:prelim}}\label{app:proofCor}

In this section of the appendix we give the proofs of several semigroup and group theory results omitted in the main paper.

\propinvtoid*
\begin{proof}
For a word $w$ over the alphabet $\mG$, define $\pi(w)$ to be the matrix obtained by multiplying consecutively the matrices appearing in $w$.

(i) If the Identity Problem has a positive answer, let $w$ be a non-empty word over the alphabet $\mG$ such that $\pi(w) = I$.
Write $w = A_i w'$, ($w'$ could be the empty word), then $A_i^{-1} = \pi(w')$. If $A_i = I$ then obviously $A_i^{-1} = A_i \in \sgmG$. If $A_i \neq I$ then $\pi(w') \neq I$ so $w'$ is not the empty word and $\pi(w') \in \sgmG$. Therefore $A_i^{-1} \in \sgmG$.
Conversely, if $A_i \in \mG_{inv}$, then either $A_i = I$ in which case $I = A_i \in \sgmG$, or $A_i^{-1} = \pi(w')$ for some non-empty word $w'$, so $I = \pi(A_i w') \in \sgmG$.

(ii)
Since every element in $\mG_{inv}$ is invertible in $\sgmG$, the semigroup $\langle \mG_{inv} \rangle$ it generates also only contains invertible elements.
Therefore, if $\mG = \mG_{inv}$ then $\sgmG$ is a group.
On the other hand, if $\sgmG$ is a group then every element of $\mG$ is invertible in $\sgmG$, so $\mG = \mG_{inv}$.
\end{proof}

{\renewcommand\footnote[1]{}\decnilp*}

\begin{proof}
A consistent polycyclic presentation~\cite[Chapter.~8]{holt2005handbook} of $G$ can be computed from a finite presentation of $G$ \cite{nickel1994computing}, so we can suppose that a consistent polycyclic presentation of $G$ is given.
Let $G$ be a finitely generated nilpotent group of class at most ten.
The set of torsion elements in $G$ forms a normal subgroup $T$ of $G$.
A set of generators of $T$ along with a presentation can be effectively computed by \cite[Theorem~8]{macdonald2019low}.
Then, by \cite[Lemma~8.38]{holt2005handbook}, a consistent polycyclic presentation for the torsion-free nilpotent group $G/T$ can be computed.
Note that $G/T$ is still nilpotent and its nilpotency class does not exceed that of $G$.
An embedding of $G/T$ as a subgroup of $\UT(n, \Q)$ for some $n$ can then be effectively computed (\cite{lo1999practical}, \cite{DEGRAAF200231}).
Using this embedding, the Identity Problem and the Group Problem can be decided in the quotient group $G/T$ by Theorem~\ref{thm:invn} and Proposition~\ref{prop:invtoid}(i), (iii).

By \cite[Theorem~2.1]{Baumslag2007LectureNO}, $G$ can be embedded (injectively) as a subgroup of a direct product $A \times G/T$, where $A$ is a finite group.
Let $\phi: G \hookrightarrow A \times G/T$ denote this embedding, and let $p: A \times G/T \rightarrow G/T$ be the natural projection.

By the injectivity of $\phi$, the Identity Problem has a positive answer for $\mG \subseteq G$ if and only if it has a positive answer for $\phi(\mG) \subseteq A \times G/T$.
We claim that the Identity Problem has a positive answer for $\phi(\mG) \subseteq A \times G/T$ if and only if it has a positive answer for $p(\phi(\mG)) \subseteq G/T$.
Indeed, for any group $H$, denote $e_H$ its natural element.
If $e_{A \times G/T} \in \langle \phi(\mG) \rangle$, then obviously $e_{G/T} \in \langle p(\phi(\mG)) \rangle$ because $p$ is a semigroup homomorphism.
If $e_{G/T} \in \langle p(\phi(\mG)) \rangle$, then there exists $a \in A$ such that $(a, e_{G/T}) \in \langle \phi(\mG) \rangle$.
Because $A$ is finite, there exists a positive integer $k$, such that $a^k = e_A$ for all $a \in A$.
Then $e_{A \times G/T} = (e_A, e_{G/T}) = (a^k, e_{G/T}^k) \in \langle \phi(\mG) \rangle$.
This proves the claim.
Therefore, the Identity Problem for $\mG \subseteq G$ is equivalent to the Identity Problem for $p(\phi(\mG)) \subseteq G/T$, which is decidable by the first part of the proof.

The injectivity of $\phi$ also shows that the Group Problem has a positive answer for $\mG \subseteq G$ if and only if it has a positive answer for $\phi(\mG) \subseteq A \times G/T$.
We claim that the Group Problem has a positive answer for $\phi(\mG) \subseteq A \times G/T$ if and only if it has a positive answer for $p(\phi(\mG)) \subseteq G/T$.
Indeed, suppose $\langle \phi(\mG) \rangle$ is a group, then there exists a non-empty word $w \in \phi(\mG)^+$ where every letter of $\phi(\mG)$ appears at least once, and whose product is equal to the neutral element $e_{A \times G/T}$.
This is because every letter $B \in \phi(\mG)$ has a inverse in $\langle \phi(\mG) \rangle$, hence multiplying $B$ with the word representing its inverse yields a word $w_B$ whose product is the neutral element, and where the letter $B$ appears.
Concatenating the words $w_B$ for all $B \in \phi(\mG)$ yields the word $w$.
Next, projecting each letter in $w$ with $p$ yields a non-empty word $p(w) \in p(\phi(\mG))^+$ where every letter of $p(\phi(\mG))$ appears at least once, and whose product is equal to the neutral element $e_{G/T}$. 
Thus every element in $p(\phi(\mG))$ is invertible in $\langle p(\phi(\mG)) \rangle$.
This show that if $\langle \phi(\mG) \rangle$ is a group then $\langle p(\phi(\mG)) \rangle$ is a group.
For the opposite implication, suppose $\langle p(\phi(\mG)) \rangle$ is a group, then there exists a non-empty word $w \in \phi(\mG)^+$ where every letter of $\phi(\mG)$ appears at least once, and whose product is equal to some element $(a, e_{G/T}) \in A \times e_{G/T}$.
Because $A$ is finite, there exists a positive integer $k$, such that $a^k = e_A$.
Hence the product of the word $w^k \in \phi(\mG)^+$ is equal to $(a^k, e_{G/T}^k) = e_{A \times G/T}$.
As every letter of $\phi(\mG)$ appears in $w^k$ at least once, every element in $\phi(\mG)$ is invertible in $\langle \phi(\mG) \rangle$.
Thus $\langle \phi(\mG) \rangle$ is a group.
Therefore, the Group Problem for $\mG \subseteq G$ is equivalent to the Group Problem for $p(\phi(\mG)) \subseteq G/T$, which is decidable by the first part of the proof.
\end{proof}

\lemcompsupp*
\begin{proof}
For $i = 1, \ldots, K$, we can check whether $i \in \supp(\Lambda)$ in the following way.
By definition, $i \in \supp(\Lambda)$ if and only if the system 
\begin{equation}\label{eq:LPZ}
(\ell_1, \ldots, \ell_K) \in V, \ell_1 \geq 0, \ldots, \ell_i > 0, \ldots, \ell_K \geq 0
\end{equation}
has an \emph{integer} solution $(\ell_1, \ldots, \ell_K) \in \Z^K$.
By the homogeneity of the system~\eqref{eq:LPZ}, it has an \emph{integer} solution if and only if it has a \emph{rational} solution.
The existence of a rational solution to system~\eqref{eq:LPZ} can be decided by linear programming in polynomial time.
Therefore, the support of $\Lambda$ can be computed in polynomial time by checking whether $i \in \supp(\Lambda)$ for every $i = 1, \ldots, K$. 
\end{proof}

\nilclass*
\begin{proof}
For an element $g \in G$ and a rational number $q \in \Q$, define $g^q \coloneqq \exp(q \log g)$.
A group $G \leq \UT(n, \Q)$ is called \emph{$\Q$-powered} if for every element $g \in G$ and $q \in \Q$, we have $g^q \in G$.
A unitriangular matrix group over $\Q$ is torsion-free, because $A^n = I \iff n \log A = 0 \iff \log A = 0 \iff A = I$.
Therefore, by \cite[Theorem~9.20(a)]{khukhro1998p}, $G$ can be embedded in a $\Q$-powered group $\hat{G}$ of the same nilpotency class $d$.\footnote{One can take the group $\hat{G}$ to be \emph{Mal'cev completion} of $G$.}
By \cite[Theorem~10.3(d)]{khukhro1998p}, $\log \hat{G}$ is a Lie algebra over $\Q$, and $\log \hat{G}$ is of nilpotency class $d$ (meaning $[\log \hat{G}]_{d+1} = \{0\}$).
Therefore, $[\log G]_{d+1} \subseteq [\log \hat{G}]_{d+1} = \{0\}$.
\end{proof}

\section{Computer-aided proof of Lemma~\ref{lem:H5}-\ref{lem:H9}}\label{app:proofH}
In this section we give the detailed account for the proof of Lemma~\ref{lem:H5}-\ref{lem:H9} using computer assistance.

We fix an integer $k$ for the whole section.
Let $\mH$ be a subset of $\UT(n, \Q)$.
For $x, y \in \mun$, we write
\[
x \overset{\Lie_{\geq 2}(\Lie_{\geq 2}(\log \mH))}{\sim} y
\]
if $x - y \in \Lie_{\geq 2}(\Lie_{\geq 2}(\log \mH))$, and
\[
x \overset{\Lie_{\geq k+1}(\log \mH)}{\sim} y
\]
if $x - y \in \Lie_{\geq k+1}(\log \mH)$.
Obviously, $\overset{\Lie_{\geq 2}(\Lie_{\geq 2}(\log \mH))}{\sim}$ and $\overset{\Lie_{\geq k+1}(\log \mH)}{\sim}$ are equivalence relations and we denote by $\sim$ the transitive closure of these two relations.

The following lemma shows the effect of the relation $\overset{\Lie_{\geq 2}(\Lie_{\geq 2}(\log \mH))}{\sim}$.
In fact, the quotient Lie algebra $L \coloneqq \Lie_{\geq 1}(\log \mH) / \Lie_{\geq 2}(\Lie_{\geq 2}(\log \mH))$ is \emph{metabelian}, meaning $[[L,L],[L,L]] = 0$.
This property allows us the permute elements in iterated Lie brackets:
\begin{lem}\label{lem:effectmod22}
For $C_1, \ldots, C_k \in \Lie_{\geq 1}(\log \mH)$ and $i = 3, \ldots, k-1$, we have
\[
    [\ldots [[\ldots [C_1, C_2], \ldots, C_i], C_{i+1}], \ldots, C_k] \overset{\Lie_{\geq 2}(\Lie_{\geq 2}(\log \mH))}{\sim}
    [\ldots [[\ldots [C_1, C_2], \ldots, C_{i+1}], C_i], \ldots, C_k].
\]
\end{lem}
\begin{proof}
For $i = 3, \ldots, k-1$, by the Jacobi identity,
\begin{align}\label{eq:LjLL}
    & [\ldots [[\ldots [C_1, C_2], \ldots, C_i], C_{i+1}], \ldots, C_k] - 
    [\ldots [[\ldots [C_1, C_2], \ldots, C_{i+1}], C_i], \ldots, C_k] \nonumber\\
    = & \; [\ldots [[\ldots [C_1, C_2], \ldots, C_{i-1}], [C_i, C_{i+1}]], \ldots, C_k] \nonumber \\
    \in & \; [\ldots[\Lie_{\geq 2}(\log \mH), \Lie_{\geq 2}(\log \mH)], \ldots, C_k]. \nonumber \\
    \subseteq & \; [\ldots[\Lie_{\geq 2}(\Lie_{\geq 2}(\log \mH)), C_{i+2}], \ldots, C_k].
\end{align}

We then show that 
\begin{equation}\label{eq:impL22}
    X \in \Lie_{\geq 2}(\Lie_{\geq 2}(\log \mH)), Y \in \Lie_{\geq 1}(\log \mH) \implies [X, Y] \in \Lie_{\geq 2}(\Lie_{\geq 2}(\log \mH)).
\end{equation}
Since $X$ is in $\Lie_{\geq 2}(\Lie_{\geq 2}(\log \mH))$, it can be written as a linear combination of elements of the form $[\ldots[X_1, X_2], \ldots, X_s]$ where $s \geq 2$, $X_i \in \Lie_{\geq 2}(\log \mH), i = 1, \ldots, s$.
Therefore it suffices to show the implication~\eqref{eq:impL22} for the case $X = [\ldots[X_1, X_2], \ldots, X_s]$ where $X_i \in \Lie_{\geq 2}(\log \mH), i = 1, \ldots, s$.
Let 
\[
X' \coloneqq [\ldots[X_1, X_2], \ldots, X_{s-1}] \in \Lie_{\geq 2(s-1)}(\log \mH) \subseteq \Lie_{\geq 2}(\log \mH),
\]
so $X = [X', X_s]$ with $X', X_s \in \Lie_{\geq 2}(\log \mH)$.
Then by the Jacobi identity, 
\[
[X, Y] = [[X', X_s], Y] = - [[X_s, Y], X'] - [[Y, X'], X_s],
\]
where 
\begin{multline*}
    [[X_s, Y], X'] \in [[\Lie_{\geq 2}(\log \mH), \Lie_{\geq 1}(\log \mH)], \Lie_{\geq 2}(\log \mH)] \\
    \subseteq [\Lie_{\geq 2}(\log \mH), \Lie_{\geq 2}(\log \mH)] \subseteq \Lie_{\geq 2}(\Lie_{\geq 2}(\log \mH))
\end{multline*}
and 
\begin{multline*}
    [[Y, X'], X_s] \in [[\Lie_{\geq 1}(\log \mH), \Lie_{\geq 2}(\log \mH)], \Lie_{\geq 2}(\log \mH)] \\
    \subseteq [\Lie_{\geq 2}(\log \mH), \Lie_{\geq 2}(\log \mH)] \subseteq \Lie_{\geq 2}(\Lie_{\geq 2}(\log \mH)).
\end{multline*}
Therefore $[X, Y] \in \Lie_{\geq 2}(\Lie_{\geq 2}(\log \mH))$, showing the implication~\eqref{eq:impL22}.

Applying this implication with $Y = C_{i+2}, C_{i+3}, \ldots, C_k$ in Equation~\eqref{eq:LjLL} shows
\begin{align*}
    & [\ldots[\Lie_{\geq 2}(\Lie_{\geq 2}(\log \mH)), C_{i+2}], \ldots, C_k] \\
    \subseteq & \; [\ldots[\Lie_{\geq 2}(\Lie_{\geq 2}(\log \mH)), C_{i+3}], \ldots, C_k] \\
    & \vdots \\
    \subseteq & \; [\Lie_{\geq 2}(\Lie_{\geq 2}(\log \mH)), C_{k}] \\
    \subseteq & \; \Lie_{\geq 2}(\Lie_{\geq 2}(\log \mH))
\end{align*}

Hence Equation~\eqref{eq:LjLL} yields
\[
    [\ldots [[\ldots [C_1, C_2], \ldots, C_i], C_{i+1}], \ldots, C_k] \overset{\Lie_{\geq 2}(\Lie_{\geq 2}(\log \mH))}{\sim}
    [\ldots [[\ldots [C_1, C_2], \ldots, C_{i+1}], C_i], \ldots, C_k].
\]
\end{proof}




Fix an integer $k$.
Define an \emph{integer partition} $P$ (of $k$) to be a series of numbers $(a_1, \ldots, a_s)$ such that $a_1 \geq a_2 \geq \cdots \geq a_s \geq 1$ and $k = a_1 + \cdots + a_s$.
Define $\max(P) \coloneqq a_1, \min(P) \coloneqq a_s$ and $\set(P) \coloneqq \{t \mid \exists a_i = t\}$.
Define a \emph{set partition} $S$ (of $\{1, \ldots, k\}$) to be a set of non-empty disjoint sets $S = \{A_1, \ldots, A_s\}$ such that $A_1 \cup \cdots \cup A_s = \{1, \ldots, k\}$.
For any $k$-tuple $\bj = (j_1, \ldots, j_k) \in \{1, \ldots, k+1\}^k$, define the \emph{associated set partition} of $\bj$ the set partition consisting of sets of indices of its distinct elements
\[
\SP(\bj) \coloneqq \bigg\{A_i \coloneqq \{l \mid j_l = i\} \;\bigg|\; i = 1, \ldots, k+1, A_i \neq \emptyset \bigg\}.
\]
For example, if $k = 6$, $\bj = (4, 2, 7, 2, 2, 4)$, then $\SP(\bj) = \{ \{1, 6\}, \{2, 4, 5\}, \{3\} \}$.

Define the \emph{associated integer partition} $\IP(S)$ of a set partition $S$ to be the series of set cardinalities in $S$ in decreasing order.
For example, if $k = 6$, $S = \{ \{1, 6\}, \{2, 4, 5\}, \{3\} \}$, then $\IP(S) = (3, 2, 1)$.
In particular, in this example we have $\max(\IP(S)) = 3, \min(\IP(S)) = 1$ and $\set(P) = \{3, 2, 1\}$.

We now fix elements $C_1, \ldots, C_{k} \in \Lie_{\geq 1}(\log \mH)$.
For a given tuple $\bj = (j_1, \ldots, j_k) \in \{1, \ldots, k+1\}^k$, define the symmetric sums
\begin{align*}
\Phi(\bj) & \coloneqq \frac{1}{\left(k+1 - \card(\SP(\bj))\right)!} \sum_{\sigma \in \Sym_{k+1}} \varphi_k(C_{\sigma(j_1)}, C_{\sigma(j_2)}, \ldots, C_{\sigma(j_{k})}), \\
M(\bj) & \coloneqq \frac{1}{\left(k+1 - \card(\SP(\bj))\right)!} \sum_{\sigma \in \Sym_{k+1}}[\ldots[C_{\sigma(j_1)}, C_{\sigma(j_2)}], \ldots, C_{\sigma(j_{k})}].
\end{align*}
Here, $\varphi_k$ is the expression defined in the Dynkin formula~\eqref{eq:Dynkin2}.
The relation between $\Phi(\bj)$ and $M(\bj)$ can be computed as follows.
\begin{align}\label{eq:PhitoM}
    \Phi(\bj) = & \frac{1}{\left(k+1 - \card(\SP(\bj))\right)!} \sum_{\sigma \in \Sym_{k+1}} \varphi_k(C_{\sigma(j_1)}, C_{\sigma(j_2)}, \ldots, C_{\sigma(j_{k})}) \nonumber \\
    = & \frac{1}{\left(k+1 - \card(\SP(\bj))\right)!} \sum_{\sigma \in \Sym_{k+1}} \sum_{\tau \in \Sym_k} \frac{(-1)^{d(\tau)}}{k^2 \binom{k-1}{d(\tau)}} [\ldots[C_{\sigma(j_{\tau(1)})}, C_{\sigma(j_{\tau(2)})}], \ldots, C_{\sigma(j_{\tau(k)})}] \nonumber \\
    = &  \sum_{\tau \in \Sym_k} \frac{(-1)^{d(\tau)}}{k^2 \binom{k-1}{d(\tau)}} \cdot \frac{1}{\left(k+1 - \card(\SP(\bj))\right)!} \sum_{\sigma \in \Sym_{k+1}} [\ldots[C_{\sigma(j_{\tau(1)})}, C_{\sigma(j_{\tau(2)})}], \ldots, C_{\sigma(j_{\tau(k)})}] \nonumber \\
    = & \sum_{\tau \in \Sym_k} \frac{(-1)^{d(\tau)}}{k^2 \binom{k-1}{d(\tau)}} \cdot M(\bj_{\tau}),
\end{align}
where $\bj_{\tau} \coloneqq (j_{\tau(1)}, j_{\tau(2)}, \ldots, j_{\tau(k)})$.

From the definition of $\Phi(\bj)$ and $M(\bj)$ it follows that that for any $\sigma \in \Sym_{k+1}$, writing $\sigma(\bj) \coloneqq (\sigma(j_1), \ldots, \sigma(j_k))$, we have $\Phi(\sigma(\bj)) = \Phi(\bj)$ and $M(\sigma(\bj)) = M(\bj)$.
By this symmetry, $\Phi(\bj)$ and $M(\bj)$ only depend on their associated set partition $\SP(\bj)$. 
Hence for any set partition $S$, we can define
\[
\Phi(S) \coloneqq \Phi(\bj), \quad M(S) \coloneqq M(\bj), \quad \text{ where } \SP(\bj) = S.
\]
From Equation~\eqref{eq:PhitoM} we get
\begin{equation}\label{eq:PhitoMset}
    \Phi(S) = \sum_{\tau \in \Sym_k} \frac{(-1)^{d(\tau)}}{k^2 \binom{k-1}{d(\tau)}} \cdot M(S_{\tau}),
\end{equation}
where $S_{\tau}$ is the set partition obtained by replacing $i$ by $\tau(i)$ in all sets of $S$ for all $i = 1, \ldots, k$:
\[
S_{\tau} \coloneqq \bigg\{\{\tau(j) \mid j \in A\} \;\bigg|\; A \in S \bigg\}.
\]

For two set partitions $S_1$ and $S_2$, $S_2$ is called a \emph{coarsening} of $S_1$ if for every $A \in S_1$, there exists $A' \in S_2$ such that $A \subseteq A'$.
For example, $\{ \{1, 3, 4\}, \{2, 5, 6\} \}$ is a coarsening of $\{ \{1, 3, 4\}, \{2\}, \{5, 6\} \}$.
In particular, any set partition is a coarsening of itself.
Denote by $S_2 \succcurlyeq S_1$ if $S_2$ is a coarsening of $S_1$.

The next lemma shows the effect of the relation $\overset{\Lie_{\geq k+1}(\log \mH)}{\sim}$ for sums over coarsenings.

\begin{lem}\label{lem:effectmodk1}
Let $\mH$ be a subset of $\UT(n, \Q)$.
Suppose $C_1, \ldots, C_{k+1} \in \Lie_{\geq 1}(\log \mH)$ and $\sum_{i = 1}^{k+1} C_i \in \Lie_{\geq 2}(\log \mH)$.
If a set partition $S$ satisfies $\min(S) = 1$, then
\begin{equation}
    \sum_{S' \succcurlyeq S} M(S') \overset{\Lie_{\geq k+1}(\log \mH)}{\sim} 0.
\end{equation}
\end{lem}
\begin{proof}
    First let us illustrate the intuition with an example.
    Let $k = 6$, $S = \{ \{1, 3, 4\}, \{2\}, \{5, 6\} \}$, then there are five coarsenings of $S$, which are:
    \[
    S, \{ \{1, 3, 4\}, \{2, 5, 6\} \}, \{ \{1, 3, 4, 2 \}, \{5, 6\} \}, \{ \{1, 3, 4, 5, 6\}, \{2\} \}, \{ \{1, 3, 4, 2, 5, 6\} \}.
    \]
Correspondingly,
\begin{align*}
& M(S) + M(\{ \{1, 3, 4\}, \{2, 5, 6\} \}) + M(\{ \{1, 3, 4, 2 \}, \{5, 6\} \}) + M(\{ \{1, 3, 4, 5, 6\}, \{2\} \}) \\
& \quad + M(\{ \{1, 3, 4, 2, 5, 6\} \}) \\
= & \frac{1}{4!}\sum_{\sigma \in \Sym_{7}}[[[[[[C_{\sigma(1)}, C_{\sigma(2)}], C_{\sigma(1)}], C_{\sigma(1)}], C_{\sigma(3)}], C_{\sigma(3)}] \\
& \quad + \frac{1}{5!}\sum_{\sigma \in \Sym_{7}}[[[[[[C_{\sigma(1)}, C_{\sigma(2)}], C_{\sigma(1)}], C_{\sigma(1)}], C_{\sigma(2)}], C_{\sigma(2)}] \\
& \quad + \frac{1}{5!}\sum_{\sigma \in \Sym_{7}}[[[[[[C_{\sigma(1)}, C_{\sigma(1)}], C_{\sigma(1)}], C_{\sigma(1)}], C_{\sigma(2)}], C_{\sigma(2)}] \\
& \quad + \frac{1}{5!}\sum_{\sigma \in \Sym_{7}}[[[[[[C_{\sigma(1)}, C_{\sigma(2)}], C_{\sigma(1)}], C_{\sigma(1)}], C_{\sigma(1)}], C_{\sigma(1)}] \\
& \quad + \frac{1}{6!}\sum_{\sigma \in \Sym_{7}}[[[[[[C_{\sigma(1)}, C_{\sigma(1)}], C_{\sigma(1)}], C_{\sigma(1)}], C_{\sigma(1)}], C_{\sigma(1)}] \\
= & \sum_{i, j, k \text{ distinct}}[[[[[[C_{i}, C_{j}], C_{i}], C_{i}], C_{k}], C_{k}] + \sum_{i \neq j = k}[[[[[[C_{i}, C_{j}], C_{i}], C_{i}], C_{k}], C_{k}] \\
& \quad + \sum_{i = j \neq k}[[[[[[C_{i}, C_{j}], C_{i}], C_{i}], C_{k}], C_{k}] + \sum_{i = k \neq j}[[[[[[C_{i}, C_{j}], C_{i}], C_{i}], C_{k}], C_{k}] \\
& \quad + \sum_{i = j = k}[[[[[[C_{i}, C_{j}], C_{i}], C_{i}], C_{k}], C_{k}] \\
= & \sum_{i = 1}^7 \sum_{j = 1}^7 \sum_{k = 1}^7 [[[[[[C_{i}, C_{j}], C_{i}], C_{i}], C_{k}], C_{k}] \\
= & \sum_{i = 1}^7 \sum_{k = 1}^7 [[[[[[C_{i}, \sum_{j = 1}^7 C_{j}], C_{i}], C_{i}], C_{k}], C_{k}] \\
\in & \sum_{i = 1}^7 \sum_{k = 1}^7 [[[[[[C_{i}, \Lie_{\geq 2}(\log \mH)], C_{i}], C_{i}], C_{k}], C_{k}] \\
\subseteq & \Lie_{\geq 7}(\log \mH).
\end{align*}
So $\sum_{S' \succcurlyeq S} M(S') \overset{\Lie_{\geq k+1}(\log \mH)}{\sim} 0$ for this particular example.

For the general case, write $S = \{A_1, \ldots, A_s\}$ with $\card(A_1) = 1$, then
\begin{align*}
    & \sum_{S' \succcurlyeq S} M(S') \\
    = & \sum_{S' \succcurlyeq S} \sum_{\overset{\bj \in \{1, \ldots, k+1\}^k}{\SP(\bj) = S'}} [\ldots[C_{j_1}, C_{j_2}], \ldots, C_{j_{k}}] \\
    = & \sum_{\overset{(j_1, \ldots, j_k) \in \{1, \ldots, k+1\}^k}{j_{i} = j_{i'} \text{ if } i, i' \text{ are in the same set of } S}} [\ldots[C_{j_1}, C_{j_2}], \ldots, C_{j_{k}}] \\
    = & \sum_{i_1 = 1}^{k+1} \cdots \sum_{i_s = 1}^{k+1} [\ldots[C_{i_{f(1)}}, C_{i_{f(2)}}], \ldots, C_{i_{f(k)}}] \quad \text{ where } f(r) \text{ is defined by } r \in A_{f(r)}. \\
    = & \sum_{i_2 = 1}^{k+1} \cdots \sum_{i_s = 1}^{k+1} [\ldots[\ldots[C_{i_{f(1)}}, C_{i_{f(2)}}], \ldots, \sum_{i_1 = 1}^{k+1} C_{i_{1}}], \ldots, C_{i_{f(k)}}] \\
    \in & \sum_{i_2 = 1}^{k+1} \cdots \sum_{i_s = 1}^{k+1} [\ldots[\ldots[C_{i_{f(1)}}, C_{i_{f(2)}}], \ldots, \Lie_{\geq 2}(\log \mH)], \ldots, C_{i_{f(k)}}] \\
    \subseteq & \; \Lie_{\geq k+1}(\log \mH).
\end{align*}
Hence $\sum_{S' \succcurlyeq S} M(S') \overset{\Lie_{\geq k+1}(\log \mH)}{\sim} 0$.
\end{proof}

Using Equation~\eqref{eq:PhitoMset}, Lemma~\ref{lem:effectmodk1} gives the following corollaries.
\begin{cor}\label{cor:effectmodk1}
    Let $\mH$ be a subset of $\UT(n, \Q)$.
    Suppose $C_1, \ldots, C_{k+1} \in \Lie_{\geq 1}(\log \mH)$ and $\sum_{i = 1}^{k+1} C_i \in \Lie_{\geq 2}(\log \mH)$.
    If a set partition $S$ satisfies $\min(S) = 1$, then
    \begin{equation}
        \sum_{S' \succcurlyeq S} \Phi(S') \overset{\Lie_{\geq k+1}(\log \mH)}{\sim} 0.
    \end{equation}
\end{cor}
\begin{proof}
    For any $\tau \in \Sym_k$, we have that $S'_{\tau} \succcurlyeq S_{\tau}$ if and only if $S' \succcurlyeq S$.
    Therefore by Equation~\eqref{eq:PhitoMset},
    \begin{multline*}
        \sum_{S' \succcurlyeq S} \Phi(S') = \sum_{S' \succcurlyeq S} \sum_{\tau \in \Sym_k} \frac{(-1)^{d(\tau)}}{k^2 \binom{k-1}{d(\tau)}} \cdot M(S'_{\tau}) = \sum_{S'_{\tau} \succcurlyeq S_{\tau}} \sum_{\tau \in \Sym_k} \frac{(-1)^{d(\tau)}}{k^2 \binom{k-1}{d(\tau)}} \cdot M(S'_{\tau}) \\
        = \sum_{\tau \in \Sym_k} \frac{(-1)^{d(\tau)}}{k^2 \binom{k-1}{d(\tau)}} \cdot \sum_{S'_{\tau} \succcurlyeq S_{\tau}} M(S'_{\tau}) \overset{\Lie_{\geq k+1}(\log \mH)}{\sim} 0.
    \end{multline*}
\end{proof}

\begin{cor}\label{cor:effectmodk2}
Let $\mH$ be a subset of $\UT(n, \Q)$.
Suppose $C_1, \ldots, C_{k+1} \in \Lie_{\geq 1}(\log \mH)$ and $\sum_{i = 1}^{k+1} C_i \in \Lie_{\geq 2}(\log \mH)$.
For any set partition $S$, the symmetric sum $\Phi(S)$ is equivalent under $\overset{\Lie_{\geq k+1}(\log \mH)}{\sim}$ to a linear combination of $\Phi(S')$ where $\min(\IP(S')) \geq 2$ (that is, every set in the partitions $S'$ has cardinality at least two).

In other words, there exist integers $\alpha_{S'}$, where $S'$ ranges over all set partitions satisfying $\min(\IP(S')) \geq 2$, such that
\[
\Phi(S) \overset{\Lie_{\geq k+1}(\log \mH)}{\sim} \sum_{S', \min(\IP(S')) \geq 2} \alpha_{S'} \Phi(S').
\]
\end{cor}
\begin{proof}
    Corollary~\ref{cor:effectmodk1} shows that if $\min(\IP(S)) = 1$, then under the equivalence $\overset{\Lie_{\geq k+1}(\log \mH)}{\sim}$, we can replace $\Phi(S)$ by $- \sum_{S' \succcurlyeq S, S' \neq S} \Phi(S')$.
    Repeat this ``coarsening'' procedure for all $\Phi(S')$, $\min(\IP(S')) = 1,$ for sufficiently many times, we can rewrite $\Phi(S)$ as a linear combination of expressions $\Phi(S')$ where $\min(\IP(S')) \geq 2$.
\end{proof}

Define a \emph{partition-integer pair} to be a pair $(P, c)$, where $P$ is an integer partition and $c$ is a number in $\set(P)$.
For a partition-integer pair $(P, c)$, define the following symmetric sum.
\[
\tM(P, c) \coloneqq M(S),
\]
where $S$ is a set partition such that $\IP(S) = P$, and $1 \in A \in S$ with $\card(A) = \max(P)$ and $2 \in A' \in S$ with $\card(A') = c$.
For example, a possible definition of $\tM((3,2,1), 1)$ can be
\begin{multline*}
\tM((3,2,1), 1) \coloneqq M(\{\{1, 3, 4\}, \{2\}, \{5, 6\}\}) \\
= \frac{1}{4!} \sum_{\sigma \in \Sym_{7}}[[[[[[C_{\sigma(2)}, C_{\sigma(7)}], C_{\sigma(2)}], C_{\sigma(2)}], C_{\sigma(4)}], C_{\sigma(4)}] \\
= \sum_{1 \leq i, j, k \leq 7, i,j,k \text{ distinct}}[[[[[[C_{i}, C_{j}], C_{i}], C_{i}], C_{k}], C_{k}].
\end{multline*}
Note that this definition \emph{a priori} depends on the choice of the set partition $S$. 
However, under the equivalence relation $\overset{\Lie_{\geq 2}(\Lie_{\geq 2}(\log \mH))}{\sim}$, different choices of $S$ result in the same equivalence class.
Indeed, let $\bj$ be a tuple whose associated set partition is $S$.
By Lemma~\ref{lem:effectmod22}, any exchange of order among the elements $j_3, j_4, \ldots, j_{k}$ will not change the equivalence class of $[\ldots[C_{\sigma(j_1)}, C_{\sigma(j_2)}], \ldots, C_{\sigma(j_{k})}]$, so it will not change the equivalence class of $M(\bj)$.
This means that the equivalent class of $M(S)$ does not change when we permute the numbers $3, 4, \ldots, k$.
For example, $M(\{ \{1, 3, 4\}, \{2\}, \{5, 6\} \}) \sim M(\{ \{1, 3, 5\}, \{2\}, \{4, 6\} \})$, because
\begin{multline*}
    M(\{ \{1, 3, 4\}, \{2\}, \{5, 6\} \}) = \frac{1}{4!} \sum_{\sigma \in \Sym_{7}}[[[[[[C_{\sigma(2)}, C_{\sigma(7)}], C_{\sigma(2)}], C_{\sigma(2)}], C_{\sigma(4)}], C_{\sigma(4)}] \\
    \overset{\Lie_{\geq 2}(\Lie_{\geq 2}(\log \mH))}{\sim} \frac{1}{4!} \sum_{\sigma \in \Sym_{7}}[[[[[[C_{\sigma(2)}, C_{\sigma(7)}], C_{\sigma(2)}], C_{\sigma(4)}], C_{\sigma(2)}], C_{\sigma(4)}] = M(\{ \{1, 3, 5\}, \{2\}, \{4, 6\}\}).
\end{multline*}
Hence, the equivalent class of $M(S)$ only depends on the integer partition $\IP(S)$ as well as the cardinality of the sets where $1$ and $2$ belong.
This is uniquely determined by the partition-cardinality pair $(P, c)$.

\begin{lem}\label{lem:effectmodk}
Let $\mH$ be a subset of $G$.
Suppose $C_1, \ldots, C_{k+1} \in \Lie_{\geq 1}(\log \mH)$ and $\sum_{i = 1}^{k+1} C_i \in \Lie_{\geq 2}(\log \mH)$.
For any set partition $S$ satisfying $\min(\IP(S)) \geq 2$, the symmetric sum $M(S)$
is equivalent (under $\sim$) to a linear combination of $\tM(P, c)$, where $(P, c)$ are partition-integer pairs satisfying $c \neq \max(P)$ and $\min(P) \geq 2$.

In other words, there exist integers $\beta_{(P, c)}$, where $(P, c)$ ranges over all partition-integer pairs with $c \neq \max(P)$ and $\min(P) \geq 2$, such that
\[
M(S) \sim \sum_{(P, c)} \beta_{(P, c)} \tM(P, c).
\]
\end{lem}
\begin{proof}
Write $S = \{A_1, \ldots, A_s\}$ with $\card(A_1) = \max(\IP(S))$.
By Lemma~\ref{lem:effectmod22}, the equivalence class of $M(S)$ does not change when we permute the numbers $3, 4, \ldots, k$.
We can therefore suppose $3 \in A_1$.
Take any tuple $\bj = (j_1, \ldots, j_k) \in \{1, \ldots, k+1\}^k$ with $\SP(\bj) = S$.
By the Jacobi identity,
\begin{multline}\label{eq:appJac}
[\ldots[[C_{\sigma(j_1)}, C_{\sigma(j_2)}], C_{\sigma(j_3)}], \ldots, C_{\sigma(j_k)}] = \\
[\ldots[[C_{\sigma(j_3)}, C_{\sigma(j_2)}], C_{\sigma(j_1)}], \ldots, C_{\sigma(j_k)}] - [\ldots[[C_{\sigma(j_3)}, C_{\sigma(j_1)}], C_{\sigma(j_2)}], \ldots, C_{\sigma(j_k)}].
\end{multline}
Summing up for $\sigma \in \Sym_{k+1}$, the expression $\sum_{\sigma \in \Sym_{k+1}}[\ldots[[C_{\sigma(j_3)}, C_{\sigma(j_2)}], C_{\sigma(j_1)}], \ldots, C_{\sigma(j_k)}]$ is equivalent to $\left(k+1 - \card(S)\right)! \cdot \tM(\IP(S), c)$, with $c = \card(A_i)$ where $j_2 \in A_i$.
Similarly, the expression 
\[
\sum_{\sigma \in \Sym_{k+1}}[\ldots[[C_{\sigma(j_3)}, C_{\sigma(j_1)}], C_{\sigma(j_2)}], \ldots, C_{\sigma(j_k)}]
\]
is equivalent to $\left(k+1 - \card(S)\right)! \cdot \tM(\IP(S), c')$, with $c' = \card(A_{i'})$ where $j_1 \in A_{i'}$.

We claim that if $c = \max(\IP(S))$, then $\tM(\IP(S), c) \sim 0$.
This is because, writing 
\[
\tM(\IP(S), c) = \frac{1}{\left(k+1 - \card(S)\right)!} \sum_{\sigma \in \Sym_{k+1}}[\ldots[[C_{\sigma(j_3)}, C_{\sigma(j_2)}], C_{\sigma(j_1)}], \ldots, C_{\sigma(j_k)}],
\]
if $j_2 \in A_i$ with $\card(A_i) = \max(\IP(S))$, then swapping $2$ and $3$ in the set partition $\SP(\bj)$ does not change its associated integer partition.
Therefore, we have
\begin{multline*}
\tM(\IP(S), \max(S')) = \frac{1}{\left(k+1 - \card(S)\right)!} \sum_{\sigma \in \Sym_{k+1}} [\ldots[[C_{\sigma(j_3)}, C_{\sigma(j_2)}], C_{\sigma(j_1)}], \ldots, C_{\sigma(j_k)}] \sim \\
- \frac{1}{\left(k+1 - \card(S)\right)!} \sum_{\sigma \in \Sym_{k+1}} [\ldots[[C_{\sigma(j_2)}, C_{\sigma(j_3)}], C_{\sigma(j_1)}], \ldots, C_{\sigma(j_k)}] \sim - \tM(\IP(S'), \max(S')),
\end{multline*}
so $\tM(\IP(S), \max(S)) \sim 0$.
This proves that if $c = \max(\IP(S))$, then $\tM(\IP(S), c) \sim 0$.

Summing up Equation~\eqref{eq:appJac} for $\sigma \in \Sym_{k+1}$, we conclude that 
\begin{multline*}
M(S) = \frac{1}{\left(k+1 - \card(S)\right)!} \sum_{\sigma \in \Sym_{k+1}} [\ldots[[C_{\sigma(j_1)}, C_{\sigma(j_2)}], C_{\sigma(j_3)}], \ldots, C_{\sigma(j_k)}] \\
= \tM(\IP(S), c) - \tM(\IP(S), c')
\end{multline*}
is equivalent (under $\sim$) to a linear combination of expressions
$\tM(\IP(S), c)$, where $c \neq \max(S)$.
\end{proof}

For any $k$, all partition-integer pairs satisfying $c \neq \max(P)$ and $\min(P) \geq 2$ can be effectively listed.
For example, when $k = 5$, there is only one pair $((3,2), 2)$.
When $k = 7$, there are three pairs 
\[
((5,2), 2), ((4,3), 3), ((3, 2, 2), 2).
\]
When $k = 9$, there are six pairs 
\[
((7,2), 2), ((6,3), 3), ((5,4), 4), ((5, 2, 2), 2), ((4,3,2), 3), ((4,3,2), 2).
\]

Combining Corollary~\ref{cor:effectmodk2}, Equation~\eqref{eq:PhitoMset} and Lemma~\ref{lem:effectmodk}, we obtain the following proposition.
\begin{prop}\label{prop:writelin}
    Suppose $C_1, \ldots, C_{k+1} \in \Lie_{\geq 1}(\log \mH)$ and $\sum_{i = 1}^{k+1} C_i \in \Lie_{\geq 2}(\log \mH)$.
    Let $m \geq 2$ and $\bj = (j_1, \ldots, j_m) \in \{1, \ldots, k+1\}^m$.
    The expression $\sum_{\sigma \in \Sym_{k+1}} H_k(C_{\sigma(j_1)}, \ldots, C_{\sigma(j_{m})})$ is equivalent (under $\sim$) to a linear combination of $\tM(P, c)$, where $(P, c)$ ranges over all partition-integer pairs with $c \neq \max(P)$ and $\min(P) \geq 2$. 
    Furthermore, this linear combination can be effectively computed.

    In other words, one can effectively compute rational numbers $\gamma_{(P, c)}$, such that
    \[
    \sum_{\sigma \in \Sym_{k+1}} H_k(C_{\sigma(j_1)}, \ldots, C_{\sigma(j_{m})}) \sim \sum_{(P, c)} \gamma_{(P, c)} \tM(P, c).
    \]
\end{prop}
\begin{proof}
By the Dynkin formula (Lemma~\ref{lem:dynkin}), the expression $\sum_{\sigma \in \Sym_{k+1}} H_k(C_{\sigma(j_1)}, \ldots, C_{\sigma(j_{m})})$ can be rewritten into a linear combination of $\Phi(\SP(\bj'))$, where $\bj'$ are subsequences (with possible repetition) of $\bj$.
Then, Corollary~\ref{cor:effectmodk2} shows that each $\Phi(\SP(\bj'))$ is equivalent (under $\sim$) to a linear combination of $\Phi(S')$ with $\min(\IP(S')) \geq 2$.
Next, Equation~\eqref{eq:PhitoMset} shows that each $\Phi(S'), \min(\IP(S')) \geq 2$ is equal to a linear combination of $M(S'')$ with $\min(\IP(S'')) \geq 2$.
The condition $\min(\IP(S'')) \geq 2$ is due to the fact that for any $\tau \in \Sym_k$ we have $\IP(S_{\tau}) = \IP(S)$.
Finally, by Lemma~\ref{lem:effectmodk}, each $M(S''), \min(\IP(S'')) \geq 2$ is equivalent (under $\sim$) to a linear combination of $\tM(P, c)$ with $c \neq \max(P)$ and $\min(P) \geq 2$.

In summary, any expression $\sum_{\sigma \in \Sym_{k+1}} H_k(C_{\sigma(j_1)}, \ldots, C_{\sigma(j_{r})})$ is equivalent to a linear combination of $\tM(P, c)$, where $(P, c)$ ranges over all partition-integer pairs with $c \neq \max(P)$ and $\min(P) \geq 2$. 
Furthermore, the proof of Corollary~\ref{cor:effectmodk2}, Equation~\eqref{eq:PhitoMset} and Lemma~\ref{lem:effectmodk} give an effective procedure that computes the coefficients of this linear combination.
\end{proof}

The effective procedure of Proposition~\ref{prop:writelin} is summarized by Algorithm~\ref{alg:lincombrewrite}.
Note that for the algorithm we fix the integer $k$, so all set partitions in the algorithm refer to set partitions of $k$.

\begin{algorithm}[ht!]
\caption{Find $\gamma_{(P, c)}$ where
$\sum_{\sigma \in \Sym_{k+1}} H_k(C_{\sigma(j_1)}, \ldots, C_{\sigma(j_{m})}) \sim \sum_{(P, c)} \gamma_{(P, c)} \tM(P, c)$} \vspace{0.1cm}
\label{alg:lincombrewrite}
\begin{description}[nosep]
\item[Input:] 
an integer $k$ and a tuple $\bj = (j_1, \ldots, j_m) \in \{1, \ldots, k+1\}^m$.
\item[Output:]
rational numbers $\gamma_{(P, c)}$, where $(P, c)$ ranges over all partition-integer pairs with $c \neq \max(P)$ and $\min(P) \geq 2$.
\end{description}
\begin{enumerate}[nosep, label = \arabic*.]
    \item \textbf{Compute rational numbers $a_{S}$ such that}
    \begin{equation}
        \sum_{\sigma \in \Sym_{k+1}} H_k(C_{\sigma(j_1)}, \ldots, C_{\sigma(j_{m})}) = \sum_{\text{set partition } S} a_{S} \Phi(S) 
    \end{equation}
    \textbf{in the following way:}
        \begin{enumerate}[nosep]
            \item Initialize with $a_S \coloneqq 0$ for all set partitions $S$.
            \item For every tuple $(i_1, \ldots, i_m) \in \Zp^m$ such that $i_1 + \cdots + i_m = k$, compute the sequence 
            \[
            \iota \coloneqq (\underbrace{j_1, \ldots, j_1}_{i_1}, \underbrace{j_2, \ldots, j_2}_{i_2}, \ldots, \underbrace{j_m, \ldots, j_m}_{i_m})
            \]
            and update $a_{\SP(\iota)} \coloneqq a_{\SP(\iota)} + \frac{(k+1 - \card(\SP(\iota)))!}{i_1 ! \cdots i_m !}$.
        \end{enumerate}
        
    \item \textbf{Compute rational numbers $b_{S}$ such that}
    \begin{equation}
        \sum_{\text{set partition } S} a_{S} \Phi(S) = \sum_{\overset{\text{set partition } S,}{\min(\IP(S)) \geq 2}} b_{S} \Phi(S)
    \end{equation}
    \textbf{in the following way:}
    \begin{enumerate}[nosep]
        \item Initialize with $b_S \coloneqq a_{S}$ for all set partitions $S$.
        \item Order all set partitions $S$ into $S_1, S_2, \ldots, S_{p}$, such that if $S_j \succcurlyeq S_i$ then $j \geq i$.
        \item For $i = 1, 2, \ldots, p$ : \\
        If $\min(\IP(S_i)) = 1$, then update $b_{S_i} \coloneqq 0$ and $b_{S_j} \coloneqq b_{S_j} - b_{S_i}$ for all $S_j \succcurlyeq S_i$.
    \end{enumerate}
    
    \item \textbf{Compute rational numbers $g_{S}$ such that}
    \begin{equation}
        \sum_{\overset{\text{set partition } S,}{\min(\IP(S)) \geq 2}} b_{S} \Phi(S) = \sum_{\overset{\text{set partition } S,}{\min(\IP(S)) \geq 2}} g_{S} M(S)
    \end{equation}
    \textbf{in the following way:}
    \begin{enumerate}[nosep]
            \item Initialize with $g_S \coloneqq 0$ for all set partitions $S, \min(\IP(S)) \geq 2$.
            \item For every set partition $S$ and every permutation $\sigma \in \Sym_{k}$, compute the set partition
            \[
            S_{\sigma} \coloneqq \Big\{\{\sigma(j) \mid j \in A\} \;\Big|\; A \in S \Big\}
            \]
            and update $g_{S_{\sigma}} \coloneqq g_{S_{\sigma}} + b_S \cdot \frac{(-1)^{d(\sigma)}}{k^2 \binom{k-1}{d(\sigma)}}$ (where $d(\cdot)$ denotes the number of descents).
        \end{enumerate}
    \item \textbf{Compute all partition-integer pairs $(P, c)$ with $c \neq \max(P)$ and $\min(P) \geq 2$.}
    \setcounter{algsplit}{\value{enumi}}
\end{enumerate}
(To be continued in the next page)
\end{algorithm}

\setcounter{algocf}{1}
\begin{algorithm}[h!]
\caption{(continued)}\vspace{0.1cm}
\begin{enumerate}[nosep, label = \arabic*.]
\setcounter{enumi}{\value{algsplit}}

    \item \textbf{Compute rational numbers $\gamma_{(P, c)}$ such that}
    \begin{equation}
        \sum_{\overset{\text{set partition } S}{\min(\IP(S)) \geq 2}} g_{S} M(S) = \sum_{\overset{(P, c)}{c \neq \max(P), \min(P) \geq 2}} \gamma_{(P, c)} \tM(P, c)
    \end{equation}
    \textbf{in the following way:}
    \begin{enumerate}[nosep]
        \item Initialize with $\gamma_{(P, c)} \coloneqq 0$ for all $(P, c)$, $c \neq \max(P)$ and $\min(P) \geq 2$.
        \item For all set partitions $S$ with $\min(\IP(S)) \geq 2$:
        \begin{enumerate}[nosep]
            \item If $1 \in A$, $\card(A) = \max(\IP(S))$ and $2 \in B$, $\card(B) \neq \max(\IP(S))$, then update $\gamma_{(\IP(S), \card(B))} \coloneqq \gamma_{(\IP(S), \card(B))} + g_S$.
            \item If $1 \in A$, $\card(A) \neq \max(\IP(S))$ and $2 \in B$, $\card(B) = \max(\IP(S))$, then update $\gamma_{(\IP(S), \card(A))} \coloneqq \gamma_{(\IP(S), \card(A))} - g_S$.
            \item If $1 \in A$, $\card(A) \neq \max(\IP(S))$ and $2 \in B$, $\card(B) \neq \max(\IP(S))$, then update $\gamma_{(\IP(S), \card(A))} \coloneqq \gamma_{(\IP(S), \card(A))} - g_S$, $\gamma_{(\IP(S), \card(B))} \coloneqq \gamma_{(\IP(S), \card(B))} + g_S$.
        \end{enumerate}
    \end{enumerate}
    \item \textbf{Return the numbers $\gamma_{(P, c)}$.}
\end{enumerate}
\end{algorithm}

We can now give computer assisted proofs of Lemmas~\ref{lem:H5}~-~\ref{lem:H9} based on Algorithm~\ref{alg:lincombrewrite}.

\begin{proof}[Proof of Lemma~\ref{lem:H5}]
(The SageMath~\cite{sagemath} code can be found at \url{https://doi.org/10.6084/m9.figshare.20124146.v1}.)
Set $k = 5$.
Using Algorithm~\ref{alg:lincombrewrite} on the tuples $(1, 2, 3, 4, 5, 6)$ and 
\[
\bj = (1, 2, 3, 4, 4, 5, 5, 6, 6, 1, 2, 3),
\]
we get 
\begin{align*}
    \sum_{\sigma \in \Sym_6} H_5(\log B_{\sigma(1)}, \ldots, \log B_{\sigma(6)})
    & \sim \tM((3,2), 2), \\
    \sum_{\sigma \in \Sym_6} H_5\big(\log B_{\sigma(j_1)}, \ldots, \log B_{\sigma(j_{12})}\big)
    & \sim - \tM((3,2), 2).
\end{align*}
Therefore,
\begin{equation*}
    \sum_{\sigma \in \Sym_{6}} H_5(\log B_{\sigma(1)}, \ldots, \log B_{\sigma(6)})
    + \sum_{\sigma \in \Sym_{6}} H_5(\log B_{\sigma(j_1)}, \ldots, \log B_{\sigma(j_{12})}) \sim 0.
\end{equation*}
\end{proof}

\begin{proof}[Proof of Lemma~\ref{lem:H7}]
(The SageMath~\cite{sagemath} code can be found at \url{https://doi.org/10.6084/m9.figshare.20124113.v1}.)
Set $k = 7$.
Using Algorithm~\ref{alg:lincombrewrite} on the tuples $(1, 2, \ldots, 8)$ and 
\begin{align*}
   \bj_1 = (j_{1,1}, j_{1,2}, \ldots, j_{1,16}) & = (1, 2, 3, 4, 5, 5, 6, 6, 7, 7, 8, 8, 1, 2, 3, 4), \\
   \bj_2 = (j_{2,1}, j_{2,2}, \ldots, j_{2,16}) & = (1, 2, 3, 4, 5, 4, 6, 7, 1, 2, 8, 3, 5, 6, 7, 8).
\end{align*}
We get 
\begin{equation*}
     \sum_{\sigma \in \Sym_8} H_7(\log B_{\sigma(1)}, \ldots, \log B_{\sigma(8)})
    \sim \frac{34}{15} \tM((5,2), 2)
    - \frac{34}{45} \tM((4,3), 3)
     + \frac{68}{15} \tM((3,2,2), 2), 
\end{equation*}
\begin{equation*}
    \sum_{\sigma \in \Sym_8} H_7\big(\log B_{\sigma(j_{1,1})}, \ldots, \log B_{\sigma(j_{1,16})}\big) 
    \sim \frac{34}{15} \tM((5,2), 2)
    + \frac{238}{45} \tM((4,3), 3)
     - \frac{68}{5} \tM((3,2,2), 2),
\end{equation*}
\begin{equation*}
    \sum_{\sigma \in \Sym_8} H_7\big(\log B_{\sigma(j_{2,1})}, \ldots, \log B_{\sigma(j_{2,16})}\big) 
    \sim - \frac{68}{15} \tM((5,2), 2)
    + \frac{34}{45} \tM((4,3), 3)
    - \frac{34}{5} \tM((3,2,2), 2).
\end{equation*}
Therefore,
\begin{equation*}
    \sum_{\sigma \in \Sym_{8}} H_7(\log B_{\sigma(1)}, \ldots, \log B_{\sigma(8)})
    + \sum_{s = 1}^2 \alpha_{s} \sum_{\sigma \in \Sym_{8}} H_7(\log B_{\sigma(j_{s,1})}, \ldots, \log B_{\sigma(j_{s,16})}) \sim 0
\end{equation*}
with $\alpha_1 = \frac{1}{15}, \alpha_2 = \frac{8}{15}$.
\end{proof}

\begin{proof}[Proof of Lemma~\ref{lem:H9}]
(The SageMath~\cite{sagemath} code can be found at \url{https://doi.org/10.6084/m9.figshare.20122979.v1}).
Set $k = 9$.
Using Algorithm~\ref{alg:lincombrewrite} on the tuples $(1, 2, \ldots, 10)$ and 
\begin{align*}
   (j_{1,1}, j_{1,2}, \ldots, j_{1,20}) & = (5, 4, 7, 10, 2, 8, 3, 8, 1, 9, 7, 6, 5, 6, 2, 3, 9, 10, 1, 4), \\
   (j_{2,1}, j_{2,2}, \ldots, j_{2,20}) & = (8, 3, 5, 7, 10, 6, 8, 2, 1, 10, 2, 4, 9, 1, 5, 9, 3, 6, 7, 4), \\
   (j_{3,1}, j_{3,2}, \ldots, j_{3,20}) & = (7, 10, 2, 6, 4, 9, 6, 4, 1, 5, 3, 5, 1, 9, 3, 7, 10, 2, 8, 8), \\
   (j_{4,1}, j_{4,2}, \ldots, j_{4,20}) & = (10, 2, 2, 6, 7, 1, 9, 3, 9, 4, 8, 7, 8, 5, 5, 1, 4, 10, 6, 3), \\
   (j_{5,1}, j_{5,2}, \ldots, j_{5,20}) & = (3, 5, 10, 1, 4, 8, 6, 9, 3, 2, 7, 6, 1, 10, 9, 7, 2, 4, 5, 8), \\
   (j_{6,1}, j_{6,2}, \ldots, j_{6,20}) & = (4, 7, 2, 10, 2, 1, 3, 5, 8, 1, 6, 9, 10, 7, 6, 8, 3, 5, 9, 4).
\end{align*}
We get 
\begin{multline*}
     \sum_{\sigma \in \Sym_{10}} H_9(\log B_{\sigma(1)}, \ldots, \log B_{\sigma(10)}) 
    \sim \frac{347}{105} \tM((7,2), 2) + \frac{347}{315} \tM((6,3), 3) \\
    + \frac{347}{105} \tM((5,4), 4) + \frac{1388}{105} \tM((5,2,2), 2) - \frac{347}{21} \tM((4,3,2), 3) + \frac{347}{21} \tM((4,3,2), 2),
\end{multline*}
\begin{multline*}
    \sum_{\sigma \in \Sym_{10}} H_9\big(\log B_{\sigma(j_{1,1})}, \ldots, \log B_{\sigma(j_{1,20})}\big) \sim -\frac{347}{105} \tM((7,2), 2) + \frac{21167}{945} \tM((6,3), 3) \\
    - \frac{4511}{315} \tM((5,4), 4) + 0 \cdot \tM((5,2,2), 2) + \frac{3817}{63} \tM((4,3,2), 3) + \frac{1735}{63} \tM((4,3,2), 2),
\end{multline*}
\begin{multline*}
    \sum_{\sigma \in \Sym_{10}} H_9\big(\log B_{\sigma(j_{2,1})}, \ldots, \log B_{\sigma(j_{2,20})}\big) \sim \frac{347}{45} \tM((7,2), 2) + \frac{18391}{945} \tM((6,3), 3) \\
    + \frac{347}{14} \tM((5,4), 4) - \frac{1388}{315} \tM((5,2,2), 2) + \frac{9022}{63} \tM((4,3,2), 3) - \frac{694}{63} \tM((4,3,2), 2),
\end{multline*}
\begin{multline*}
    \sum_{\sigma \in \Sym_{10}} H_9\big(\log B_{\sigma(j_{3,1})}, \ldots, \log B_{\sigma(j_{3,20})}\big) \sim \frac{16309}{42} \tM((7,2), 2) + \frac{85709}{630} \tM((6,3), 3) \\
    + \frac{241859}{1260} \tM((5,4), 4) + \frac{30883}{126} \tM((5,2,2), 2) - \frac{8675}{63} \tM((4,3,2), 3) + \frac{94037}{630} \tM((4,3,2), 2),
\end{multline*}
\begin{multline*}
    \sum_{\sigma \in \Sym_{10}} H_9\big(\log B_{\sigma(j_{4,1})}, \ldots, \log B_{\sigma(j_{4,20})}\big) \sim \frac{20473}{210} \tM((7,2), 2) - \frac{314729}{1890} \tM((6,3), 3) \\
    + \frac{4511}{140} \tM((5,4), 4) + \frac{137759}{630} \tM((5,2,2), 2) - \frac{23249}{315} \tM((4,3,2), 3) + \frac{33659}{210} \tM((4,3,2), 2),
\end{multline*}
\begin{multline*}
    \sum_{\sigma \in \Sym_{10}} H_9\big(\log B_{\sigma(j_{5,1})}, \ldots, \log B_{\sigma(j_{5,20})}\big) \sim \frac{347}{210} \tM((7,2), 2) + \frac{35741}{1890} \tM((6,3), 3) \\
    - \frac{18391}{1260} \tM((5,4), 4) + \frac{1041}{70} \tM((5,2,2), 2) - \frac{347}{63} \tM((4,3,2), 3) + \frac{1735}{126} \tM((4,3,2), 2),
\end{multline*}
\begin{multline*}
    \sum_{\sigma \in \Sym_{10}} H_9\big(\log B_{\sigma(j_{6,1})}, \ldots, \log B_{\sigma(j_{6,20})}\big) \sim - \frac{1388}{105} \tM((7,2), 2) - \frac{56561}{945} \tM((6,3), 3) \\
    + \frac{4511}{126} \tM((5,4), 4) - \frac{3123}{70} \tM((5,2,2), 2) - \frac{28454}{315} \tM((4,3,2), 3) - \frac{51703}{630} \tM((4,3,2), 2).
\end{multline*}
Therefore,
\begin{equation*}
    \sum_{\sigma \in \Sym_{10}} H_9(\log B_{\sigma(1)}, \ldots, \log B_{\sigma(10)})
    + \sum_{s = 1}^6 \alpha_{s} \sum_{\sigma \in \Sym_{10}} H_9(\log B_{\sigma(j_{s,1})}, \ldots, \log B_{\sigma(j_{s,20})}) \sim 0
\end{equation*}
with $\alpha_1 = \frac{44566633}{13702661}, \alpha_2 = \frac{557040}{13702661}, \alpha_3 = \frac{205175}{3915046}, \alpha_4 = \frac{1307207}{13702661}, \alpha_5 = \frac{86275275}{27405322}, \alpha_6 = \frac{4105194}{1957523}$.
\end{proof}
\end{appendices}

\end{document}